\documentclass[10pt]{article}

\title{Analysis of Langevin Monte Carlo via convex optimization}

\usepackage[a4paper]{geometry}

\usepackage{amsthm,amsmath}
\usepackage[utf8]{inputenc}
\usepackage{amssymb}
\usepackage{makeidx}
\usepackage[english]{babel}
\usepackage{graphicx}
\usepackage{amsfonts,amssymb}
\usepackage{oldgerm}
\usepackage{mathrsfs}
\usepackage[active]{srcltx}
\usepackage{verbatim}
\usepackage{enumitem, kantlipsum} 
\usepackage{aliascnt}
\usepackage{bbm}
\usepackage{array}
\usepackage{hyperref}
\usepackage[ruled,vlined]{algorithm2e}
 \usepackage[disable]{todonotes}
\usepackage{xargs}
\usepackage{cellspace}
\usepackage{xr}

\usepackage[Symbolsmallscale]{upgreek}
\usepackage{xcolor}
\definecolor{mydarkblue}{rgb}{0,0.08,0.45}
\usepackage{multicol} 
\usepackage{subcaption} 
\usepackage{wrapfig}
\usepackage[final]{fixme}
\fxsetup{layout=inline}

\usepackage{accents}
\usepackage{dsfont}
\usepackage{aliascnt}
\usepackage{cleveref}

\makeatletter
\newtheorem{theorem}{Theorem}
\crefname{theorem}{theorem}{Theorems}
\Crefname{Theorem}{Theorem}{Theorems}

\newaliascnt{lemma}{theorem}
\newtheorem{lemma}[lemma]{Lemma}
\aliascntresetthe{lemma}
\crefname{lemma}{lemma}{lemmas}
\Crefname{Lemma}{Lemma}{Lemmas}

\newaliascnt{corollary}{theorem}
\newtheorem{corollary}[corollary]{Corollary}
\aliascntresetthe{corollary}
\crefname{corollary}{corollary}{corollaries}
\Crefname{Corollary}{Corollary}{Corollaries}

\newaliascnt{proposition}{theorem}
\newtheorem{proposition}[proposition]{Proposition}
\aliascntresetthe{proposition}
\crefname{proposition}{proposition}{propositions}
\Crefname{Proposition}{Proposition}{Propositions}

\newaliascnt{definition}{theorem}

\aliascntresetthe{definition}
\crefname{definition}{definition}{definitions}
\Crefname{Definition}{Definition}{Definitions}

\newaliascnt{remark}{theorem}

\aliascntresetthe{remark}
\crefname{remark}{remark}{remarks}
\Crefname{Remark}{Remark}{Remarks}

\crefname{example}{example}{examples}
\Crefname{Example}{Example}{Examples}

\crefname{figure}{figure}{figures}
\Crefname{Figure}{Figure}{Figures}

\newtheorem{assumption}{\textbf{A}\hspace{-3pt}}
\Crefname{assumption}{\textbf{A}\hspace{-3pt}}{\textbf{A}\hspace{-3pt}}
\crefname{assumption}{\textbf{A}}{\textbf{A}}

\Crefname{assumptionG}{\textbf{G}\hspace{-3pt}}{\textbf{G}\hspace{-3pt}}
\crefname{assumptionG}{\textbf{G}}{\textbf{G}}

\Crefname{assumptionQ}{\textbf{Q}\hspace{-3pt}}{\textbf{Q}\hspace{-3pt}}
\crefname{assumptionQ}{\textbf{Q}}{\textbf{Q}}

\newcounter{hypA}
\def\xstar{x^\star}

\newcommandx{\functionspace}[2][1=+]{\mathbb{F}_{#1}(#2)}

\newcommand{\argmin}{\operatorname*{arg\,min}}

\newcommandx{\VarDeux}[3][3=]{\operatorname{Var}^{#3}_{#1}\left\{#2 \right\}}

\newcommand{\1}{\mathbbm{1}}

\newcommand{\Borel}{\mathcal{B}}

\newcommand{\LeftEqNo}{\let\veqno\@@leqno}

\newcommand{\ceil}[1]{\left\lceil #1 \right\rceil}

\newcommand{\N}{\ensuremath{\mathbb{N}}}

\newcommand{\PE}{\mathbb{E}}

\newcommand{\abs}[1]{\left\vert #1 \right\vert}
\newcommand{\absLigne}[1]{\vert #1 \vert}
\newcommand{\tvnorm}[1]{\| #1 \|_{\mathrm{TV}}}

\newcommandx{\Vnorm}[2][1=V]{\| #2 \|_{#1}}
\newcommandx{\VnormEq}[2][1=V]{\left\| #2 \right\|_{#1}}
\newcommandx{\norm}[2][1=]{\ifthenelse{\equal{#1}{}}{\left\Vert #2 \right\Vert}{\left\Vert #2 \right\Vert^{#1}}}
\newcommandx{\normLigne}[2][1=]{\ifthenelse{\equal{#1}{}}{\Vert #2 \Vert}{\Vert #2\Vert^{#1}}}

\newcommand{\parenthese}[1]{\left(#1 \right)}

\newcommand{\parentheseDeux}[1]{\left[ #1 \right]}
\newcommand{\defEns}[1]{\left\lbrace #1 \right\rbrace }
\newcommand{\defEnsLigne}[1]{\lbrace #1 \rbrace }

\newcommand{\ensemble}[2]{\left\{#1\,:\eqsp #2\right\}}

\newcommand{\ps}[2]{\left\langle#1,#2 \right\rangle}
\newcommand{\psLigne}[2]{\langle#1,#2 \rangle}

\def\rset{\mathbb{R}}
\def\nset{\mathbb{N}}
\def\nsets{\mathbb{N}^{*}}

\newcommand{\probaLigne}[1]{\mathbb{P}( #1 )}
\newcommandx\probaMarkovTilde[2][2=]
{\ifthenelse{\equal{#2}{}}{{\widetilde{\mathbb{P}}_{#1}}}{\widetilde{\mathbb{P}}_{#1}\left[ #2\right]}}

\newcommand{\expe}[1]{\PE \left[ #1 \right]}

\newcommand{\expeLigne}[1]{\PE [ #1 ]}

\newcommand{\bigO}{\ensuremath{\mathcal O}}

\newcommand{\couplage}[2]{\Pi(#1,#2)}
\newcommand{\Pens}{\mathcal{P}}

\newcommand{\plusinfty}{+\infty}

\newcounter{hypoconbis}
\newcounter{saveconbis}
\newcommand\debutH{\begin{list}
{\textbf{H\arabic{hypoconbis}}}{\usecounter{hypoconbis}}\setcounter{hypoconbis}{\value{saveconbis}}}
\newcommand\finH{\end{list}\setcounter{saveconbis}{\value{hypoconbis}}}

\def\ie{\textit{i.e.}}
\def\eqsp{\;}
\newcommand{\coint}[1]{\left[#1\right)}
\newcommand{\ocint}[1]{\left(#1\right]}
\newcommand{\ooint}[1]{\left(#1\right)}
\newcommand{\ccint}[1]{\left[#1\right]}

\newcommand{\ocintLigne}[1]{(#1]}
\newcommand{\oointLigne}[1]{(#1)}

\newcommand{\boule}[2]{\operatorname{B}(#1,#2)}

\def\rmd{\mathrm{d}}
\def\rmZ{\mathrm{Z}}
\def\rme{\mathrm{e}}

\newcommandx\sequence[3][2=,3=]
{\ifthenelse{\equal{#3}{}}{\ensuremath{( #1_{#2})}}{\ensuremath{( #1_{#2})_{ #2 \in #3 }}}}
\newcommandx{\sequencen}[2][2=n\in\N]{\ensuremath{( #1)_{ #2 }}}
\newcommandx{\sequencek}[2][2=k\in\N]{\ensuremath{( #1)_{ #2 }}}
\newcommandx{\sequencens}[2][2=n\in\N^*]{\ensuremath{( #1)_{ #2 }}}
\newcommandx{\sequenceks}[2][2=k\in\N^*]{\ensuremath{( #1)_{ #2 }}}
\newcommandx{\sequencet}[2][2=t \geq 0]{\ensuremath{( #1)_{ #2 }}}
\newcommandx\sequenceDouble[4][3=,4=]
{\ifthenelse{\equal{#3}{}}{\ensuremath{( (#1_{#3},#2_{#3}) )}}{\ensuremath{(  (#1_{#3},#2_{#3}), \eqsp #3 \in #4 )}}}
\newcommandx{\sequencenDouble}[3][3=n\in\N]{\ensuremath{( (#1_{n},#2_{n}), \eqsp #3 )}}

\def\iid{i.i.d.}
\def\lsc{l.s.c}

\def\eg{e.g.}
\newcommand{\wrt}{w.r.t.}

\newcommand{\opnorm}[1]{{\left\vert\kern-0.25ex\left\vert\kern-0.25ex\left\vert #1 
    \right\vert\kern-0.25ex\right\vert\kern-0.25ex\right\vert}}

\def\generator{\mathcal{A}}

\def\bfe{\mathbf{e}}

\def\Id{\operatorname{Id}}

\newcommandx{\CPE}[3][1=]{{\mathbb E}_{#1}\left[#2 \left \vert #3 \right. \right]} 
\newcommandx{\CPVar}[3][1=]{\mathrm{Var}^{#3}_{#1}\left\{ #2 \right\}}
\newcommand{\CPP}[3][]
{\ifthenelse{\equal{#1}{}}{{\mathbb P}\left(\left. #2 \, \right| #3 \right)}{{\mathbb P}_{#1}\left(\left. #2 \, \right | #3 \right)}}

\def\YL{\mathbf{Y}}
\def\XEM{X}

\def\steps{\gamma}
\def\step{\steps}
\def\tstep{\tilde{\steps}}
\newcommandx{\stepa}[1][1=]{\steps_{#1}}
\def\Step{\Gamma}

\newcommandx{\Stepa}[1][1=]{\Step_{#1}}
\def\weights{\lambda}
\def\weight{\weights}
\newcommandx{\weighta}[1][1=]{\weights_{#1}}
\def\Weight{\Lambda}

\newcommandx{\Weighta}[1][1=]{\Weight_{#1}}
\newcommandx{\WeightaInv}[1][1=]{\Weight_{#1}^{-1}}

\def\KL{\operatorname{KL}}
\newcommand{\KLarg}[2]{\KL\left(#1 \middle \vert #2  \right)}

\def\wasserstein{W}

\newcommand{\wassersteinTarg}[2]{\wasserstein_2^2\left(#1 ,  #2  \right)}

\newcommandx{\osc}[2][1=]{\mathrm{osc}_{#1}(#2)}

\def\Id{\operatorname{Id}}

\def\Lone{\mathrm{L}^1}

\newcommand\density[2]{\frac{\rmd #1}{\rmd #2}}

\def\boreleanA{\mathsf{A}}
\def\openU{\mathsf{U}}

\def\setProba{\mathcal{P}}

\def\eventA{\boreleanA}
\def\msa{\eventA}

\def\transpose{\operatorname{T}}

\def\Fscr{\mathscr{F}}
\def\Hscr{\mathscr{H}}
\def\Escr{\mathscr{E}}
\def\Sscr{\mathscr{S}}
\def\tFscr{\tilde{\mathscr{F}}}

\def\tEscr{\tilde{\mathscr{E}}}

\def\divergence{\operatorname{div}}

\def\Pensa{\Pens^{\text{a}}}
\def\tmu{\tilde{\mu}}

\def\trho{\tilde{\rho}}
\def\brho{\bar{\rho}}
\def\YL{\mathbf{Y}}
\def\XEM{X}
\def\steps{\gamma}

\def\Leb{\operatorname{Leb}}
\newcommand{\AC}[1]{\operatorname{AC}(#1)}
\newcommand{\ACloc}[1]{\operatorname{AC}_\text{loc}(#1)}
\def\proj{\operatorname{proj}}

\def\Rker{R}
\def\Sker{S}
\def\Tker{T}

\def\bRker{{\bar R}}
\def\bSker{{\bar S}}

\def\tRker{{\tilde R}}
\def\tSker{{\tilde S}}

\newcommandx{\Rkers}[1][1=]{R_{\steps_{#1}}}
\newcommandx{\Qkers}[1][1=]{Q_{\steps}^{#1}}
\newcommandx{\Rkera}[1][1=]{R_{\steps_{#1}}}
\newcommandx{\Qkera}[1][1=]{Q_{\steps}^{#1}}
\newcommandx{\Skers}[1][1=]{S_{\steps_{#1}}}
\newcommandx{\Tkers}[1][1=]{T_{\steps_{#1}}}
\newcommandx{\Skera}[1][1=]{S_{\steps_{#1}}}
\newcommandx{\Tkera}[1][1=]{T_{\steps_{#1}}}

\newcommandx{\bRkers}[1][1=]{\bar{R}_{\steps_{#1},\steps_{#1+1}}}
\newcommandx{\bQkers}[1][1=]{\bar{Q}_{\steps}^{#1}}
\newcommandx{\bRkera}[1][1=]{\bar{R}_{\steps_{#1},\steps_{#1+1}}}
\newcommandx{\bQkera}[1][1=]{\bar{Q}_{\steps}^{#1}}
\newcommandx{\bSkers}[1][1=]{\bar{S}_{\steps_{#1}}}
\newcommandx{\bTkers}[1][1=]{\bar{T}_{\steps_{#1}}}
\newcommandx{\bSkera}[1][1=]{\bar{S}_{\steps_{#1}}}
\newcommandx{\bTkera}[1][1=]{\bar{T}_{\steps_{#1}}}

\newcommandx{\tRkers}[1][1=]{\tilde{R}_{\steps_{#1},\steps_{#1+1}}}
\newcommandx{\tQkers}[1][1=]{\tilde{Q}_{\steps}^{#1}}
\newcommandx{\tRkera}[1][1=]{\tilde{R}_{\steps_{#1},\steps_{#1+1}}}
\newcommandx{\tQkera}[1][1=]{\tilde{Q}_{\steps}^{#1}}
\newcommandx{\tSkers}[1][1=]{\tilde{S}_{\steps_{#1}}}
\newcommandx{\tTkers}[1][1=]{\tilde{T}_{\steps_{#1}}}
\newcommandx{\tSkera}[1][1=]{\tilde{S}_{\steps_{#1}}}
\newcommandx{\tTkera}[1][1=]{\tilde{T}_{\steps_{#1}}}

\def\ClyapU{M}

\def\msz{\mathsf{Z}}
\def\mcz{\mathcal{Z}}
\def\gradrv{Z}
\def\gradst{\Theta}

\def\vargrad{\upsilon}

\def\bnu{\bar{\nu}}
\def\bX{\bar{X}}

\def\tS{\tilde{S}}
\def\tgamma{\tilde{\gamma}}
\def\tnu{\tilde{\nu}}
\def\tZ{\tilde{Z}}
\def\tstep{\tilde{\step}} 
\def\tX{\tilde{X}}

\def\Lt{\tilde{L}}
\def\tL{\Lt}
\def\tz{\tilde{z}}
\def\tmcz{\tilde{\mcz}}
\def\tmsz{\tilde{\msz}}
\def\tgradst{\tilde{\gradst}}
\def\teta{\tilde{\eta}}

\def\prox{\operatorname{prox}}

\newcommand\fraca[2]{#1/#2}

\def\signop{\operatorname{sign}}
\def\bfe{\mathbf{e}}
\def\tm{\tilde{m}}
\def\tL{\tilde{L}}
\def\vareps{\varepsilon}
\def\xstarf{x_f}

\usepackage{authblk}

\author[1]{Alain Durmus}
\author[2]{Szymon Majewski}
\author[3]{Błażej Miasojedow}

\affil[1]{CMLA - \'Ecole normale supérieure Paris-Saclay, CNRS, Université Paris-Saclay, 94235 Cachan, France.}
\affil[2]{Institute of Mathematics, Polish Academy of Science}
\affil[3]{Institute of Applied Mathematics and Mechanics, University of Warsaw and\\ Institute of Mathematics, Polish Academy of Sciences }

\begin{document}
\footnotetext[1]{
Email: alain.durmus@cmla.ens-cachan.fr}
\footnotetext[2]{Email: smajewski@impan.pl}
\footnotetext[3]{Email: B.Miasojedow@mimuw.edu.pl}
\maketitle

\begin{abstract}
  In this paper, we provide new insights on the Unadjusted Langevin
  Algorithm.  We show that this method can be formulated as a first
  order optimization algorithm of an objective functional defined on
  the Wasserstein space of order $2$. Using this interpretation and
  techniques borrowed from convex optimization, we give a
  non-asymptotic analysis of this method to sample from logconcave
  smooth target distribution on $\rset^d$. Based on this
  interpretation, we propose two new methods for sampling from a
  non-smooth target distribution, which we analyze as well. Besides,
  these new algorithms are natural extensions of the Stochastic Gradient Langevin
  Dynamics (SGLD) algorithm, which is a popular extension of the Unadjusted
  Langevin Algorithm. Similar to SGLD, they only rely on
  approximations of the gradient of the target log density and can be
  used for large-scale Bayesian inference.
\end{abstract}

\section{Introduction}

This paper deals with the problem of sampling from a probability
measure $\pi$ on $(\rset^d,\Borel(\rset^d))$ which admits a density,
still denoted by $\pi$, with respect to the Lebesgue measure given for
all $x \in \rset^d$ by 
\begin{equation*}
  \pi(x) = \left. \rme^{-U(x)} \middle / \int_{\rset^d} \rme^{-U(y)} \rmd y \right. \eqsp,
\end{equation*}
where $U : \rset^d \to \rset$.  This problem arises in various fields
such that Bayesian statistical inference
\cite{gelman:carlin:stern:rudin:2014}, machine learning
\cite{andrieu:defreitas:doucet:jordan:2003}, ill-posed inverse
problems \cite{stuart:2010} or computational physics
\cite{krauth:2006}. Common and current methods to tackle this issue
are Markov Chain Monte Carlo methods
\cite{brooks:gelman:jone:meng:2011}, for example the
Hastings-Metropolis algorithm
\cite{metropolis:rosenbluth:rosenbluth:teller:teller:1953,hastings:1970}
or Gibbs sampling \cite{geman:geman:1984}.  All these methods boil
down to building a Markov kernel on $(\rset^d,\Borel(\rset^d))$ whose
invariant probability distribution is $\pi$.  Yet, choosing an
appropriate proposal distribution for the Hastings-Metropolis
algorithm is a tricky subject. For this reason, it has been proposed
to consider continuous dynamics which naturally leave the target
distribution $\pi$ invariant.  Perhaps, one of the most famous such
examples are the over-damped Langevin diffusion \cite{parisi:1981}
associated with $U$, assumed to be continuously differentiable:
\begin{equation}
\label{eq:langevin}
  \rmd \YL_t =  - \nabla U ( \YL_t) \rmd t + \sqrt{2} \rmd B_t \eqsp,
\end{equation}
where $(B_t)_{t \geq 0}$ is a $d$-dimensional Brownian motion. 
On appropriate conditions on $U$, this SDE admits a unique strong
solution $(\YL_t)_{t \geq 0}$ and defines a strong Markov semigroup
$(P_t)_{t \geq 0}$ which converges to $\pi$ in total variation
\cite[Theorem 2.1]{roberts:tweedie:1996} or Wasserstein distance
\cite{bolley:gentil:guillin:2012}.  However, simulating path solutions
of such stochastic differential equations is not possible in most
cases, and discretizations of these equations are used instead. In
addition, numerical solutions associated with these schemes define
Markov kernels for which $\pi$ is not invariant anymore. Therefore
quantifying the error introduced by these approximations is crucial to
justify their use to sample from the target $\pi$.  We consider in
this paper the Euler-Maruyama discretization of \eqref{eq:langevin}
which defines the (possibly inhomogenous) Markov chain $(\XEM_k)_{k
  \geq 0}$ given for all $k \geq 0$ by 
\begin{equation}
\label{eq:definition_em}
  \XEM_{k+1} = \XEM_k - \steps_{k+1} \nabla U(\XEM_k) + \sqrt{2 \steps_{k+1} } G_{k+1} \eqsp,
\end{equation}
where $(\steps_k)_{k \geq 1}$ is a sequence of step sizes which can be
held constant or converges to $0$, and $(G_{k})_{k \geq 1}$ is a
sequence of \iid~standard $d$-dimensional Gaussian random
variables. The use of the Euler-Maruyama discretization
\eqref{eq:definition_em} to approximatively sample from $\pi$ is
referred to as the Unadjusted Langevin Algorithm (ULA) (or the
Langevin Monte Carlo algorithm (LMC)), and has already been the matter
of many works. For example, weak error estimates have been obtained in
\cite{talay:tubaro:1991}, \cite{mattingly:stuart:higham:2002} for the
constant step size setting and in \cite{lamberton:pages:2003},
\cite{lemaire:2005} when $(\steps_k)_{k \geq 1}$ is non-increasing and
goes to $0$. Explicit and non-asymptotic bounds on the total
variation (\cite{dalalyan:2014}, \cite{durmus:moulines:2017}) or the
Wasserstein distance (\cite{durmus2016high}) between the distribution
of $X_k$ and $\pi$ have been obtained. Roughly, all these results are
based on the comparison between the discretization and the diffusion
process and quantify how the error introduced by the discretization
accumulate throughout the algorithm. In this paper, we propose an
other point of view on ULA, which shares nevertheless some relations
with the Langevin diffusion \eqref{eq:langevin}. Indeed, it has been
shown in \cite{jordan1998variational} that the family of distributions
$(\mu_0 P_t)_{t \geq 0}$, where $(P_t)_{t \geq 0}$ is the semi-group
associated with \eqref{eq:langevin} and $\mu_0$ is a probability
measure on $\Borel(\rset^d)$ admitting a second moment, is the
solution of a gradient flow equation in the Wasserstein space of order
$2$ associated with a particular functional $\Fscr$, see
\Cref{sec:ula-as-gradient}. Therefore, if $\pi$ is invariant for
$( P_t)_{t \geq 0}$, then it is a stationary solution of this
equation, and is the unique minimizer of $\Fscr$ if $U$ is convex.
Starting from this observation, we interpret ULA as a
first order optimization  algorithm on the Wasserstein space of order
$2$ with objective functional $\Fscr$. Namely, we adapt some proofs of
convergence for the gradient descent algorithm from the convex
optimization literature to obtain non-asymptotic and explicit bounds
between the Kullback-Leibler divergence from $\pi$ to averaged
distributions associated with ULA for the constant and non-increasing
step-size setting. Then, these bounds easily imply computable bounds
in total variation norm and Wasserstein distance. If the potential $U$
is strongly convex and gradient Lipschitz, we get back the results of
\cite{durmus:moulines:2017}, \cite{durmus2016high} and
\cite{cheng:bartlett:2017},  when the step-size is held
constant in \eqref{eq:definition_em} (see
\Cref{tab:comparison_gamma_strong_convex}). In the case where $U$ is only
convex and from a warm start, we get a bound on the complexity for ULA of order
$d \bigO(\varepsilon^{-2})$ and $d \bigO(\varepsilon^{-4})$ to get one
sample close from $\pi$ with an accuracy $\varepsilon >0$, in Kullback
Leibler (KL) divergence and total variation distance respectively
(\Cref{tab:comparison_gamma__convex_warm_start}.
The bounds we get starting
from a minimizer of $U$ are presented in
\Cref{tab:comparison_gamma__convex}. 

\begin{table}[h]
\centering
\begin{normalsize}
\begin{tabular}{|c|c|c|c|}
   \hline
 & Total variation  & Wasserstein distance & KL divergence   \\
\hline
\cite{durmus2016high} & $d \bigO(\varepsilon^{-2})$ &  $d \bigO(\varepsilon^{-2})$ & $-$ \\
  \hline
\cite{cheng:bartlett:2017} & $d \bigO(\varepsilon^{-2})$ & $d \bigO(\varepsilon^{-2})$ & $d \bigO(\varepsilon^{-1})$\\
  \hline
This paper & $d \bigO(\varepsilon^{-2})$ & $d \bigO(\varepsilon^{-2})$ & $d \bigO(\varepsilon^{-1})$\\
\hline
\end{tabular}
\end{normalsize}
\caption{\normalsize{Complexity for ULA when $U$ is strongly convex and gradient Lipschitz (up to logarithmic terms)}}
\label{tab:comparison_gamma_strong_convex}
\end{table}

\begin{table}[h]
\centering
\begin{normalsize}
\begin{tabular}{|c|c|c|c|}
   \hline
  & Total variation  & Wasserstein distance & KL divergence   \\  
  \hline
\cite{cheng:bartlett:2017} & $d \bigO(\varepsilon^{-6})$ & - & $d \bigO(\varepsilon^{-3})$\\
  \hline
This paper & $d \bigO(\varepsilon^{-4})$ & - & $d \bigO(\varepsilon^{-2})$\\
\hline
\end{tabular}
\end{normalsize}
\caption{\normalsize{Complexity of ULA from a warm start when $U$ is convex and gradient Lipschitz (up to logarithmic terms)}}
\label{tab:comparison_gamma__convex_warm_start}
\end{table}

\begin{table}[h]
\centering
\begin{normalsize}
\begin{tabular}{|c|c|c|c|}
   \hline
  & Total variation  & Wasserstein distance & KL divergence   \\
  \hline
\cite{durmus:moulines:2017} & $d^5 \bigO(\varepsilon^{-2})$ & - &  - \\
  \hline
 This paper & $d^3 \bigO(\varepsilon^{-4})$ & - & $d^3 \bigO(\varepsilon^{-2})$\\
\hline
\end{tabular}
\end{normalsize}
\caption{\normalsize{Complexity of ULA when $U$ is convex and gradient Lipschitz (up to logarithmic terms)}}
\label{tab:comparison_gamma__convex}
\end{table}

In addition, we propose two new algorithms to sample from a class of
non-smooth log-concave distributions for which we derive computable
non-asymptotic bounds as well. The first one can be applied to
Lipschitz convex potential for which unbiased estimates of
subgradients are available. Remarkably, the bounds we obtain for this
algorithm depend on the dimension only through the initial condition
and the variance of the stochastic sub-gradient estimates.  The second
method we propose is a generalization of the Stochastic Gradient
Langevin Dynamics algorithm \cite{welling:teh:2011}. This latter is a
popular extension of ULA, in which the gradient is replaced by a
sequence of \iid~unbiased estimators. For this new scheme, we assume
that $U$ can be decomposed as the sum of two functions $U_1$ and
$U_2$, where $U_1$ is at least continuously differentiable and $U_2$
is only convex, and use stochastic gradient estimates for $U_1$ and
the proximal operator associated with $U_2$. This new method is close
to the one proposed in \cite{durmus:moulines:pereyra:2016} but is
different. To get computable bounds from the target
distribution $\pi$, we interpret this algorithm as a first order
optimization algorithm and provide explicit bounds between the
Kullback-Leibler divergence from $\pi$ to distributions associated
with SGLD. In the case where $U$ is strongly convex and gradient
Lipschitz (\ie~$U_2 = 0$), we get back the same complexity as
\cite{dalalyan:karagulyan:2017} which is of order
$d \bigO(\varepsilon^{-2})$ for the Wasserstein distance. We obtain
the same complexity for the total variation distance and a complexity
of order $d \bigO(\varepsilon^{-1})$ for the KL divergence
(\Cref{tab:comparison_gamma_strong_convex_sgld}).  In the case where
$U$ is only convex and from a warm start, we get a complexity of order
$d \bigO(\varepsilon^{-2})$ and $d \bigO(\varepsilon^{-4})$ to get one
sample close from $\pi$ with an accuracy $\varepsilon >0$ in KL
divergence and total variation distance respectively, see
\Cref{tab:comparison_gamma_convex_sgld}. The bounds we get starting
from a minimizer of $U$ are presented in
\Cref{tab:comparison_gamma_convex_sgld_xstar}.

Furthermore, SGLD has
been also analyzed in a general setting, \ie~the potential $U$ is not
necessarily convex. In \cite{teh:vollmer:zygalakis:2015}, a study of this scheme is done by weak
error estimates.  Finally, \cite{raginsky2017non} and
\cite{xu2017global} gives some results regarding the potential use of
SGLD as an optimization algorithm to minimize the potential $U$ by
targeting a target density proportional to $x \mapsto \rme^{-\beta U(x)}$
for some $\beta >0$.

\begin{table}
\centering
\begin{normalsize}
\begin{tabular}{|c|c|c|c|c|}
   \hline
 & Total variation  & Wasserstein distance & KL divergence   \\
\hline
\cite{dalalyan:karagulyan:2017} & $-$ &  $d \bigO(\varepsilon^{-2})$ & $-$ \\
  \hline
This paper & $d \bigO(\varepsilon^{-2})$ & $d \bigO(\varepsilon^{-2})$ & $d \bigO(\varepsilon^{-1})$\\
\hline
\end{tabular}
\end{normalsize}
\caption{\normalsize{Complexity for SGLD when $U$ is strongly convex and gradient Lipschitz (up to logarithmic terms)}}
\label{tab:comparison_gamma_strong_convex_sgld}
\end{table}

\begin{table}[h!]
\centering
\begin{normalsize}
\begin{tabular}{|c|c|c|c|c|}
   \hline
 & Total variation  & Wasserstein distance & KL divergence   \\
\hline
This paper & $d \bigO(\varepsilon^{-4})$ & $-$ & $d \bigO(\varepsilon^{-2})$\\
\hline
\end{tabular}
\end{normalsize}
\caption{\normalsize{Complexity for SGLD from a warm start when $U$ is convex and gradient Lipschitz}}
\label{tab:comparison_gamma_convex_sgld}
\end{table}

\begin{table}[h!]
\centering
\begin{normalsize}
\begin{tabular}{|c|c|c|c|c|}
   \hline
 & Total variation  & Wasserstein distance & KL divergence   \\
\hline
This paper & $d^3 \bigO(\varepsilon^{-4})$ & $-$ & $d^3 \bigO(\varepsilon^{-2})$\\
\hline
\end{tabular}
\end{normalsize}
\caption{\normalsize{Complexity for SGLD from a warm start when $U$ is convex and gradient Lipschitz}}
\label{tab:comparison_gamma_convex_sgld_xstar}
\end{table}

In summary, our contributions are the following:
\begin{itemize}
\item We give a new interpretation of ULA and use it to get bounds on
  the Kullback-Leibler divergence from $\pi$ to the iterates of
  ULA. We recover the dependence on the dimension of \cite[Theorem
  3]{cheng:bartlett:2017} in the strongly convex case and get tighter
  bounds. Note that this result implies previously known bounds
  between $\pi$ and ULA in Wasserstein distance and the total
  variation distance but with a completely different technique. We
  also give computable bounds when $U$ is only convex which improves the
  results of \cite{durmus:moulines:2017}, \cite{dalalyan:2014} and \cite{cheng:bartlett:2017}.
\item We give two new methodologies to sample from a non-smooth
  potential $U$ and make a non-asymptotic analysis of them. These two new algorithms are generalizations of SGLD.
\end{itemize}

The paper is organized as follows. In \Cref{sec:ula-as-gradient}, we
give some intuition on the strategy we take to analyze ULA and its
variants. These ideas come from gradient flow theory in
Wasserstein space. In \Cref{sec:main-lemmas}, we give the main results
we obtain on ULA and their proof. In \Cref{sec:expl-bounds-extens},
two variants of ULA are presented and analyzed. Finally, numerical
experiments on logistic regression models are presented in \Cref{sec:numer-exper} to support our
theoretical findings regarding our new methodologies.

\subsection*{Notations and conventions}
Denote by $\Borel(\rset^d)$ the Borel $\sigma$-field of $\rset^d$,
$\Leb$ the Lebesgue measure on $\Borel(\rset^d)$,
$\functionspace[]{\rset^d}$ the set of all Borel measurable functions
on $\rset^d$ and for $f \in \functionspace[]{\rset^d}$,
$\Vnorm[\infty]{f}= \sup_{x \in \rset^d} \abs{f(x)}$.  For $\mu$ a
probability measure on $(\rset^d, \mathcal{B}(\rset^d))$ and $f \in
\functionspace[]{\rset^d}$ a $\mu$-integrable function, denote by
$\mu(f)$ the integral of $f$ \wrt~$\mu$.  Let $\mu$ and $\nu$ be two
sigma-finite measures on $(\rset^d, \Borel(\rset^d))$. Denote by $\mu
\ll \nu$ if $\mu$ is absolutely continuous \wrt~$\nu$ and $\rmd \mu /
\rmd \nu$ the associated density. Let $\mu,\nu$ be two probability
measures on $(\rset^d, \Borel(\rset^d))$. Define the Kullback-Leibler
divergence of $\mu$ from $\nu$ by 
\begin{equation*}
  \KLarg{\mu}{\nu} = 
  \begin{cases}
    \int_{\rset^d} \frac{\rmd \mu}{\rmd \nu}(x) \log \parenthese{\frac{\rmd \mu}{\rmd \nu} (x)} \rmd \nu (x) \eqsp, & \text{if } \mu \ll \nu \\
\plusinfty & \text{ otherwise} \eqsp.
  \end{cases}
\end{equation*}

We say that $\zeta$ is a
transference plan of $\mu$ and $\nu$ if it is a probability measure on
$(\rset^d \times \rset^d, \mathcal{B}(\rset^d \times \rset^d) )$ such
that for all measurable set $\boreleanA$ of $\rset^d$,
$\zeta(\boreleanA \times \rset^d) = \mu(\boreleanA)$ and
$\zeta(\rset^d \times \boreleanA) = \nu(\boreleanA)$. We denote by
$\couplage{\mu}{\nu}$ the set of transference plans of $\mu$ and
$\nu$. Furthermore, we say that a couple of $\rset^d$-random variables
$(X,Y)$ is a coupling of $\mu$ and $\nu$ if there exists $\zeta \in
\couplage{\mu}{\nu}$ such that $(X,Y)$ are distributed according to
$\zeta$.  For two probability measures $\mu$ and $\nu$, we define the
Wasserstein distance of order $2$ as
\begin{equation}
\label{eq:definition_wasserstein}
W_2(\mu,\nu) =  \left( \inf_{\zeta \in \couplage{\mu}{\nu}} \int_{\rset^d \times \rset^d} \norm[2]{x-y}\rmd \zeta (x,y)\right)^{1/2} \eqsp.
\end{equation}
By \cite[Theorem 4.1]{VillaniTransport}, for all $\mu,\nu$ probability
measures on $\rset^d$, there exists a transference plan $\zeta^\star
\in \couplage{\mu}{\nu}$ such that for any coupling $(X,Y)$
distributed according to $\zeta^\star$, $W_2(\mu,\nu) =
\PE[\norm[2]{X-Y}]^{1/2}$. This kind of transference plan
(respectively coupling) will be called an optimal transference plan
(respectively optimal coupling) associated with $W_2$.  We denote by
$\setProba_2(\rset^d)$ the set of probability measures with finite
second moment: for all $\mu \in \setProba_2(\rset^d)$, $\int_{\rset^d}
\norm[2]{x} \rmd \mu( x) < \plusinfty$. By \cite[Theorem
6.16]{VillaniTransport}, $\setProba_2(\rset^d)$ equipped with the
Wasserstein distance $W_2$ of order $2$ is a complete separable metric
space. Denote by $\Pensa(\rset^d) = \{ \mu \in \Pens_2(\rset^d) \, : \, \mu \ll \Leb \}$.

For two probability measures $\mu$ and $\nu$ on $\rset^d$,  the total variation distance distance between $\mu$ and $\nu$ is defined by
$\tvnorm{\mu-\nu} = \sup_{\eventA \in \Borel(\rset^d)}\abs{\mu(\eventA) - \nu(\eventA)}$.

Let $n \in \nset \cup \{\infty\}$ and $\openU \subset \rset^d$ be an
open set of $\rset^d$. Denote by $C^n(\openU)$ the set of $n$-th
continuously differentiable function from $\openU$ to $\rset$. Denote
by $C^n_c(\openU)$ the set of $n$-th continuously differentiable
function from $\openU$ to $\rset$ with compact support.  Let $I
\subset \rset$ be an interval and $f : I \to \rset$. $f$ is absolutely
continuous on $I$ if for all $\varepsilon >0$, there exists $\delta >0$
such that for all $n \in \nset^*$ and $t_1,\ldots,t_{2n} \in I$, $t_1
\leq \cdots \leq t_{2n}$,
\begin{equation*}
  \text{if $\sum_{k=1}^n\defEns{t_{2k}-t_{2k-1}} \leq \delta$ } \text{ then } \sum_{k=1}^n \abs{f(t_{2k})-f(t_{2k-1})} \leq \varepsilon \eqsp.
\end{equation*}
 In the sequel, we take
the convention that $\sum_{p}^n =0$ and $\prod_p ^n = 1$ for $n,p \in
\nset$, $n <p$.


\section{Interpretation of ULA as an optimization algorithm}
\label{sec:ula-as-gradient}

Throughout this paper, we assume that $U$ satisfies the following condition for $m \geq 0$.
\begin{assumption}[$m$]
  \label{assum:convexity}
\begin{sf}
$U : \rset^d \to \rset$ is $m$-convex, \ie~for all $x,y \in \rset^d$, 
\begin{equation*}
U(t x + (1-t) y  ) \leq t U(x) + (1-t) U(y) -t(1-t)(m/2)\norm[2]{x-y}
\end{equation*}
\end{sf}
\end{assumption}

Note that \Cref{assum:convexity}$(m)$ includes the case where $U$ is
only convex when $m=0$.  
We consider in this Section the following additional
condition on $U$ which will be relaxed in \Cref{sec:expl-bounds-extens}.
\begin{assumption}
  \label{assum:grad_lip}
  \begin{sf}
    $U$ is continuously differentiable and $L$-gradient Lipschitz, \ie~there exists $L \geq 0$ such that for all $x,y \in \rset^d$, $\norm{\nabla U(x)-\nabla U(y)} \leq L\norm{x-y}$
  \end{sf}
\end{assumption}

Under \Cref{assum:convexity} and \Cref{assum:grad_lip}, the Langevin
diffusion \eqref{eq:langevin} has a unique strong solution
$(\YL_t^x)_{t \geq 0}$ starting at $x \in \rset^d$. The Markovian
semi-group $(P_t)_{t \geq 0}$, given for all $t \geq 0$, $x \in
\rset^d$ and $\eventA \in \Borel(\rset^d)$ by $P_t(x,\eventA) =
\probaLigne{\YL_t^x \in \eventA}$, is reversible with respect to $\pi$
and $\pi$ is its unique invariant probability measure, see
\cite[Theorem 1.2, Theorem 1.6]{ambrosio:savare:zambotti:2009}. Using
this probabilistic framework, \cite[Theorem 1.2]{roberts:tweedie:1996}
shows that $(P_t)_{t \geq 0}$ is irreducible with respect to the
Lebesgue measure, strong Feller and $\lim_{t \to \plusinfty}
\tvnorm{P_t(x,\cdot) - \pi} = 0$ for all $x \in \rset^d$. But to
study the properties of the semi-group $(P_t)_{t \geq 0}$, an other
complementary and significant approach can be used.
This dual point of view is based on the adjoint of the infinitesimal
generator associated with $(P_t)_{t \geq 0}$. The \textit{strong}
generator of \eqref{eq:langevin} $(\generator,
\mathrm{D}(\generator))$ is defined for all $f \in
\mathrm{D}(\generator)$ and $x \in \rset^d$ by
\begin{equation*}
  \generator f(x)  = \lim_{t \to 0} t^{-1}(P_t f(x) -f(x)) \eqsp,
\end{equation*}
where $\mathrm{D}(\generator)$ is the subset of $C_0(\rset^d)$ such
that for all $f \in \mathrm{D}(\generator)$, there exists $g \in
C_0(\rset^d)$ such that $ \lim_{t \to 0} \norm{t^{-1}(P_t f
  -f)-g}_{\infty} =0$. In particular for $f \in C_c^2(\rset^d)$, we get by Itô's formula
\begin{equation*}
  \generator f = \ps{\nabla f}{\nabla U} + \Delta f \eqsp.
\end{equation*}
In addition, by \cite[Proposition 1.5]{ethier:kurtz:1986}, for all $f
\in C^2_c(\rset^d)$, $P_tf(x) \in \mathrm{D}(\generator)$ and  for $x \in
\rset^d$, $t \mapsto P_t f(x)$ is continuously differentiable,
\begin{equation}
\label{eq:backward_kolmogorov}
  \frac{\rmd P_t f(x)}{\rmd t} = \generator P_t f(x) = P_t \generator f(x) \eqsp.
\end{equation}

For all $\mu_0 \in \Pensa_2(\rset^d)$ and $t > 0$, by Girsanov's Theorem
\cite[Theorem~5.1, Corollary~5.16, Chapter~3]{karatzas:shreve:1991},
$\mu_0 P_t(\cdot)$ admits a density with respect to the
Lebesgue measure denoted by $\rho_t^{x}$. This density is solution by
\eqref{eq:backward_kolmogorov} of the Fokker-Planck 
equation (in the weak sense):
\begin{equation*}
  \frac{\partial \rho_t^{x}}{\partial t} = \divergence (\nabla \rho_t^{x}+ \rho_t^{x} \nabla U(x)) \eqsp,
\end{equation*}
meaning that for all $\phi \in C_c^{\infty}(\rset^d)$ and $t >0$,
\begin{equation}
\label{eq:backward_kolmogorov_2}
  \frac{\partial }{\partial t} \int_{\rset^d} \phi(y) \rho^x_t(\rmd y) = \int_{\rset^d} \generator \phi(y) \, \rho^x_t(\rmd y) \eqsp.
\end{equation}
In the landmark paper \cite{jordan1998variational}, the authors shows
that if $U$ is infinitely continuously differentiable, $(\rho^x_{t})_{t > 0}$ is the limit of the minimization scheme which
defines a sequence of probability measures  $(\tilde{\rho}_{k,\steps}^x)_{k \in \nset}$ as follows. For $x \in \rset^d$ and $\gamma >0$ set $\rho_{0,\steps}^x = \rmd \mu_0 / \rmd \Leb$ and 
\begin{equation}
\label{eq:min_step_descent}
\begin{aligned}
\tilde{\rho}_{k,\steps} = \frac{\rmd \tilde{\mu}_{k,\steps}}{\rmd \Leb} & \, , \,  \tilde{\mu}_{k,\steps} = \underset{\mu \in \Pensa_2(\rset^d)}{\text{argmin}}
& & \wasserstein_2(\tmu_{k,h},\mu) + \steps \Fscr(\mu) \eqsp, \eqsp k \in \nset \eqsp,
\end{aligned}
\end{equation} 
where $\Fscr : \Pens_2(\rset^d) \to \ocint{- \infty , \plusinfty}$ is the free energy functional, 
\begin{equation}
  \label{eq:def_free_energy}
  \Fscr =
\Hscr + \Escr \eqsp,
\end{equation}
$\Hscr, \Escr : \Pens_2(\rset^d) \to \ocint{- \infty ,
  \plusinfty}$ are the Boltzmann H-functional and the potential energy functional, given for all $\mu \in \Pens_2(\rset^d)$ by 
\begin{align}
\label{eq:def_Boltz_entropy}
  \Hscr(\mu) & =  
  \begin{cases}
    \int_{\rset^d} \density{\mu}{\Leb} (x)  \log \parenthese{\density{\mu}{\Leb} (x) } \rmd x & \text{ if } \mu \ll \Leb \\
\plusinfty \text{ otherwise} \eqsp,
  \end{cases}
\\
\label{eq:def_potential_energy}
  \Escr(\mu) & = \int_{\rset^d} U(x) \rmd \mu(x)  \eqsp.
\end{align}
More precisely, setting $\brho^x_{0,\steps} = \rmd \mu_0 / \rmd \Leb$ and
$\brho_{t,\steps} = \trho_{k,\steps}$ for $t \in
\coint{k\steps,(k+1)\steps}$, \cite[Theorem
5.1]{jordan1998variational} shows that for all $t >0$,
$\brho_{t,\steps}$ converges to $\rho_{t,\steps}$ weakly in
$\Lone(\rset^d)$ as $\steps$ goes to $0$.  This result has been
extended and cast into the framework of gradient flows in the
Wasserstein space $(\Pens_2(\rset^d), \wasserstein_2)$, see
\cite{ambrosio2008gradient}. We provide a short introduction to this
topic in \Cref{sec:defin-usef-results} and present useful concepts and
results for our proofs.  Note that this scheme can be seen as a
proximal type algorithm (see \cite{martinet:1970} and
\cite{rockafeller:1976}) on the Wasserstein space $(\Pens_2(\rset^d),
\wasserstein_2)$ used to minimize the functional $\Fscr$. The
following lemma shows that $\pi$ is the unique minimizer of
$\Fscr$. As a result, the distribution of the Langevin diffusion is
the steepest descent flow of $\Fscr$ and we get back intuitively that
this process converges to the target distribution $\pi$.
\begin{lemma}\label{lem:kl_minimizer} 
  Assume \Cref{assum:convexity}$(0)$. The following holds:
  \begin{enumerate}[label=\alph*)]
  \item \label{item:1:lem:kl_minimizer}  $\pi \in \Pens_2(\rset^d)$, $\Escr(\pi) < \plusinfty$ and $\Hscr(\pi) < \plusinfty$.
  \item \label{item:2:lem:kl_minimizer}  For all 
  $\mu \in \mathcal{P}_2(\rset^d)$ satisfying $\Escr(\mu) <
  \infty$
\begin{equation}
\label{eq:lem:kl_minimizer} 
\Fscr(\mu) - \Fscr(\pi) = \KLarg{\mu}{\pi} \eqsp.
\end{equation}
  \end{enumerate}
\end{lemma}

\begin{proof}
  The proof is postponed to \Cref{sec:proof-crefl}.
\end{proof}

Based on this interpretation, we could think about minimizing $\Fscr$
on the Wasserstein space to get close to $\pi$ using the minimization
scheme \eqref{eq:min_step_descent}. However, while this scheme is
shown in \cite{jordan1998variational} to be well-defined, finding
explicit recursions $(\trho_{k,\steps})_{k \in \nset}$ is as
difficult as minimizing $\Fscr$ and therefore can not be used in
practice. In addition, to the authors knowledge,
there is no efficient and practical schemes to optimize this
functional.  On the other hand, discretization schemes have been used to
approximate the Langevin diffusion $(\YL_t)_{t \geq 0}$
\eqref{eq:langevin} and its long-time behaviour. One of the most
popular method is the Euler-Maruyama discretization $(\XEM_k)_{k \in
  \nset}$ given in \eqref{eq:definition_em}. While most work study the
theoretical properties of this discretization to ensure to get samples
close to the target distribution $\pi$, by comparing the distributions
of $(\XEM_k)_{k \in \nset}$ and $(\YL_t)_{t \geq 0}$ through couplings
or weak error expansions, we interpret this scheme as a first order optimization algorithm for the  objective functional $\Fscr$.

\section{Main results for the Unadjusted Langevin algorithm} \label{sec:main-lemmas}

Let $f : \rset^d \to \rset$ be a convex continuously differentiable
objective function with $\xstarf \in \argmin_{\rset^d} f \not =
\emptyset$. The \textit{inexact} or \textit{stochastic} gradient descent algorithm used to estimate
$f(\xstarf)$ defines the sequence $(x_k)_{k \in \nset}$ starting from $x_0
\in \rset^d$ by the following recursion for $n \in \nset$:
\begin{equation*}
  x_{n+1} = x_n - \step_{n+1} \nabla f(x_n) + \step_{n+1} \Xi(x_n) \eqsp,
\end{equation*}
where $(\step_k)_{ k\in \nsets}$ is a non-increasing sequence of step
sizes and $\Xi : \rset^d \to \rset^d$ is a deterministic or/and stochastic
perturbation of $\nabla f$. To get explicit bound on the convergence (in expectation)
of the sequence $(f(x_n))_{n \in \nset}$ to $f(\xstarf)$, one
possibility (see \eg~\cite{beck:teboule:2009}) is to show that the
following inequality holds: for all $n \in \nset$,
\begin{equation}
\label{eq:ineq_form_stoch_gradient}
2\step_{n+1}(f(x_{n+1}) - f(\xstarf)) \leq \Vert x_{n} - \xstarf\Vert - \Vert x_{n+1} - \xstarf\Vert_2^2 + C\step^2_{n+1} \eqsp,
\end{equation}
for some constant $C \geq 0$.
In a
similar manner as for inexact gradient algorithms, in this section we will establish that ULA satisfies an inequality of the form
\eqref{eq:ineq_form_stoch_gradient} with the objective function $\Fscr$
defined by \eqref{eq:def_free_energy} on $\Pens_2(\rset^d)$, but instead of the Euclidean norm, the Wasserstein distance of order
$2$ will be used.


Consider the family of Markov kernels
$(\Rker_{\gamma_k})_{k \in \nset^*}$ associated with the Euler-Maruyama discretization $(\XEM_k)_{k \in
  \nset}$, \eqref{eq:definition_em}, for a sequence of step sizes $(\gamma_k)_{k \in \nset^*}$, given
for all $\steps >0, x \in \rset^d$ and $\eventA \in \Borel(\rset^d)$ by
\begin{equation}
  \label{eq:def_Rker_euler}
\Rker_{\steps}(x,\eventA) =  (4 \uppi \steps)^{-d/2} \int_{\eventA} \exp\parenthese{-\norm[2]{y-x-\steps \nabla U(x)}/{(4\steps)}} \rmd y \eqsp. 
\end{equation}
\begin{proposition} \label{thm:basic-one-step}
Assume \Cref{assum:convexity}$(m)$ for $m \geq 0$ and \Cref{assum:grad_lip}.
For all $\steps \in \ocintLigne{0,L^{-1}}$ and $\mu \in \Pens_2(\rset^d)$, we have
\begin{equation}
2\steps \defEns{\Fscr(\mu \Rker_{\steps}) - \Fscr( \pi )} \leq (1-m\steps)\wasserstein_2^2(\mu, \pi) - \wasserstein_2^2(\mu \Rker_{\steps} , \pi) + 2\steps^2 L d \eqsp,
\end{equation}
where $\Fscr$ is defined in \eqref{eq:def_free_energy}.
\end{proposition}

For our analysis, we decompose $\Rker_{\steps}$ for all $\steps >0$ in
the product of two elementary kernels $\Sker_{\steps}$ and $\Tker_{\steps}$ given for all  $x \in \rset^d$ and $\eventA \in \Borel(\rset^d)$ by
\begin{equation}
  \label{eq:def_Sker_Tker}
  \Sker_{\steps}(x,\eventA) =  \updelta_{x-\steps \nabla U(x)}(\eventA) \eqsp, \eqsp \Tker_{\steps}(x,\eventA) = (4 \uppi \steps)^{-d/2} \int_{\eventA} \exp\parenthese{-\norm[2]{y-x}/{(4\steps)}} \rmd y \eqsp.
\end{equation}
We take the convention that $\Sker_0 = \Tker_0 = \Id$ is the identity
kernel given for all $x \in \rset^d$ by $\Id(x,\{x\}) = 1$.  $\Skers$
is the deterministic part of the Euler-Maruyama discretization, which
corresponds to gradient descent step relative to $U$ for the $\Escr$ functional, whereas $\Tkers$
is the random part, that corresponds to going along the gradient flow of $\Hscr$. Note then $\Rkers = \Skers \Tkers$ and consider the
following decomposition
\begin{equation}
\label{eq:decomposition_free_energy}
  \Fscr(\mu \Rker_{\steps}) - \Fscr( \pi ) = \Escr(\mu \Rker_{\steps}) -\Escr(\mu \Sker_{\steps})
+\Escr(\mu \Sker_{\steps} ) - \Escr( \pi ) + \Hscr( \mu \Rker_{\steps} )-\Hscr( \pi )\eqsp.
\end{equation}
The proof of \Cref{thm:basic-one-step} then consists in bounding each
difference in the decomposition above. This is the matter of the following Lemma:
\begin{lemma} \label{lem:conv-potential-change}
Assume \Cref{assum:grad_lip}. 
For all $\mu \in \Pens_2(\rset^d)$ and $\steps >0$, 
\begin{equation*}
  \Escr(\mu \Tkers) - \Escr(\mu) \leq L d \steps \eqsp.
\end{equation*}
\end{lemma}
\begin{proof}
 First note that by \cite[Lemma 1.2.3]{nesterov:2004}, for all $x,\tilde{x}\in \rset^d$, we have 
\begin{equation}
  \label{eq:1:lem:basic-potential-gradient-step}
\abs{  U(\tilde{x}) - U(x)-  \ps{\nabla U(x)}{\tilde{x}-x}} \leq (L/2) \norm[2]{\tilde{x}-x}\eqsp.
\end{equation}
  Therefore, for all $\mu \in \Pens_2(\rset^d)$ and $\steps >0$, we get 
\begin{align*}
    \Escr(\mu \Tkers) - \Escr(\mu) &= (4\uppi \steps)^{-d/2} \int_{\rset^d} \int_{\rset^d}  \defEns{U(x+y) - U(x)} \rme^{\norm{y}^2/(4\steps)} \rmd y \rmd \mu(x) \\
& \leq  (4\uppi \steps)^{-d/2} \int_{\rset^d} \int_{\rset^d}  \defEns{\ps{\nabla U(x)}{y}+ (L/2) \norm[2]{y} }\rme^{\norm{y}^2/(4\steps)} \rmd y \rmd \mu(x) \eqsp,
\end{align*}
which concludes the proof. 
\end{proof}

\begin{lemma}
\label{lem:basic-potential-gradient-step}
Assume \Cref{assum:convexity}$(m)$ for $m \geq 0$ and \Cref{assum:grad_lip}.
For all $\steps \in \ocintLigne{0, L^{-1}}$ and $\mu,\nu \in \Pens_2(\rset^d)$, 
\[
2\steps\defEns{\Escr(\mu \Skers) - \Escr(\nu)} \leq (1-m\steps) \wasserstein_2^2(\mu, \nu) - \wasserstein_2^2(\mu \Skers, \nu) -\steps^2 (1-\steps L)\int_{\rset^d} \norm[2]{\nabla U(x)} \rmd \mu (x) \eqsp,
\]
where $\Escr$ and $\Tkers$ are defined in \eqref{eq:def_potential_energy} and \eqref{eq:def_Sker_Tker} respectively.
\end{lemma}
\begin{proof}
Using \eqref{eq:1:lem:basic-potential-gradient-step} and  \Cref{assum:convexity}$(m)$, for all $x,y \in \rset^d$, we get 
\begin{align*}
  U(x-\steps\nabla U(x)) - U(y) & =   U(x-\steps\nabla U(x))-U(x)+U(x) - U(y)  \\
& \leq -\steps(1-\steps L/2)\norm[2]{\nabla U(x)} +\ps{\nabla U(x)}{x-y} - (m/2) \norm[2]{y-x} \eqsp.
\end{align*}
Multiplying both sides by $2\steps$ we obtain:
\begin{multline}
  \label{eq:2:lem:basic-potential-gradient-step}
2 \steps\defEns{  U(x-\steps\nabla U(x)) - U(y)} \leq (1-m\steps)\norm[2]{x-y} - \norm[2]{x-\steps \nabla U(x) -y}  \\
-\steps^2(1-\steps L)\norm[2]{\nabla U(x)} \eqsp.
\end{multline}
Let now $(X,Y)$ be an optimal coupling between $\mu$ and $\nu$.
Then by definition and \eqref{eq:2:lem:basic-potential-gradient-step}, we get 
\begin{multline*}
2\steps\defEns{\Escr(\mu \Skers) - \Escr(\nu)} \leq (1-m\steps)\wasserstein_2^2(\mu, \nu) - \expe{\norm[2]{X-\steps \nabla U(X) -Y}}\\ -\steps^2(1-\steps L)\expe{\norm[2]{\nabla U(X)}} \eqsp.
\end{multline*}
Using that $\wasserstein^2_2(\mu \Skers,\nu) \leq \expeLigne{\norm[2]{X-\steps \nabla U(X) -Y}}$ concludes the proof.
\end{proof}

\begin{lemma}
 \label{lem:ent-grad-flow-step-inequality} 
Let $\mu,\nu \in \Pens_2(\rset^d)$, $\Hscr(\nu) < \infty$.
Then for all $\steps > 0$, 
\[
2\steps\defEns{\Hscr(\mu \Tkers) - \Hscr(\nu)} \leq W_2^2(\mu, \nu) - W_2^2(\mu \Tker_{\steps}, \nu) \eqsp,
\]
where $\Tkers$ is given in \eqref{eq:def_Sker_Tker}.
\end{lemma}
\begin{proof}
  Denote for all $t \geq 0$ by $\mu_t = \mu \Tker_t$.  Then, $(\mu_t)_{ t \geq 0}$ is the solution (in the
  sense of distribution) of the Fokker-Plank equation:
  \begin{equation*}
    \frac{\partial \mu_t}{\partial t} = \Delta \mu_t \eqsp,
  \end{equation*}
and $\mu_t$ goes to $\mu$ as $t$ goes to $0$ in $(\Pens_2(\rset^d),W_2)$. 
 Let $\nu \in \Pens_2(\rset^d)$ and $\steps >0$.
  Then by \Cref{theo:heat_flow_prop}, for all $\epsilon \in \ooint{0,\steps}$,
  there exists $(\delta_t) \in \Lone(\ooint{\epsilon, \steps})$ such
  that
 \begin{align}
\label{eq:1:lem:ent-grad-flow-step-inequality}  
 &  \wasserstein_2^2(\mu_{\steps},\nu) -   \wasserstein_2^2(\mu_{\epsilon},\nu) = \int_\epsilon ^\steps \delta_s \rmd s \\
\label{eq:2:lem:ent-grad-flow-step-inequality}  
&\delta_s/2 \leq \Hscr(\nu) - \Hscr(\mu_s) \eqsp, \text{ for almost all } s \in \ooint{\epsilon,\steps} \eqsp.
 \end{align}
 In addition by \cite[Particular case 24.3]{VillaniTransport}, $s
 \mapsto \Hscr(\mu_s)$ is non-increasing on $\rset^*_+$ and therefore \eqref{eq:2:lem:ent-grad-flow-step-inequality} becomes 
 \begin{equation*}
   \delta_s/2 \leq \Hscr(\nu) - \Hscr(\mu_{\steps}) \eqsp, \text{ for almost all } s \in \ooint{\epsilon,\steps} \eqsp.
 \end{equation*}
Plugging this bound in \eqref{eq:1:lem:ent-grad-flow-step-inequality} yields that for all $\epsilon \in \rset^*_+$,
\begin{equation*}
  \wasserstein_2^2(\mu_t,\nu) -   \wasserstein_2^2(\mu_{\epsilon},\nu) \leq 2(\steps-\epsilon) \defEns{\Hscr(\nu) - \Hscr(\mu_{\steps})} \eqsp.
\end{equation*}
Taking $\epsilon \to 0$ concludes the proof.
\end{proof}

We now have  all the tools to prove  \Cref{thm:basic-one-step}. 
\begin{proof}[Proof of \Cref{thm:basic-one-step}] 
Let $\mu \in \Pens_2(\rset^d)$ and $\steps \in \rset^*_+$.
By \Cref{lem:conv-potential-change}, we get
\begin{equation*}
 \Escr(\mu \Rker_{\steps}) -\Escr(\mu \Sker_{\steps}) =  \Escr(\mu \Sker_{\steps} \Tkers) -\Escr(\mu \Sker_{\steps})  \leq L d \steps \eqsp.
\end{equation*}
By \Cref{lem:basic-potential-gradient-step} since $\pi \in \Pens_2(\rset^d)$ by \Cref{lem:kl_minimizer}-\ref{item:1:lem:kl_minimizer},
\begin{equation*}
 2\steps \defEns{\Escr(\mu \Sker_{\steps} ) - \Escr( \pi )} \leq (1-m\steps) \wasserstein_2^2(\mu, \nu) - \wasserstein_2^2(\mu \Skers, \nu) \eqsp.
\end{equation*}
By \Cref{lem:ent-grad-flow-step-inequality} and \Cref{lem:kl_minimizer}-\ref{item:1:lem:kl_minimizer},
\begin{align*} 
2 \steps \defEns{ \Hscr( \mu \Rker_{\steps} )-\Hscr( \pi )} &= 
2 \steps \defEns{ \Hscr( (\mu \Skers) \Tker_{\steps} )-\Hscr( \pi )} \\
 & \leq \wasserstein_2^2(\mu \Skers, \pi) - \wasserstein_2^2(\mu \Rkers, \pi) \eqsp.
\end{align*}
Plugging these bounds in \eqref{eq:decomposition_free_energy} concludes the proof.
\end{proof}


Based on inequalities of the form \eqref{eq:ineq_form_stoch_gradient} and using the convexity of $f$,
for all $n \in \nset$, non-asymptotic bounds (in expectation) between
$f(\bar{x}_n)$ and $f(\xstarf)$ can be derived, where $(\bar{x}_k)_{k
  \in \nset}$ is the sequence of averages of $(x_k)_{k \in \nset}$
given for all $n \in \nset$ by $\bar{x}_n = n^{-1} \sum_{k=1}^n
x_k$. Besides, if $f$ is assumed to be strongly convex, a bound on
$\mathbb{E}[\normLigne{x_n-\xstarf}^2]$ can be established.  We will adapt this methodology to get some bounds on the
convergence of sequences of averaged measures defined as follows. Let
$\sequenceks{\step_k}$ and $\sequenceks{\weight_k}$ be two
non-increasing sequences of reals numbers referred to as the sequence
of step sizes and weights respectively. Define for all $n,N \in \nset$, $n \geq 1$,
\begin{equation}
\label{eq:definition_sum_step_weight}
  \Stepa[N,N+n] = \sum_{k=N+1}^{N+n} \stepa[k] \eqsp, \qquad \Weighta[N,N+n] = \sum_{k=N+1}^{N+n} \weighta[k] \eqsp.
\end{equation}
Let
$\mu_0 \in \Pens_2(\rset^d)$ be an initial distribution.
The sequence of probability measures $(\nu^{N}_n)_{n \in \nsets}$ is defined for all $n,N \in \nset$, $n \geq 1$, by
\begin{equation}
  \label{eq:definition_average_measure}
\nu^N_n = \Weight_{N,N+n} ^{-1}\sum_{k=N+1}^{N+n} \weight_{k}\, \mu_0 \Qkers[k] \eqsp, \qquad \Qkers[k] = \Rkers[1] \cdots \Rkers[k] \eqsp, \text{ for $k \in \nsets$} \eqsp,
\end{equation}
where $\Rkers$ is defined by \eqref{eq:def_Rker_euler} and $N$ is a burn-in time. We take in the
following, the convention that $\Qkers[0]$ is the identity operator.


\begin{theorem}
\label{thm:const_step_conv}
Assume \Cref{assum:convexity}($m$) for $m \geq 0$ and \Cref{assum:grad_lip}. Let
$\sequenceks{\step_k}$ and $\sequenceks{\weight_k}$ be two
non-increasing sequences of positive real numbers satisfying  $\step_1 \leq
L^{-1}$, and for all $k\in \nsets$,  $\weighta[k+1](1-m\stepa[k+1])/\stepa[k+1] \leq\weighta[k]/\stepa[k]$. Let
$\mu_0 \in \Pens_2(\rset^d)$ and $N \in \nset$.  Then for all $n \in
\nsets$, it holds:
\begin{multline*}
\KLarg{\nu^N_n}{\pi} + \left.  \weighta[N+n]\wassersteinTarg{\mu_0 \Qkers[N+n]}{ \pi} \middle / 
(2 \stepa[N+n] \Weighta[N,N+n])
\right.
\\ \leq \left. \weighta[N+1](1-m\stepa[N+1]) \wassersteinTarg{\mu_0 \Qkers[N]}{ \pi} \middle/ (2 \stepa[N] \Weighta[N,N+n]) \right.
+ (Ld/\Weighta[N,N+n]) \sum_{k=N+1}^{N+n} \stepa[k] \weighta[k]  \eqsp,
\end{multline*}
where $\nu^{N}_n$ and $\Qkers[N]$ are defined in \eqref{eq:definition_average_measure}.
\end{theorem}
\begin{proof}
Using the convexity of Kullback-Leibler divergence (see \cite[Theorem 2.7.2]{cover:thomas:2006} or \cite[Theorem 11]{canerven:harremos:2014})  and \Cref{thm:basic-one-step}, we obtain
\begin{align*}
\KLarg{\nu^N_n}{\pi} &\leq \Weight_{N,N+n} ^{-1}\sum_{k=N+1}^{N+n} \weight_{k} \KLarg{\mu_0 \Qkers[k]}{\pi} \\
& \leq (2 \Weight_{N,N+n}) ^{-1} \left[ \frac{(1-m\stepa[N+1])\weighta[N+1]}{\stepa[N+1]} \wassersteinTarg{\mu_0 \Qkers[N]}{ \pi} - 
\frac{\weighta[N+n]}{\stepa[N+n]} \wassersteinTarg{\mu_0 \Qkers[N+n]}{ \pi} \right. \\
& \qquad \left. + \sum_{k=N+1}^{N+n-1}  \defEns{\frac{(1-m \stepa[k+1])\weighta[k+1]}{\stepa[k+1]}-\frac{\weighta[k]}{\stepa[k]}}  \wassersteinTarg{\mu_0 \Qkers[k]}{ \pi} + 
 \sum_{k=N+1}^{N+n}  Ld \weighta[k] \stepa[k] \right] \eqsp.
\end{align*}
We get the thesis using that $\weighta[k+1](1-m\stepa[k+1])/\stepa[k+1] \leq\weighta[k]/\stepa[k]$ for all $k\in \nsets$.
\end{proof}
\begin{corollary}
  \label{coro:eps_just_convex_ula}
  Assume \Cref{assum:convexity}($0$) and \Cref{assum:grad_lip}. Let
  $\varepsilon >0$ and $\mu_0 \in \Pens_2(\rset^d)$. Let
  \begin{align*}
    \steps_{\varepsilon}  \leq \min\defEns{\varepsilon/(2Ld), L^{-1}}   \eqsp, \qquad \qquad   n_{\varepsilon}  \geq \lceil W_2^2(\mu_0, \pi) \steps_{\varepsilon}^{-1} \varepsilon^{-1} \rceil \eqsp.
  \end{align*}
 Then it holds  $\KLarg{\nu_{n_{\varepsilon}}}{\pi} \leq \varepsilon$
where $\nu_{n_{\varepsilon}} = n_{\varepsilon}^{-1} \sum_{k=1}^{n_{\varepsilon}} \mu_0 \Rker_{\steps_{\varepsilon}}^k$.
\end{corollary}
\begin{proof}
We apply \Cref{thm:const_step_conv} with $\steps_k = \steps_{\varepsilon}$ and $\weight_k = 1$ for all $k \geq 1$. We obtain
\[
\KLarg{\nu_{n_{\varepsilon}}}{\pi} + \left.  \wassersteinTarg{\mu_0 \Qkers[n_{\varepsilon}]}{ \pi} \middle / 
(2 \steps_{\varepsilon} n_{\varepsilon})
\right. \leq \left. \wassersteinTarg{\mu_0}{ \pi} \middle/ (2 \steps_{\varepsilon} n_{\varepsilon}) \right.
+ (Ld/n_{\varepsilon}) \sum_{k=1}^{n_{\varepsilon}} \steps_{\varepsilon}  \eqsp,
\]
and the proof is concluded by a straightforward calculation using the definition of $\gamma_{\varepsilon}$ and $n_{\varepsilon}$. 
\end{proof}

\begin{corollary}
  \label{coro:ula_non_increas_sz}
  Assume \Cref{assum:convexity}($m$) for $m \geq 0$ and
  \Cref{assum:grad_lip}. Let $ \alpha \in \ooint{0,1}$. Define
  $(\gamma_k)_{k \in \nsets}$ and $(\lambda_k)_{k \in \nsets}$ for all
  $k \in \nsets$ by $\step_k= 
  \fraca{\step_1}{k^{\alpha}}$, $ \lambda_k =
  \fraca{\step_1}{(k+1)^{\alpha}}$, $\gamma_1 \in \ooint{0,L^{-1}}$. Then,
  there exists $C \geq 0$ such that
  for all $n \in \nsets$ we have
  $\KLarg{\nu_n^0}{\pi} \leq C \max(n^{\alpha - 1}, n^{- \alpha})$, if $\alpha \not = 1/2$,  and 
  for $\alpha = 1/2$, we have
  $\KLarg{\nu_n^0}{\pi} \leq C (\ln(n) + 1) n^{-1/2}$, where  $\nu_n^0$ is defined by \eqref{eq:definition_average_measure}. 
\end{corollary}
\begin{proof}
  The proof is postponed to \Cref{sec:proof-crefc_coro:ula_non_increas_sz}.
\end{proof}

{
In the case where  a warm start is available for the Wasserstein distance,
\ie~$W_2^2(\mu_0,\pi) \leq C$, for some absolute constant $C \geq 0$,
then \Cref{coro:eps_just_convex_ula} implies that the complexity of
ULA to obtain a sample close from $\pi$ in KL with a precision target
$\varepsilon>0$ is of order $d\bigO(\varepsilon^{-2})$. In addition,
by Pinsker inequality, we have for all probability measure $\mu$ on
$(\rset^d,\mathcal{B}(\rset^d))$,
$\tvnorm{\mu-\pi} \leq \{ 2\KL(\mu | \pi) \}^{1/2}$, which implies
that the complexity of ULA for the total variation distance is of order
$d\bigO(\varepsilon^{-4})$. This discussion justifies the bounds we
state in \Cref{tab:comparison_gamma__convex_warm_start}.
}

{
In addition if we have access to $\eta >0$ and $M_{\eta}\geq 0$,
independent of the dimension, such that for all $x \in \rset^d$,
$x \not \in \boule{\xstar}{M_{\eta}}$,
$U(x) - U(\xstar) \geq \eta \norm{x-\xstar}$, $\xstar \in \argmin_{\rset^d} U$,
\Cref{propo:bound_wasser_init_condition} in Appendix~\ref{sec:second-order-moment} shows that for all
$ \int_{\rset^d} \norm[2]{x - \xstar} \rmd \pi(x) \leq 2 \eta^{-2}d
(1+d) + M_{\eta}^2$. Therefore, starting at $\updelta_{\xstar}$, the
overall complexity for the KL is in this case
$d^{3} \bigO(\varepsilon^{-2})$ and $d^{3} \bigO(\varepsilon^{-4})$
for the total variation distance. This discussion justifies the bound we
state in \Cref{tab:comparison_gamma__convex}.
}

We specify the consequences of \Cref{thm:const_step_conv} when $U$ is strongly convex. 
\begin{theorem}
\label{thm:const_step_conv_wasser}
Assume \Cref{assum:convexity}(m) for $m>0$ and \Cref{assum:grad_lip}. Let
$\sequenceks{\step_k}$ be a
non-increasing sequence of positive real numbers, $\step_1 \in\ocint{0,
L^{-1}}$, and
$\mu_0 \in \Pens_2(\rset^d)$.  Then for all $n \in
\nsets$, it holds
\begin{equation*}
\wassersteinTarg{\mu_0 \Qkers[n]}{ \pi} \leq \defEns{\prod_{k=1}^n (1-m\stepa[k])} \wassersteinTarg{\mu_0}{ \pi}
+ 2Ld \sum_{k=1}^n \stepa[k]^2 \prod_{i=k+1}^n(1-m\stepa[i]) \eqsp,
\end{equation*}
where $\Qkers[n]$ is defined in \eqref{eq:definition_average_measure}.
\end{theorem}
\begin{proof}
  Using \Cref{thm:basic-one-step} and since the
  Kullback-Leibler divergence is non-negative, we get for all $k \in \{1,\ldots,n\}$,
\begin{equation*}
\wassersteinTarg{\mu_0 \Qkers[k]}{ \pi} \leq (1-m\stepa[k])\wassersteinTarg{\mu_0 \Qkers[k-1]}{ \pi}  + 2 L d \stepa[k]^2 \eqsp.
\end{equation*}
The proof then follows from a direct induction.
\end{proof}
\begin{corollary}
  \label{cor:const_step_str_conv_wass_analysis}
Assume \Cref{assum:convexity}(m) for $m>0$ and \Cref{assum:grad_lip}. Let $\varepsilon >0$ and $\mu_0 \in \Pens_2(\rset^d)$. Define:
  \begin{equation*}
    \steps_{\varepsilon}  \leq \min\defEns{m\varepsilon/(4Ld), L^{-1}} \eqsp, \qquad \qquad 
    n_{\varepsilon}  \geq \lceil \ln( 2W_2^2(\mu_0, \pi)/ \varepsilon) \steps_{\varepsilon}^{-1} m^{-1} \rceil\eqsp. 
  \end{equation*}
Then we have
$\wassersteinTarg{\mu_0 \Rker^{n_{\varepsilon}}_{\steps_{\vareps}}}{ \pi} \leq \varepsilon
$, where $ \Rker_{\steps_{\vareps}}$ is defined by \eqref{eq:def_Rker_euler}.
\end{corollary}
\begin{proof}
By \Cref{thm:const_step_conv_wasser}, we have
\begin{equation*}
\wassersteinTarg{\mu_0 \Qkers[n_{\varepsilon}]}{\pi} \leq  \left( 1-m\steps_{\varepsilon} \right)^{n_{\varepsilon}} \wassersteinTarg{\mu_0}{ \pi}
+ 2Ld \sum_{k=1}^{n_{\varepsilon}} \steps_{\varepsilon}^2 (1-m\steps_{\varepsilon})^{n_{\varepsilon} - k} \eqsp.
\end{equation*}
On one hand, by definition of $\gamma_{\varepsilon}$, we get $2Ld\sum_{k=1}^{n_{\varepsilon}} \steps_{\varepsilon}^2 (1-m\steps_{\varepsilon})^{n_{\varepsilon} - k}  \leq 2Ld\steps_{\varepsilon}/m \leq \varepsilon/2$. On the other hand, using that for all $t \in \rset_+$, $1-t \leq \exp(-t)$ and the definition of $n_{\varepsilon}$, we obtain $\left( 1-m\steps_{\varepsilon} \right)^{n_{\varepsilon}} \wassersteinTarg{\mu_0}{ \pi} \leq \exp(-m \steps_{\varepsilon} n_{\varepsilon}) \wassersteinTarg{\mu_0}{ \pi} \leq \varepsilon/2$.
Then the thesis of the corollary follows  directly from the above inequalities.
\end{proof}

{
Note that the bound in the right hand side of
\Cref{thm:const_step_conv_wasser} is tighter than the previous bound given in
\cite[Theorem 1]{dalalyan:karagulyan:2017} (for constant step-size)
and \cite[Theorem 5]{durmus2016high} (for both constant and
non-increasing step-sizes). Indeed \cite[Theorem
1]{dalalyan:karagulyan:2017} shows that, in the constant step-size
setting $\gamma_k = \gamma$, for all
$k \in \nset$,
\begin{equation*}
  W_2(\mu_0 \Qkers[k], \pi) \leq (1-m\gamma)^k W_2(\mu_0,\pi) + 1.65 (L/m)(\step d)^{1/2} \eqsp. 
\end{equation*}
On the other hand, the inequality $(t+s)^{1/2}\leq t^{1/2}+s^{1/2}$ for $t,s \geq 0$ and  \Cref{thm:const_step_conv_wasser} imply that for all $k \in \nset$,
\begin{equation}
  \label{eq:discu_gaussian_ula}
  W_2(\mu_0 \Qkers[k], \pi) \leq (1-m\gamma)^{k/2} W_2(\mu_0,\pi) + \{2 \step d L/m\}^{1/2} \eqsp. 
\end{equation}
Thus, the dependency on the condition number $L/m$ is improved. This
bound is in agreement for the case where $\pi$ is the zero-mean
$d$-dimensional Gaussian distribution with covariance matrix $\Sigma$. In that case, all the iterates $(X_k)_{k \in \nsets}$ defined by
\eqref{eq:definition_em} for $\step >0$, starting from
$x \in \rset^d$, follows a Gaussian distribution with mean
$(\Id-\step \Sigma)^{k}x$ and covariance matrix
$2 \step \sum_{i=0}^{k-1}(1-\step \Sigma)^{2i}$. Since the Wasserstein
distance between $d$-dimensional Gaussian distributions can be
explicitly computed, see \cite{givens:shortt:1984}, denoting by $L$
and $m$ the largest and smallest eigenvalues of $\Sigma$
respectively, we have by an explicit calculation for
$\step \in \ocint{0,L^{-1}}$,
\begin{equation*}
    W_2(\mu_0 \Qkers[k], \pi) \leq  (1-m\gamma)^k W_2(\mu_0,\pi) + (d/ m )^{1/2}\defEns{(1-\step L/2)^{-1/2}-1} \eqsp.  
  \end{equation*}
  Since for $t \in \ccint{0,1/2}$, $(1-t)^{-1/2} -1-t \leq 0$, we get
\begin{equation*}
    W_2(\mu_0 \Qkers[k], \pi) \leq   (1-m\gamma)^k W_2(\mu_0,\pi) + 2^{-1} \step (d/m )^{1/2} \defEns{(1-\step L)^{-1/2}-1} \eqsp.  
  \end{equation*}
  Using that $\step \leq L^{-1}$, we get that the second term in the
  right hand side is bounded by $(d L \step/m)^{1/2}$, which is precisely
  the order we get from \eqref{eq:discu_gaussian_ula}. 
}

{ Finally, if $(\gamma_k)_{k \in \nsets}$ is given for all
  $k \in \nsets$, by $\gamma_k=\gamma_1/k^{\alpha}$, for
  $\alpha \in \ooint{0,1}$, then using \cite[Lemma 7]{durmus2016high}
  and the same calculation of \cite[Section
  6.1]{durmus:moulines:2015:supplement}, we get that there exists
  $C \geq 0$ such that for all $n \in \nsets$,
  $ W_2(\mu_0 \Qkers[n], \pi) \leq C n^{-\alpha/2}$. }

{
Based on \Cref{thm:const_step_conv_wasser}, we can improve \Cref{coro:eps_just_convex_ula} in the case where $U$ is strongly convex using an appropriate burn-in time.  
}
\begin{corollary}
\label{cor:const_step_str_conv_kl_analysis}
  Assume \Cref{assum:convexity}(m) for $m>0$ and \Cref{assum:grad_lip}. Let $\varepsilon >0$, $\mu_0 \in \Pens_2(\rset^d)$ and
  \begin{align*}
  \steps_{\varepsilon} & \leq \min\defEns{m\varepsilon/(4Ld), L^{-1}}   \eqsp, \qquad \quad  &    \tilde{\steps}_{\varepsilon} & \leq \min\defEns{\varepsilon/2Ld, L^{-1}} \eqsp, \\
    N_{\varepsilon} & \geq \lceil \ln( 2W_2^2(\mu_0, \pi)/ \varepsilon) (\steps_{\varepsilon}m)^{-1} \rceil \eqsp, \qquad \qquad & 
    n_{\varepsilon} & \geq \lceil \tilde{\steps_{\varepsilon}}^{-1} \rceil \eqsp.
  \end{align*}
  Let $(\step_k)_{k \in \nset}$ defined by  $\steps_k = \steps_{\varepsilon}$ for $ k \in \{1,\ldots,  N_{\varepsilon}\}$ and $ \steps_k = \tilde{\steps_{\varepsilon}}$ for $k > N_{\varepsilon}$.  Then we have $ \KLarg{\nu^{N_{\varepsilon}}_{n_{\varepsilon}}}{\pi} \leq \varepsilon$
where $\nu_{n_{\varepsilon}}^{N_{\varepsilon}} = n_{\varepsilon}^{-1} \sum_{k=1}^{n_{\varepsilon}} \mu_0 \Rker_{\steps_{\varepsilon}}^{N_{\varepsilon}}  \Rker_{\tilde{\steps}_{\varepsilon}}^k$.
\end{corollary}
\begin{proof}
Using \Cref{cor:const_step_str_conv_wass_analysis}, we have $\wassersteinTarg{\mu_0 \Qkers[N_{\varepsilon}]}{ \pi} \leq \varepsilon$. 
Now applying \Cref{thm:const_step_conv} we get:
\begin{align*}
\KLarg{\nu^{N_{\varepsilon}}_{n_{\varepsilon}}}{\pi}  \leq \left. \wassersteinTarg{\mu_{N_{\varepsilon}}}{ \pi} \middle/ (2 \tilde{\steps_{\varepsilon}} n_{\varepsilon}) \right.
+ (Ld/n_{\varepsilon}\tilde{\steps_{\varepsilon}}) \sum_{k=N_{\varepsilon} + 1}^{N_{\varepsilon} + n_{\varepsilon}} (\tilde{\steps_{\varepsilon}})^2  \leq \varepsilon / (2 \tilde{\steps_{\varepsilon}} n_{\varepsilon})+ Ld \tilde{\steps_{\varepsilon}} \leq \varepsilon
\end{align*}
\end{proof}

{ By \cite[Proposition 1]{durmus2016high}, we have
  $\int_{\rset^d} \norm[2]{x-\xstar} \rmd \pi(x) \leq d/m$, where $\xstar = \argmin_{\rset^d} U$. Therefore
  we have that in the constant step size setting,
  $\step_k = \step \in \ocintLigne{0,L^{-1}}$ for all $k \in \nsets$,
  \Cref{cor:const_step_str_conv_wass_analysis} implies that a
  sufficient number of iterations to have
  $W_2(\updelta_{\xstar} \Qkers[n], \pi) \leq \varepsilon$ is of order
  $\bigO(\varepsilon^{-2}d)$. Then
  \Cref{cor:const_step_str_conv_kl_analysis} implies that a sufficient
  number of iterations to get $\KLarg{\nu^N_n}{\pi} \leq \varepsilon$,
  for $\varepsilon >0$, is of order $\bigO(\varepsilon^{-1}d)$. By
  Pinsker inequality, we obtain that a sufficient number of iterations
  to get $\tvnorm{\nu^N_n-\pi} \leq \varepsilon$, for
  $\varepsilon >0$, is of order $d\bigO(\varepsilon^{-2})$.  }

For a sufficiently small constant step size $\step$, ULA produces a Markov Chain with a stationary measure
$\pi_{\step}$. In general this measure is different from the measure
of interest $\pi$. Based on our previous results, we establish computable
bounds on the distance between $\pi$ and $\pi_{\step}$.
\begin{theorem}
\label{thm:pi_pi_gamma_ULA}
  Assume \Cref{assum:convexity}$(m)$ for $m \geq 0$ and \Cref{assum:grad_lip}.
Let $\steps \in \ocint{0,L^{-1}}$. Then there exists a measure $\pi_{\step}$, such that $\pi_{\step} \Rker_{\steps} = \pi_{\step}$
where $\Rker_{\steps}$ is defined by \eqref{eq:def_Rker_euler}. In addition,  we have
\[
\KLarg{\pi_{\step}}{\pi} \leq Ld\step, \qquad \qquad \tvnorm{\nu^N_n-\pi} \leq \sqrt{2Ld\step}
\]
Furthermore, if $m > 0$ we also have $\wassersteinTarg{\pi_{\step}}{ \pi} \leq \fraca{2Ld \step}{m}$.
\end{theorem}
\begin{proof}
  Under \Cref{assum:convexity} and \Cref{assum:grad_lip},
\cite[Proposition 13]{durmus:moulines:2017} shows that $R_{\gamma}$
satisfies a geometric Foster-Lyapunov drift condition for
$\gamma \leq L^{-1}$. In addition, it is easy to see that $R_{\gamma}$
is $\Leb$-irreducible and weak Feller and therefore by \cite[Theorem 6.0.1 together with Theorem 5.5.7 ]{meyn:tweedie:2009}, all compact sets are small. Then, by \cite[Theorem 16.0.1]{meyn:tweedie:2009},
$R_{\gamma}$ has a unique invariant distribution $\pi_{\gamma}$.

Second, taking $\mu = \pi_{\step}$ in \Cref{thm:basic-one-step} we obtain:
\begin{equation}
2\steps \KLarg{\pi_{\steps}\Rker_{\steps}}{\pi} \leq (1-m\steps)\wasserstein_2^2(\pi_{\step}, \pi) - \wasserstein_2^2(\pi_{\steps} \Rker_{\steps} , \pi) + 2\steps^2 L d \eqsp,
\end{equation}
and because $\pi_{\steps} \Rker_{\steps} = \pi_{\steps}$, the above implies $2 \KLarg{\pi_{\steps}}{\pi} + m\wasserstein_2^2(\pi_{\steps}, \pi) \leq 2Ld\steps$.
Since both the KL-divergence and Wasserstein distance are positive, the desired bounds in KL and $\wasserstein_2^2$ follow. The bound in total variation follows from the bound in KL-divergence and Pinsker inequality.
\end{proof}


\section{Extensions of ULA}
\label{sec:expl-bounds-extens}

In this section, two extensions of ULA are presented and
analyzed. These two algorithms can be applied to non-continuously differentiable 
convex potential $U: \rset^d \to \rset$ and therefore
\Cref{assum:grad_lip} is not assumed anymore.  In addition, for the two new
algorithms we present, only ~\iid~unbiased estimates of
(sub-)gradients of $U$ are necessary as in Stochastic Gradient
Langevin Dynamics (SGLD) \cite{welling:teh:2011}. The main difference
in these two approaches is that one relies on the sub-gradient
of $U$ 
while 
the other is based on proximal operators which are tools commonly used in non-smooth optimization. However, theoretical
results that we can show for these two algorithms, hold for different sets of
conditions.

\subsection{Stochastic Sub-Gradient Langevin Dynamics}
\label{sec:stoch-sub-grad}

Note that if
$U$ is convex and \lsc~then for any point $x \in \rset^d$, its
sub-differential $\partial U(x)$ defined by
\begin{equation}
\label{eq:definition_partial_U}
  \partial U(x) = \ensemble{v \in  \rset^d}{U(y) \geq U(x) + \ps{v}{y-x} \text{ for all $y \in \rset^d$}} \eqsp,
\end{equation}
is non empty, see \cite[Proposition 8.12, Theorem
8.13]{rockafellar:wets:1998}. For all $x \in \rset^d$, any elements of
$\partial U(x)$ is referred to as a sub-gradient of $U$ at $x$.
Consider the following condition on $U$ which assumes that we have
access to unbiased estimates of sub-gradients of $U$ at any point $x \in \rset^d$.
\begin{assumption}
  \label{assum:stochastic_subgradient}
\begin{sf}
  \begin{enumerate}[label=(\roman*)]
  \item   \label{assum:stochastic_subgradient_i} The potential $U$ is  $\ClyapU$-Lipschitz, \ie~for
    all $x,y \in \rset^d$, $\abs{U(x) -U(y) } \leq \ClyapU\norm{x-y}$.
  \item \label{assum:stochastic_subgradient_ii} There exists a measurable space $(\msz,\mcz)$, a probability
  measure $\eta$ on $(\msz,\mcz)$ and a measurable function $\gradst : \rset^d \times \msz \to \rset^d$  for all $x \in \rset^d$,
  \begin{equation*}
    \int_{\msz} \gradst(x, z) \rmd \eta(z) \in \partial U(x) \eqsp.
  \end{equation*}
  \end{enumerate}
\end{sf}
\end{assumption}
Note that under
\Cref{assum:stochastic_subgradient}-\ref{assum:stochastic_subgradient_i},
for all $x \in \rset^d$ and  $v \in \partial U(x)$, 
\begin{equation}
  \label{eq:stochastic_subgradient_i}
\norm{v} \leq \ClyapU \eqsp.
\end{equation}

Let $\sequenceks{Z_k}$ be a sequence of
\iid~random variables distributed according to $\eta$,
$\sequenceks{\step_k}$ be a sequence of non-increasing step sizes and
$\bX_0$ distributed according to $\mu_0 \in
\Pens_2(\rset^d)$. Stochastic Sub-Gradient Langevin Dynamics (SSGLD)
defines the sequence of random variables $\sequencek{\bX_k}$ starting
at $\bX_0$ for $n \geq 0$ by
\begin{equation}
\label{eq:definition_SSGLD}
  \bX_{n+1} = \bX_n - \stepa[n+1] \gradst(\bX_n,Z_{n+1})  + \sqrt{2 \stepa[n+2]} G_{n+1} \eqsp,
\end{equation}
where $\sequenceks{G_k}$ is a sequence of \iid~$d$-dimensional standard Gaussian random variables, independent of $\sequenceks{Z_k}$, see \Cref{algo:SSDLG}.  Consequently
this method defines a new sequence of Markov kernels
$\sequenceks{\bRkera[k]}$ given for all $\step,\tstep >0$, $x \in \rset^d$ and $\eventA \in \Borel(\rset^d)$ by
\begin{equation}
\label{eq:definition_rker_ssgld}
  \bRker_{\step,\tstep}(x,\eventA)  =  (4 \uppi \tstep)^{-d/2} \int_{\eventA \times \msz} \exp\parenthese{-\norm[2]{y-x+\step \gradst(x, z)}/(4 \tstep)} \rmd \eta(z) \rmd y \eqsp. 
\end{equation}

\begin{algorithm}[H]
  \KwData{initial distribution $\mu_0 \in \Pens_2(\rset^d)$, non-increasing sequence $(\gamma_k)_{k \geq 1}$, $U,\Theta, \eta$ satisfying \Cref{assum:stochastic_subgradient}}
 \KwResult{$(\bX_k)_{k \in \nset}$}
 \Begin{
   Draw $\bX_0 \sim \mu_0$ \;
   \For{$k \geq 0$ }{ 
   Draw $G_{k+1} \sim \mathcal{N}(0, \Id)$ and $Z_{k+1} \sim \eta$ \;
   Set $\bX_{k+1} = \bX_k - \stepa[k+1] \gradst(\bX_k,Z_{k+1})  + \sqrt{2 \stepa[k+2]} G_{k+1}$ 
    } 
    }
 \caption{SSGLD}
\label{algo:SSDLG}
\end{algorithm}

Let $\sequenceks{\step_k}$ and $\sequenceks{\weight_k}$ be two non-increasing sequences of reals
numbers and $\mu_0 \in \Pens_2(\rset^d)$ be an initial
distribution. The weighted averaged distribution associated with
\eqref{eq:definition_SSGLD} $(\bnu^{N}_n)_{n \in \nset}$ is defined
for all $N, n \in \nset$, $n \geq 1$ by
\begin{equation}
  \label{eq:definition_average_measure}
\bnu^N_n = \Weight_{N,N+n} ^{-1}\sum_{k=N+1}^{N+n} \weight_{k}\, \mu_0 \bQkers[k] \eqsp, \qquad \bQkers[k] = \bRker_{\stepa[1],\stepa[2]} \cdots \bRkers[k] \eqsp, \text{ for $k \in \nsets$} \eqsp,
\end{equation}
where 
$N$ is a burn-in time and $\Weight_{N,N+n}$ is defined in \eqref{eq:definition_sum_step_weight}. We take in the
following the convention that $\bQkers[0]$ is the identity operator.

Under \Cref{assum:stochastic_subgradient}, define for all $\mu \in
\Pens_2(\rset^d)$,
\begin{equation}
  \label{eq:definition_var_grad_sto}
\vargrad_{\gradst}(\mu) =  \int_{\rset^d \times \msz} \norm[2]{\gradst(x,z)- \int_{\msz} \gradst(x,\tilde{z}) \rmd \eta(\tilde{z})} \rmd \eta(z) \rmd \mu(x) = \expe{\norm[2]{\gradst(\bX_0,Z_1) - v}} \eqsp,
\end{equation}
where $\bX_0,Z_1$ are  independent random variables with distribution $\mu$ and $\eta_1$ respectively and $v \in \partial U(X_0)$ almost surely.
In addition, consider $\bSker_{\step}$, the Markov kernel on $(\rset^d,\Borel(\rset^d))$
defined for all $x \in \rset^d$ and $\eventA \in \Borel(\rset^d)$ by
\begin{equation}
  \label{eq:definition_sker_ssgld}
  \bSker_{\step}(x,\eventA) =  \int_{\msz} \1_{\eventA}\parenthese{x-\steps \gradst(x, z)} \rmd \eta(z) \eqsp.
\end{equation}

\begin{theorem}
\label{thm:step_conv_ss}
Assume \Cref{assum:convexity}($0$) and
\Cref{assum:stochastic_subgradient}. Let $\sequenceks{\step_k}$ and
$\sequenceks{\weight_k}$ be two non-increasing sequences of positive
real numbers satisfying for all $k\in \nsets$,
$\weighta[k+1]/\stepa[k+2]
\leq\weighta[k]/\stepa[k+1]$. Let $\mu_0 \in \Pens_2(\rset^d)$ and $N
\in \nset$.  Then for all $n \in \nsets$, it holds
\begin{multline*}
\KLarg{\bnu^N_n}{\pi} \leq \left. \weighta[N+1] \wassersteinTarg{\mu_0 \bQkers[N] \bSker_{\stepa[N+1]}}{ \pi} \middle/ (2 \stepa[N+2] \Weighta[N,N+n]) \right.\\
+ (2\Weighta[N,N+n])^{-1}  \sum_{k=N+1}^{N+n} \defEns{\stepa[k+1] \weighta[k] \parenthese{\ClyapU^2 + \vargrad_{\gradst}(\mu_0 \bQkera[k])}} \eqsp,
\end{multline*}
where $\bnu^N_n$ and $\bQkers[N]$ are defined in \eqref{eq:definition_average_measure}.
\end{theorem}
\begin{proof}
The proof is postponed to \Cref{sec:proof-crefthm:st}..
\end{proof}

Note that in the bound given by \Cref{thm:step_conv_ss}, we need to
control the ergodic average of the variance of the stochastic gradient
estimates. When \Cref{assum:stochastic_subgradient} is satisfied, a
possible assumption is that $x \mapsto \vargrad(\updelta_x)$ is
uniformly bounded. This assumption will be satisfied for example when
the potential $U$ is a sum of Lipschitz continuous functions.

\begin{corollary}
  \label{coro:ssgld_fixed_one}
Assume \Cref{assum:convexity}($0$) and
\Cref{assum:stochastic_subgradient}. Assume that $\sup_{x \in \rset^d} \vargrad_{\gradst}(\updelta_x) \leq D^2 < \infty$. Let $\sequenceks{\step_k}$ and $\sequenceks{\weight_k}$ given for all $k \in \nsets$ by  $\weight_k = \steps_k = \steps >0$. Let $\mu_0 \in \Pens_2(\rset^d)$. Then for any $N \in \nset, n \in \nsets$ we have
\[
\KLarg{\bnu^N_n}{\pi} \leq \left. \wassersteinTarg{\mu_0 \bQkers[N] \bSker_{\step}}{ \pi} \middle/ (2 n \steps) \right.
+ (\steps/2) \left(\ClyapU^2 + D^2\right) \eqsp.
\]
Furthermore, let $\varepsilon > 0$ and
\begin{equation*}
    \steps_{\varepsilon}  \leq \varepsilon/(M^2 + D^2)  \eqsp, \qquad \qquad  n_{\varepsilon}  \geq \lceil W_2^2(\mu_0 \bSker_{\steps}, \pi) (\steps_{\varepsilon} \varepsilon)^{-1} \rceil \eqsp.
  \end{equation*}
Then for $\steps = \steps_{\varepsilon}$ we have $\KLarg{\bnu^0_{n_\varepsilon}}{\pi} \leq \varepsilon$.
\end{corollary}
\begin{proof}
The first inequality is a direct consequence of \Cref{thm:step_conv_ss}. The bound for $\KLarg{\bnu^0_{n_\varepsilon}}{\pi}$ follows directly from this inequality and definitions of $\steps_{\varepsilon}$ and $n_{\varepsilon}$.
\end{proof}

{ In the case where a warm start is available for the Wasserstein distance,
  \ie~$W_2^2(\mu_0,\pi) \leq C$, for some absolute constant
  $C \geq 0$, then \Cref{coro:ssgld_fixed_one} implies that the
  complexity of SSGLD to obtain a sample close from $\pi$ in KL with a
  precision target $\varepsilon>0$ is of order
  $(M^2+D^2)\bigO(\varepsilon^{-2})$. Therefore, this complexity bound
  depends on the dimension only trough $M$ and $D^2$ contrary to ULA. In addition,  Pinsker inequality implies
  that the complexity of SSGLD for the total variation distance is of
  order $(M^2+D^2)\bigO(\varepsilon^{-4})$.  }

{
In addition if we have access to $\eta >0$ and $M_{\eta}\geq 0$,
independent of the dimension, such that for all $x \in \rset^d$,
$x \not \in \boule{\xstar}{M_{\eta}}$,
$U(x) - U(\xstar) \geq \eta \norm{x-\xstar}$, where $\xstar \in \argmin_{\rset^d} U$,
\Cref{propo:bound_wasser_init_condition} and \Cref{assum:stochastic_subgradient}-\ref{assum:stochastic_subgradient_i} imply that starting at $\updelta_{\xstar}$, the
overall complexity of SSGLD for the KL is in this case
$(\eta^{-2}d^{2}+M_\eta^2 + M^2)(M^2+D^2) \bigO(\varepsilon^{-2})$ and $(\eta^{-2} d^{2}+M_\eta^2+M^2)(M^2+D^2) \bigO(\varepsilon^{-4})$
for the total variation distance.
}

{If $(\gamma_k)_{k \in \nsets}$ and $(\lambda_k)_{k \in \nsets}$
  are given for all $k \in \nsets$ by
  $\gamma_k = \lambda_k = \gamma_1/k^{-\alpha}$, with $\alpha \in \ooint{0,1}$, then by the same reasoning as in
  the proof of \Cref{coro:ula_non_increas_sz}, we obtain that there
  exists $C \geq 0$ such that for all $n \in \nsets$, we have
  $\KLarg{\bnu_n^0}{\pi} \leq C \max(n^{\alpha - 1}, n^{- \alpha})$, if $\alpha \not = 1/2$,  and 
  for $\alpha = 1/2$, we have
  $\KLarg{\bnu_n^0}{\pi} \leq C (\ln(n) + 1) n^{-1/2}$.
}

 We can have a better control on  the variance terms using the
following conditions on $\Theta$.
\begin{assumption}
  \label{assum:cocoercitivity_sto_grad}
\begin{sf}
  There exists $\Lt \geq 0$ such that for $\eta$-almost every $z
  \in \msz$, $x \mapsto \gradst(x,z)$ is $1/\Lt$-cocoercive, \ie~for all $x \in \rset^d$,
  \begin{equation*}
    \ps{\gradst(x,z) - \gradst(y,z)}{x-y} \geq (1/\tL) \norm[2]{\gradst(x,z) - \gradst(y,z)}\eqsp.
  \end{equation*}
\end{sf}
\end{assumption}

{
This assumption is for example satisfied if $\eta$-almost every $z$,
$x \mapsto \Theta(x,z)$ is the gradient of a continuously
differentiable convex function with Lipschitz gradient, see \cite[Thereom 2.1.5]{nesterov:2004} and \cite{Zhu_Mar_1995}. }

\begin{proposition}
  \label{propo:bound_variance_sto_grad}
  Assume \Cref{assum:stochastic_subgradient} and \Cref{assum:cocoercitivity_sto_grad}. 
Then we have for all $x \in \rset^d$ and $\gamma,\tilde{\gamma} >0$, $\step \leq \tilde{L}^{-1}$
\begin{equation*}
2\gamma(\tilde{L}^{-1}-\gamma)\vargrad_{\gradst}(\updelta_x) \leq \norm{x-\xstar}^2 - \int_{\rset^d} \norm[2]{y-\xstar}\bar{R}_{\gamma,\tgamma}(x,\rmd y) + 2 \gamma^2 \vargrad_{\gradst}(\updelta_{\xstar})+  2 \tgamma d \eqsp,
\end{equation*}
where $\vargrad_{\gradst}$ is defined by \eqref{eq:definition_var_grad_sto}.
\end{proposition}
\begin{proof}
  Consider $ \bX_1 = x - \gamma \gradst(x,\gradrv_1) + \sqrt{2 \tgamma}
  G_1$, where $\gradrv_1$ and $G_1$ are two independent random
  variables, $\gradrv_1$ has distribution $\eta$ and $G_1$ is a
  standard Gaussian random variables.  Then using
  \Cref{assum:cocoercitivity_sto_grad}, we have
\begin{align*}
  \expe{\norm[2]{\bX_1 - \xstar}} &= \expe{\norm{x - \gamma \gradst(x,\gradrv_1) - \xstar}} + 2 \tgamma d \\
&=  \norm[2]{x-\xstar}+  \expe{\gamma^2 \norm[2]{\gradst(x,\gradrv_1)}-2\gamma\ps{\gradst(x,\gradrv_1)}{x-\xstar} } + 2\tgamma d \\
& \leq \norm[2]{x-\xstar} -2\gamma(\tilde{L}^{-1}-\gamma)\expe{\norm[2]{\gradst(x,\gradrv_1) - \gradst(\xstar,\gradrv_1)}}\\
&\phantom{aaaaaaaaaaaaaaaaaaaaaaaaaaaa} +2 \gamma^2\expe{\norm[2]{\gradst(\xstar,\gradrv_1)}}+  2 \tgamma d  \eqsp.
\end{align*}
The proof is completed upon noting that $\vargrad_{\gradst}(\updelta_x) \leq \expeLigne{\norm[2]{\gradst(x,\gradrv_1) - \gradst(\xstar,\gradrv_1)}}$ and $\vargrad_{\gradst}(\updelta_{\xstar}) = \expe{\norm[2]{\gradst(\xstar,\gradrv_1)}}$
\end{proof}

Combining \Cref{thm:step_conv_ss} and \Cref{propo:bound_variance_sto_grad}, we get the following result. 
\begin{corollary}
\label{coro:fixed_step_conv_ss}
Assume \Cref{assum:convexity}($0$)-\Cref{assum:stochastic_subgradient} and \Cref{assum:cocoercitivity_sto_grad}. Let $\sequenceks{\step_k}$ and
$\sequenceks{\weight_k}$ defined for all $k\in \nsets$ by $\step_k=\weight_k = \gamma \in \oointLigne{0,\tilde{L}^{-1}}$. Let $\mu_0 \in \Pens_2(\rset^d)$.  Then for all  $N
\in \nset$ and $n \in \nsets$, we have
\begin{multline*}
  \KLarg{\bnu^N_n}{\pi} \leq \left.  \wassersteinTarg{\mu_0 \bRker_{\step,\step}^N \bSker_{\step}}{ \pi} \middle/ (2\step n) \right.
\\+ \step M^2/2 + (2(\tilde{L}^{-1}-\gamma))^{-1}\defEns{(2n)^{-1}\int_{\rset^d} \norm[2]{x-\xstar} \rmd \mu_0 \bRker_{\step,\step}^{N+1}(x) + \step^2 \vargrad_{\gradst}(\updelta_{\xstar})+  \step d} \eqsp.
\end{multline*}
Furthermore, let $\varepsilon > 0$ and
\begin{align*}
    \steps_{\varepsilon}  \leq  \min\parentheseDeux{ \left. \varepsilon\middle /\defEns{2M^2 + 4\tilde{L}d} \right. , \sqrt{ \varepsilon  \left(4 \tilde{L} \vargrad_{\gradst}(\updelta_{\xstar})  \right)^{-1}}, (2\tilde{L})^{-1}}   \eqsp, \\   n_{\varepsilon}  \geq 2\max\defEns{ \ceil{ W_2^2(\mu_0 \bSker_{\steps_{\varepsilon}}, \pi) (\steps_{\varepsilon} \varepsilon)^{-1} }, \ceil{\tilde{L}\varepsilon^{-1} \int_{\rset^d} \norm[2]{x-\xstar} \rmd \mu_0 \bRker_{\steps_{\varepsilon},\steps_{\varepsilon}}(x) }} \eqsp.
  \end{align*}
Then for  $\step = \steps_{\varepsilon}$, then we have $\KLarg{\bnu^0_{n_\vareps}}{\pi} \leq \varepsilon$.
\end{corollary}
\begin{proof}
The proof is postponed to \Cref{sec:proof-coro_fixed_step_conv_ss}.
\end{proof}

{
Note that compared to \Cref{coro:ssgld_fixed_one}, the dependence on
the variance of the stochastic sub-gradients in the bound on
$n_{\varepsilon}$, given in \Cref{coro:fixed_step_conv_ss}, is less
significant since $n_{\varepsilon}$ scales as
$(\vargrad_{\gradst}(\updelta_{\xstar}))^{1/2}$ and not as 
$\sup_{x \in \rset^d} \vargrad_{\gradst}(\updelta_{x})$. However, the
dependency on the dimension deteriorates a little. 
}

\subsection{Stochastic Proximal Gradient Langevin Dynamics}
\label{sec:stoch-prox-grad}
In this section, we propose and analyze an other algorithm to handle
non-smooth target distribution using stochastic gradient estimates and
proximal operators. For $m \geq 0$, consider the following assumptions
on the gradient.

\begin{assumption}[$m$]
  \label{assum:stochastic_gradient_proximal}
\begin{sf}
There exists $U_1: \rset^d \to \rset$ and $U_2 : \rset^d \to \rset$ such that $U=U_1+U_2$ and satisfying the following assumptions:
\begin{enumerate}
\item $U_1$ satisfies \Cref{assum:convexity}($m$) and \Cref{assum:grad_lip}. In addition, 
there exists a measurable space $(\tmsz,\tmcz)$, a probability
measure $\teta_1$ on $(\tmsz,\tmcz)$ and a measurable function $\tgradst_1 : \rset^d \times \msz \to \rset^d$ such that for all $x \in \rset^d$,
  \begin{equation*}
    \int_{\tmsz} \tgradst_1(x, \tz) \rmd \teta_1(\tz) = \nabla U_1(x) \eqsp.
  \end{equation*}
\item $U_2$ satisfies \Cref{assum:convexity}($0$) and is $M_2$-Lipschitz.
\end{enumerate}
\end{sf}
\end{assumption}

Under \Cref{assum:stochastic_gradient_proximal}, consider the proximal
operator associated with $U_2$ with parameter $\step >0$ (see~\cite[Chapter 1 Section G]{rockafellar:wets:1998}), defined for all $x \in \rset^d$ by 
\begin{equation*}
  \prox_{U_2}^\step(x) = \argmin_{y \in \rset^d} \defEns{U_2(y) + (2 \step)^{-1}\norm[2]{x-y}} \eqsp.
\end{equation*}

Let $\sequenceks{\tZ_k}$ be a sequence of
\iid~random variables distributed according to $\eta_1$,
$\sequenceks{\step_k}$ be a sequence of non-increasing step sizes and
$\tX_0$ distributed according to $\mu_0 \in
\Pens_2(\rset^d)$. Stochastic Proximal Gradient Langevin Dynamics (SPGLD)
defines the sequence of random variables $\sequencen{\tX_n}$ starting
at $\tX_0$ for $n \geq 0$ by
\begin{equation}
\label{eq:definition_SPGLD}
  \tX_{n+1} = \prox_{\stepa[n+1]}^{U_2}(\tX_n) - \stepa[n+2]  \tgradst_1\{\prox_{\stepa[n+1]}^{U_2}(\tX_n),\tZ_{n+1}\} + \sqrt{2 \stepa[n+2]} G_{n+1} \eqsp,
\end{equation}
where $\sequenceks{G_k}$ is a sequence of \iid~$d$-dimensional standard Gaussian random variables, independent of $\sequenceks{Z_k}$.  The recursion
\eqref{eq:definition_SPGLD} is associated with the family of Markov
kernels $\sequenceks{\tRkera[k]}$ given for all $\step,\tstep >0$,
$x \in \rset^d$ and $\eventA \in \Borel(\rset^d)$ by
\begin{multline}
\label{eq:definition_rker_spgld}
  \tRker_{\step,\tstep}(x,\eventA) \\ = (4 \uppi \tstep)^{-d/2}  \int_{\msa \times \msz} \exp\parenthese{-\left. \norm[2]{y-\prox_{\step}^{U_2}(x)+\tstep  \tgradst_1\{\prox_{\step}^{U_2}(x), z\}}\middle/(4 \tstep) \right.} \rmd \eta_1(z) \rmd y \eqsp. 
\end{multline}

Note that for all $\step,\tstep >0$, $\tRker_{\step,\tstep}$ can be decomposed as the product $ \tS^2_{\step} \tS^1_{\tstep}\Tker_{\tstep}$ where $\Tker_{\tstep}$ is defined by \eqref{eq:def_Sker_Tker} and for all $x \in \rset^d$ and $\eventA \in \Borel(\rset^d)$ 
\begin{equation}
  \label{eq:definition_sker_spgld}
    \tS^1_{\tstep}(x,\eventA)  =   \int_{ \msz}\1_{\msa}(x-\tstep  \tgradst_1(x, z))  \rmd \eta_1(z) \eqsp, \qquad     \tS^2_{\step}(x,\eventA)  =  \updelta_{\prox_{\step}^{U_2}(x)}(\eventA) \eqsp.
\end{equation}

\begin{algorithm}[H]
  \KwData{initial distribution $\mu_0 \in \Pens_2(\rset^d)$, non-increasing sequence $(\gamma_k)_{k \geq 1}$, $U = U_1 + U_2,\tilde{\Theta}_1, \eta_1$ satisfying \Cref{assum:stochastic_gradient_proximal}}
 \KwResult{$(\tX_k)_{k \in \nset}$}
 \Begin{
   Draw $\tX_0 \sim \mu_0$\; 
   \For{$k \geq 1$ }{ 
   Draw $G_{k+1} \sim \mathcal{N}(0, \Id)$ and $\tZ_{k+1} \sim \eta_1$ \;
   Set $\tX_{k+1} = \prox_{\stepa[k+1]}^{U_2}(\tX_k) - \stepa[k+2]  \tgradst_1(\prox_{\stepa[k+1]}^{U_2}(\tX_k),\tZ_{k+1}) + \sqrt{2 \stepa[k+2]} G_{k} $ 
    } 
    }
 \caption{SPGLD}
\label{algo:SPGLD}
\end{algorithm}

Let$\sequenceks{\step_k}$ and  $\sequenceks{\weight_k}$ be two non-increasing sequences of reals
numbers and $\mu_0 \in \Pens_2(\rset^d)$ be an initial
distribution. The weighted averaged distribution associated with
\eqref{eq:definition_SPGLD} $(\tnu^{N}_n)_{n \in \nset}$ is defined
for all $N, n \in \nset$, $n \geq 1$ by
\begin{equation}
  \label{eq:definition_average_measure_spgld}
\tnu^N_n = \Weight_{N,N+n} ^{-1}\sum_{k=N+1}^{N+n} \weight_{k}\, \mu_0 \tQkers[k] \eqsp, \qquad \tQkers[k] = \tRker_{\stepa[1],\stepa[2]} \cdots \tRkers[k] \eqsp, \text{ for $k \in \nsets$} \eqsp,
\end{equation}
where 
$N$ is a burn-in time and $\Weight_{N,N+n}$ is defined in \eqref{eq:definition_sum_step_weight}. We take in the
following the convention that $\tQkers[0]$ is the identity operator.

Under \Cref{assum:stochastic_subgradient}, define for all $\mu \in
\Pens_2(\rset^d)$,
\begin{multline}
  \label{eq:definition_var_grad_sto_spgl}
\vargrad_{1}(\mu) =  \int_{\rset^d \times \msz} \norm[2]{ \tgradst_1(x,z)- \int_{\msz}  \tgradst_1(x,\tilde{z}) \rmd \eta_1(\tilde{z})} \rmd \eta(z) \rmd \mu(x) \\= \expe{\norm[2]{ \tgradst_1(\tX_0,Z_1) - \nabla U_1(\tX_0)}} \eqsp,
\end{multline}
where $\tX_0,\tZ_1$ are  independent random variables with distribution $\mu$ and $\eta_1$ respectively.

\begin{theorem}
\label{thm:step_conv_sp}
Assume \Cref{assum:stochastic_gradient_proximal}$(m)$, for $m \geq 0$. Let $\sequenceks{\step_k}$ and
$\sequenceks{\weight_k}$ be two non-increasing sequences of positive
real numbers satisfying $\step_1 \in\ocint{0,
L^{-1}}$, and for all $k\in \nsets$,
$\weighta[k+1]/\stepa[k+2]
\leq\weighta[k]/\stepa[k+1]$. Let $\mu_0 \in \Pens_2(\rset^d)$ and $N
\in \nset$.  Then for all $n \in \nsets$, we have
\begin{multline*}
\KLarg{\tnu^N_n}{\pi} \leq \left. \weighta[N+1] \wassersteinTarg{\mu_0 \tQkers[N] \tSker_{\stepa[N+1]}^2}{ \pi} \middle/ (2 \stepa[N+2] \Weighta[N,N+n]) \right.\\
+(2 \Weighta[N,N+n])^{-1} \sum_{k=N+1}^{N+n}   \weighta[k] \stepa[k+1]\{2Ld + (1+\stepa[k+1] L) \vargrad_1(\mu_0 \Qkers[k-1] \tS^2_{\step_{k}}) + 2 M_2^2\} \eqsp.
\end{multline*}

\end{theorem}
\begin{proof}
The proof is postponed to \Cref{sec:proof-crefthm:st-1}.
\end{proof}

\begin{corollary}
  \label{coro:spgld_fixed_one}
  Assume \Cref{assum:stochastic_gradient_proximal}$(m)$, for $m \geq 0$.  Assume that $\sup_{x \in \rset^d} \vargrad_{1}(\updelta_x) \leq D^2 < \infty$. Let $\sequenceks{\step_k}$ and $\sequenceks{\weight_k}$ given for all $k \in \nsets$ by  $\weight_k = \steps_k = \steps \in \ocint{0, L^{-1}}$. Let $\mu_0 \in \Pens_2(\rset^d)$. Then for any $N \in \nset, n \in \nsets$ we have
\[
\KLarg{\tnu^N_n}{\pi} \leq \left. \wassersteinTarg{\mu_0 \bQkers[N] \bSker_{\step}}{ \pi} \middle/ (2 n \steps) \right.
+  \steps \left(Ld + M_2^2 + D^2\right) \eqsp,
\]
Furthermore, let $\varepsilon > 0$ and
\begin{equation*}
    \steps_{\varepsilon}  \leq \min\defEns{\varepsilon/(2(Ld + M_2^2 + D^2)), L^{-1}}  \eqsp, \qquad    n_{\varepsilon}  \geq \lceil W_2^2(\mu_0 \tS^2_{\stepa[1]}, \pi) (\steps_{\varepsilon} \varepsilon)^{-1} \rceil \eqsp.
  \end{equation*}
Then we have $\KLarg{\tnu^0_{n_{\varepsilon}}}{\pi} \leq \varepsilon$.
\end{corollary}

{ In the case where a warm start is available for the Wasserstein distance,
  \ie~$W_2^2(\mu_0,\pi) \leq C$, for some absolute constant
  $C \geq 0$, then \Cref{coro:spgld_fixed_one} implies that the
  complexity of SPGLD to obtain a sample close from $\pi$ in KL with a
  precision target $\varepsilon>0$ is of order
  $(d+M^2_2+D^2)\bigO(\varepsilon^{-2})$. Therefore, this complexity bound
  depends on the dimension only trough $M_2$ and $D^2$ contrary to ULA. In addition,  Pinsker inequality implies
  that the complexity of SPGLD for the total variation distance is of
  order $(d+M^2_2+D^2)\bigO(\varepsilon^{-4})$.  }

{
In addition if we have access to $\eta >0$ and $M_{\eta}\geq 0$,
independent of the dimension, such that for all $x \in \rset^d$,
$x \not \in \boule{\xstar}{M_{\eta}}$,
$U(x) - U(\xstar) \geq \eta \norm{x-\xstar}$,
\Cref{propo:bound_wasser_init_condition} and \Cref{assum:stochastic_subgradient}-\ref{assum:stochastic_subgradient_i} imply that starting at $\updelta_{\xstar}$, the
overall complexity of SSGLD for the KL is in this case
$(\eta^{-2}d^{2}+M_\eta^2 + M^2)(d+M^2_2+D^2) \bigO(\varepsilon^{-2})$ and $(\eta^{-2} d^{2}+M_\eta^2+M^2)(d+M^2_2+D^2) \bigO(\varepsilon^{-4})$
for the total variation distance.
}

{If $(\gamma_k)_{k \in \nsets}$ and $(\lambda_k)_{k \in \nsets}$
  are given for all $k \in \nsets$ by
  $\gamma_k = \lambda_k = \gamma_1/k^{-\alpha}$,
  $\gamma_1 \in \ocint{0,L^{-1}}$. Then by the same reasoning as in
  the proof of \Cref{coro:ula_non_increas_sz}, we obtain that there
  exists $C \geq 0$ such that for all $n \in \nsets$, we have
  $\KLarg{\bnu_n^0}{\pi} \leq C \max(n^{\alpha - 1}, n^{- \alpha})$, if $\alpha \not = 1/2$,  and 
  for $\alpha = 1/2$, we have
  $\KLarg{\bnu_n^0}{\pi} \leq C (\ln(n) + 1) n^{-1/2}$.
}

If $\sup_{x \in \rset^d} \vargrad_1(\updelta_x) < \plusinfty$ does not hold, we can control the variance of stochastic gradient estimates using
\Cref{assum:cocoercitivity_sto_grad} again based on this following result.

\begin{proposition}
  \label{propo:bound_variance_sto_grad_spgld}
  Assume \Cref{assum:stochastic_gradient_proximal}  and  $\tilde{\Theta}_1$ satisfies \Cref{assum:cocoercitivity_sto_grad}. 
Then we have for all $x \in \rset^d$ and $\gamma \in \ocintLigne{0, \tilde{L}^{-1}}$
\begin{equation*}
2\gamma(\tilde{L}^{-1}-\gamma)\vargrad_{1}(\updelta_x) \leq  \norm[2]{x-\xstar}-\int_{\rset^d} \norm[2]{y-\xstar}(\tS^{1}_{\step}  T_{\step} \tS^{2}_{\step}) (x,\rmd y) +2 \gamma^2 \vargrad_1(\updelta_{\xstar}) + 2 \gamma d \eqsp,
\end{equation*}
where $\tS^{1}_{\step}, \tS^2_\step$ and  $\vargrad_1$ are defined by \eqref{eq:definition_sker_spgld}-\eqref{eq:definition_var_grad_sto_spgl} respectively. 
\end{proposition}
\begin{proof}
  Let $\gamma >0$, $x \in \rset^d$ and consider
  $ \tX_1 = \prox_{U_2}^{\step}\defEns{x-\step \tilde{\Theta}_1(x,\gradrv_1) +
    \sqrt{2\step} G_1}$, where $\gradrv_1$ and $G_1$ are two
  independent random variables, $\gradrv_1$ has distribution $\eta_1$
  and $G_1$ is a standard Gaussian random variable, so that $\tX_1$ has distribution $\tS^1_\step T_{\step} \tS^2_\step (x,\cdot)$.  First by
  \cite[Theorem 26.2(vii)]{bauschke:combettes:2011}, we have that
  $\xstar = \prox_{U_2}^{\step}(\xstar - \step \nabla U_1(\xstar))$
  and by \cite[Proposition 12.27]{bauschke:combettes:2011}, the
  proximal is non-expansive, for all $x,y \in \rset^d$,
  $\normLigne{\prox^{\step}_{U_2}(x) -\prox^{\step}_{U_2}(y)} \leq
  \norm{x-y}$. Using these two results and the fact that $\tilde{\Theta}_1$ satisfies
  \Cref{assum:cocoercitivity_sto_grad}, we have
\begin{align*}
 & \expe{\norm[2]{\tX_1 - \xstar}} = \expe{\norm{ \prox_{U_2}^{\step}\defEns{x-\step \tilde{\Theta}_1(x,\gradrv_1) +
    \sqrt{2\step} G_1} -\prox_{U_2}^{\step}\{\xstar - \step \nabla U_1(\xstar)\} }^2} \\
& \qquad \leq \expe{\norm{ \left(x-\step \tilde{\Theta}_1(x,\gradrv_1) +
    \sqrt{2\step} G_1\right) - \left( \xstar - \step \nabla U_1(\xstar)\right) }^2} \\
&  \qquad \leq  \norm[2]{x-\xstar}\\
& \qquad \qquad +  \expe{2\gamma \ps{x-\xstar }{\nabla U_1(\xstar) - \tilde{\Theta}_1(x,\gradrv_1)} + \gamma^2 \norm[2]{\nabla U_1(\xstar) - \tilde{\Theta}_1(x,\gradrv_1)}} + 2 \gamma d \\
& \qquad \leq \norm[2]{x-\xstar} - 2 \gamma(\tilde{L}^{-1}-\gamma) \expe{\norm[2]{ \tilde{\Theta}_1(x,\gradrv_1) -  \tilde{\Theta}_1(\xstar,\gradrv_1)}}\\
&\phantom{aaaaaaaaaaaaaaaaaaaaaaaaaaaa} +2 \gamma^2\expe{\norm[2]{ \tgradst_1(\xstar,\gradrv_1) - \nabla U_1(\xstar)}}+  2 \gamma d  \eqsp.
\end{align*}
The proof is completed upon noting that $\vargrad_1(\updelta_x) \leq \expeLigne{\norm[2]{\gradst_1(x,\gradrv_1) - \gradst_1(\xstar,\gradrv_1)}}$.
\end{proof}

Combining \Cref{thm:step_conv_sp} and \Cref{propo:bound_variance_sto_grad_spgld}, we get the following result. 
\begin{corollary}
\label{coro:fixed_step_conv_sp}
Assume \Cref{assum:stochastic_gradient_proximal}($m$) for $m \geq 0$ and that $\tilde{\Theta}_1$ satisfies \Cref{assum:cocoercitivity_sto_grad}. Let $\sequenceks{\step_k}$ and
$\sequenceks{\weight_k}$ be two non-increasing sequences of positive
real numbers given for all $k\in \nsets$ by  $\step_k=\weight_k = \gamma \in \ocintLigne{0, L^{-1}}$ ,  $\steps < \tilde{L}^{-1}$. Let $\mu_0 \in \Pens_2(\rset^d)$ and $N
\in \nset$.  Then for all $n \in \nsets$, it holds
\begin{multline*}
\KLarg{\tnu^N_n}{\pi} \leq \left.  \wassersteinTarg{\mu_0 \tQkers[N] \tSker_{\stepa[N+1]}^{2}}{ \pi} \middle/ (2\step n) \right.
+ \step (Ld+M_2^2) \\ + (1 + \steps L)(2(\tilde{L}^{-1}-\gamma))^{-1}\defEns{(2n)^{-1}\int_{\rset^d} \norm[2]{x-\xstar} \rmd \mu_0 \tQkers[N] \tS^2_{\steps} (y) +  \step^2 \vargrad_1(\updelta_{\xstar})+  \step d} \eqsp.
\end{multline*}

Furthermore, for $\varepsilon > 0$, consider step-size and a number of iterations satisfying:
{
\begin{align*}
    \steps_{\varepsilon}  \leq \min\parentheseDeux{\left. \varepsilon \middle/\defEns{4M_2^2 + 4Ld + 8\tilde{L}d \right. }, \sqrt{\varepsilon / \left( 8\tilde{L}\vargrad_1(\updelta_{\xstar}) \right)}, L^{-1}, (2\tilde{L})^{-1} }   \eqsp, \\   n_{\varepsilon}  \geq 2\max\defEns{ \ceil{ W_2^2(\mu_0\tSker_{\step_{\varepsilon}}^{2}, \pi) (\steps_{\varepsilon} \varepsilon)^{-1} }, \ceil{2\tilde{L}\varepsilon^{-1} \int_{\rset^d}  \norm[2]{x-\xstar} \rmd \mu_0\tS^2_{\steps}(y)}  } \eqsp.
  \end{align*}
}
Then, we have $\KLarg{\tnu^0_{n_{\varepsilon}}}{\pi} \leq \varepsilon$.
\end{corollary}
\begin{proof}
The proof of the corollary is a direct consequence of \Cref{thm:step_conv_sp} and \Cref{propo:bound_variance_sto_grad_spgld}, and is postponed to \Cref{sec:proof-cor:fixed_step_conv_sp}.
\end{proof}

Note that the dependency on the variance of the stochastic gradients is improved compared to the bound given by \Cref{coro:spgld_fixed_one}.  
We specify once again the result of \Cref{thm:step_conv_sp} for strongly convex potential.

\begin{theorem}
\label{thm:step_conv_sp_wasser}
Assume \Cref{assum:stochastic_gradient_proximal}$(m)$, for $m >
0$. 
Let $\sequenceks{\step_k}$ be a non-increasing sequences of positive
real numbers satisfying for all $k \in \nsets$, $\step_k \in \ocint{0, L^{-1}}$. Let
$\mu_0 \in \Pens_2(\rset^d)$.  Then for all
$n \in \nsets$, it holds
\begin{multline*}
W_2^2(\mu_0  \tQkers[n] \tSker_{\stepa[n+1]}^2, \pi) \leq \defEns{\prod_{k=1}^n(1-m\stepa[k+1])} W_2^2(\mu_0  \tSker_{\stepa[1]}^2, \pi)\\
+ \sum_{k=1}^{n}  \stepa[k+1]^2  \defEns{\prod_{i=k+2}^{n+1}(1-m\stepa[i])} \{2Ld + (1+\stepa[k+1] L) \vargrad_1(\mu_0 \tQkers[k-1] \tS^2_{\step_{k}}) + 2 M_2^2\} \eqsp.
\end{multline*}
\end{theorem}
\begin{proof}
The proof is postponed to \Cref{sec:proof-crefthm:st-2}.
\end{proof}

\begin{corollary}
Assume \Cref{assum:stochastic_gradient_proximal}$(m)$, for $m >
0$. Assume that $\sup_{x \in \rset^d} \vargrad_1(\updelta_x) \leq D^2 < \infty$.
 Let $\varepsilon > 0 $, $\mu_0 \in \Pens_2(\rset^d)$, and
\[
    \steps_{\varepsilon}  \leq \min\defEns{m\varepsilon/(4(Ld + D^2 + M_2^2)), L^{-1}} \eqsp, \qquad 
    n_{\varepsilon}  \geq \lceil \ln( 2W_2^2(\mu_0 \tS^2_{\step_{\vareps}}, \pi)/ (\varepsilon \steps_{\varepsilon} m)^{-1} \rceil\eqsp. 
\]
Then $W_2^2(\mu_0  \tRker^{n_{\varepsilon}}_{\gamma_{\varepsilon}, \gamma_{\varepsilon}} \tSker_{\step_{\varepsilon}}^2, \pi) \leq \varepsilon$, where $\tRker_{\gamma, \gamma}$ and $\tSker^2_{\gamma}$ are defined by \eqref{eq:definition_rker_spgld} and \eqref{eq:definition_sker_spgld} respectively. 
\end{corollary}
\begin{proof}
Since $\step_{\varepsilon} \leq L^{-1}$, we have $(1 + \step_{\varepsilon}L) \vargrad_1(\mu_0  \tRker^{k}_{\gamma_{\varepsilon}}\tS^2_{\step}) \leq 2D^2$ for all $k \geq 1$. Using \Cref{thm:step_conv_sp_wasser} then concludes the proof.
\end{proof}

Note that the bounds given by \Cref{thm:step_conv_sp_wasser} are
tighter the one given by \cite[Theorem 3]{dalalyan:karagulyan:2017}
which  shows under \Cref{assum:stochastic_gradient_proximal} with $U_2=0$ and $\sup_{x \in \rset^d} \vargrad_1(\updelta_x) \leq D^2$ that
\begin{equation*}
  W_2(\mu_0 \tRker_{\gamma,\gamma},\pi) \leq (1-m h) W_2(\mu_0,\pi) + 1.65(L/m)( \gamma d)^{1/2} + D^2(\gamma d)^{1/2} / (1.65 L + Dm) \eqsp.
\end{equation*}
Indeed, for constant step-size
$\gamma_k = \gamma \in \ocintLigne{0,L^{-1}}$ for all $k \in \nsets$,
\Cref{thm:step_conv_sp_wasser} implies with the same assumptions that
\begin{equation*}
  W_2(\mu_0 \tRker_{\gamma,\gamma},\pi) \leq (1-m h)^{1/2} W_2(\mu_0,\pi) + (2Ld \gamma/m)^{1/2} +  ((1+\gamma)\gamma/m)^{1/2} D \eqsp.
\end{equation*}
As for ULA, the dependency on the condition number $L/m$ is improved.

In the strongly convex case, we can improve the dependency on the variance of the stochastic gradient under the following condition.
\begin{assumption}
  \label{ass:cocoer_two_st_convex}
\begin{sf}
  There exist $\tilde{L}_1,\tilde{m}_1>0$ such that for all for $\eta$-almost every $z \in \mathsf{Z}$, for all $x,y \in \rset^d$, we have
  \begin{equation*}
    \ps{\tgradst_1(x,z) - \tgradst_1(y,z)}{x-y} \geq \tilde{m}_1 \norm[2]{x-y} + (1/\tL_1) \norm[2]{\tgradst_1(x,z) - \tgradst_1(y,z)}\eqsp.
  \end{equation*}
\end{sf}
\end{assumption}
The condition \Cref{ass:cocoer_two_st_convex} is for example satisfied
if $\eta$-almost surely, $x \mapsto \tgradst_1(x,z)$ is strongly
convex, see \cite[Theorem 2.1.12]{nesterov:2004}.

\begin{proposition}
  \label{propo:coco_st_convex}
  Assume \Cref{assum:stochastic_gradient_proximal}$(m)$ for $m >0$ and \Cref{ass:cocoer_two_st_convex}. Then for all $\step >0$ we have
\begin{multline*}
2\gamma(\tilde{L}^{-1}_1-\gamma)\vargrad_{1}(\updelta_x) \leq  (1-\tm_1\gamma)\norm[2]{x-\xstar}-\int_{\rset^d} \norm[2]{y-\xstar}(\tS^{1}_{\step}  T_{\step} \tS^{2}_{\step}) (x,\rmd y) +2 \gamma^2 \vargrad_1(\updelta_{\xstar}) + 2 \gamma d \eqsp,
\end{multline*}
  where $\tS^{1}_{\step}, \tS^2_\step$ and  $\vargrad_1$ are defined by \eqref{eq:definition_sker_spgld}-\eqref{eq:definition_var_grad_sto_spgl} respectively. 
\end{proposition}

\begin{proof}
  The proof is similar to the proof of \Cref{propo:bound_variance_sto_grad_spgld}. It is postponed to \Cref{sec:proof-crefpr}.
\end{proof}
\begin{corollary}
  \label{cor:spgula_str_conv_wass}
  Assume \Cref{assum:stochastic_gradient_proximal}$(m)$, for $m > 0$
  and that $\tilde{\Theta}_1$ satisfies
  \Cref{ass:cocoer_two_st_convex}.  Let $(\gamma_k)_{k \in \nsets}$
  defined for all $k \in \nsets$ by
  $\gamma_k = \step \in \ocintLigne{0,L^{-1} \wedge
    (2\tilde{L}_1)^{-1}}$. Let $\mu_0 \in \Pens_2(\rset^d)$.
  Define $\tm = \min(m,\tm_1)$ and
\begin{equation}\label{eq:coro_delta_1}
  \begin{aligned}
  \Delta_1 & = 2(Ld+M_2)/m + \{2 \tilde{L}_1(1+\gamma L) /\tm\}  d  \\
  \Delta_2 &= \{2 \tilde{L}_1(1+\gamma L) /\tm\} \vargrad_1(\updelta_{\xstar})\\
\Delta_3 & = \gamma \tilde{L}_1(1+\gamma L)  \defEns{\int_{\rset^d}\norm[2]{x-\xstar} \rmd \mu_0  \tS^2_{\step_{\vareps}}(x)} \eqsp.
\end{aligned}
\end{equation}
Then for all
$n \in \nsets$, it holds
\begin{multline}
  \label{cor:spgula_str_conv_wass_eq_1}
W_2^2(\mu_0  \tRker_{\gamma,\gamma}^n \tSker_{\step}^2, \pi) \leq  (1-m\step)^n W_2^2(\mu_0  \tSker_{\step}^2, \pi)+ (1-\tm\gamma)^{n} \Delta_3+\steps \Delta_1 + \steps^2 \Delta_2     \eqsp,
\end{multline}
where $\tRker_{\gamma, \gamma}$ and $\tSker^2_{\gamma}$ are defined by \eqref{eq:definition_rker_spgld} and \eqref{eq:definition_sker_spgld}.

Therefore, for $\varepsilon > 0 $ and 
\begin{align*}
    \steps_{\varepsilon} & \leq \min\defEns{\varepsilon/(4\Delta_1), [\varepsilon /(4\Delta_2)]^{1/2}, L^{-1}, (2\tilde{L}_1)^{-1}} \eqsp \\ 
    n_{\varepsilon}  & \geq  \max\defEns{\lceil \ln( 4W_2^2(\mu_0 \tS^2_{\step_{\varepsilon}}, \pi)/ \varepsilon) (\steps_{\varepsilon} m)^{-1} \rceil , \lceil \ln(4 \Delta_3 / \varepsilon) (\steps_{\varepsilon}\tm)^{-1} \rceil}\eqsp,
\end{align*}
it holds $
W_2^2(\mu_0 \tRker_{\gamma_{\vareps},\gamma_{\vareps}}^{n_{\varepsilon}} \tSker_{\step_{\varepsilon}}^2, \pi) \leq \varepsilon$.
\end{corollary}
\begin{proof}
 The proof of the corollary is postponed to Section \ref{sec:proof-cor:spgld}.
\end{proof}

\begin{corollary}
    Assume \Cref{assum:stochastic_gradient_proximal}$(m)$, for $m > 0$
  and that $\tilde{\Theta}_1$ satisfies
  \Cref{ass:cocoer_two_st_convex}.  Define $\tm = \min(m,\tm_1)$.
  Let $\varepsilon >0$, $\mu_0 \in \Pens_2(\rset^d)$ and
  \begin{align*}
       \steps_{\varepsilon} & \leq \min\defEns{\varepsilon/(4\Delta_1), [\varepsilon /(4\Delta_2)]^{1/2}, L^{-1}, (2\tilde{L}_1)^{-1}} \eqsp, \\ 
    N_{\varepsilon}  & \geq \max\defEns{\lceil \ln( 4W_2^2(\mu_0 \tS^2_{\step_{\varepsilon}}, \pi)/ \varepsilon) (\steps_{\varepsilon} m)^{-1} \rceil , \lceil \ln(4 \Delta_3 / \varepsilon) (\steps_{\varepsilon}\tm)^{-1} \rceil} \\
    \tstep_{\varepsilon}&  \leq \min\parentheseDeux{\left. \varepsilon \middle/\defEns{4M_2^2 + 4Ld + 8\tilde{L}d \right. }, \sqrt{\varepsilon / \left( 8\tilde{L}\vargrad_1(\updelta_{\xstar}) \right)}, L^{-1}, (2\tilde{L})^{-1} }   \eqsp, \\
    n_{\varepsilon} & \geq 2\max\defEns{ \ceil{\steps_{\varepsilon} ^{-1} }, \ceil{2\tilde{L}\varepsilon^{-1} \int_{\rset^d}  \norm[2]{x-\xstar} \rmd \mu_0\tRker_{\steps_{\varepsilon},\steps_{\varepsilon}}^{N_{\varepsilon}}\tS^2_{\steps}(y) } } \eqsp,
  \end{align*}
  where $\Delta_1,\Delta_2,\Delta_3$ are defined in \eqref{eq:coro_delta_1} and $\tRker_{\gamma, \gamma}$ and $\tSker^2_{\gamma}$ are defined by \eqref{eq:definition_rker_spgld} and \eqref{eq:definition_sker_spgld}.
  Let $(\step_k)_{k \in \nset}$ defined by  $\steps_k = \steps_{\varepsilon}$ for $ k \in \{1,\ldots,  N_{\varepsilon}\}$ and $ \steps_k = \tilde{\steps_{\varepsilon}}$ for $k > N_{\varepsilon}$.  Then we have $ \KLarg{\tnu^{N_{\varepsilon}}_{n_{\varepsilon}}}{\pi} \leq \varepsilon$
where $\tnu_{n_{\varepsilon}}^{N_{\varepsilon}} = n_{\varepsilon}^{-1} \sum_{k=1}^{n_{\varepsilon}} \mu_0 \tRker_{\steps_{\varepsilon},\steps_{\varepsilon}}^{N_{\varepsilon}}  \tRker_{\tilde{\steps}_{\varepsilon},\tilde{\steps}_{\varepsilon}}^k$.

\end{corollary}
\begin{proof}
 \Cref{cor:spgula_str_conv_wass} implies that after the burn in phase of $N_{\varepsilon}$ steps with step-size $\step_{\varepsilon}$, we  have $W_2^2(\mu_0  \tQkers[N_{\varepsilon}] \tSker_{\step_{\varepsilon}}^2, \pi) \leq \varepsilon$. Then, since we can treat $\mu_0  \tQkers[N_{\varepsilon}]$ as a new starting measure, \Cref{coro:fixed_step_conv_sp} concludes the proof.
\end{proof}


\section{Numerical experiments}

\label{sec:numer-exper}
In this section, we experiment SPGLD and SSGLD on a Bayesian logistic
regression problem, see \eg~\cite{holmes:held:2006},
\cite{gramacy:polson:2012} and \cite{park:hastie:2007}. Consider
\iid~observations $(X_i,Y_i)_{i \in \{1,\ldots,N\}}$, where
$(Y_i)_{i \in \{1,\ldots,N\}}$ are binary response variables and
$(X_i)_{i \in \{1,\ldots,N\}}$ are $d$-dimensional covariance
variables. For all $i \in \{1,\ldots,N\}$, $Y_i$ is assumed to be a
Bernoulli random variable with parameter
$\Phi(\beta^{\transpose} X_i)$ where $\beta$ is the parameter of
interest and for all $u\in \rset$, $\Phi(u) =
\rme^{u}/(1+\rme^{u})$. We choose as prior distributions (see \cite{genkin2007large} and \cite{li:lin:2010}) a
$d$-dimensional Laplace distribution and a combination of the Laplace
distribution and the Gaussian distribution, with density with respect
to the Lebesgue measure given respectively for all $\beta \in \rset^d$
by
  \begin{equation*}
    \mathrm{p}_{1}(\beta) \propto \exp\parenthese{-a_1\sum_{i=1}^d \absLigne{\beta_i} } \eqsp, \, \qquad   \mathrm{p}_{1,2}(\beta) \propto \exp\parenthese{-a_1\sum_{i=1}^d \absLigne{\beta_i} -a_2 \sum_{i=1}^d \beta_i^2} \eqsp,
  \end{equation*}
  where $a_1$ is set to $1$ in the case of $\mathrm{p}_{1}$ and
  $a_1=0.9,\eqsp a_2=0.1$ in the case of $\mathrm{p}_{1,2}$. We obtain
  then two different a posteriori distributions $\mathrm{p}_1(\cdot |(X,Y)_{i \in \{1,\ldots,N\}})$ and $\mathrm{p}_{1,2}(\cdot |(X,Y)_{i \in \{1,\ldots,N\}})$ with potentials given, respectively, by 
\[
\beta \mapsto \sum_{n=1}^N\ell_n(\beta) + a_1\sum_{i=1}^d \absLigne{\beta_i} \eqsp, \qquad \beta \mapsto \sum_{n=1}^N\ell_n(\beta) +a_2 \sum_{i=1}^d \beta_i^2 +a_1\sum_{i=1}^d \absLigne{\beta_i}\eqsp.
\]
 where 
\[
 \ell_n(\beta) = -Y_n\beta^{\transpose} X_n+\log[1+\exp(\beta^{\transpose} X_n)]\eqsp.
\]

 \begin{figure}[!h]
\begin{subfigure}{0.32\textwidth}
\includegraphics[width=0.9\linewidth, height=4cm]{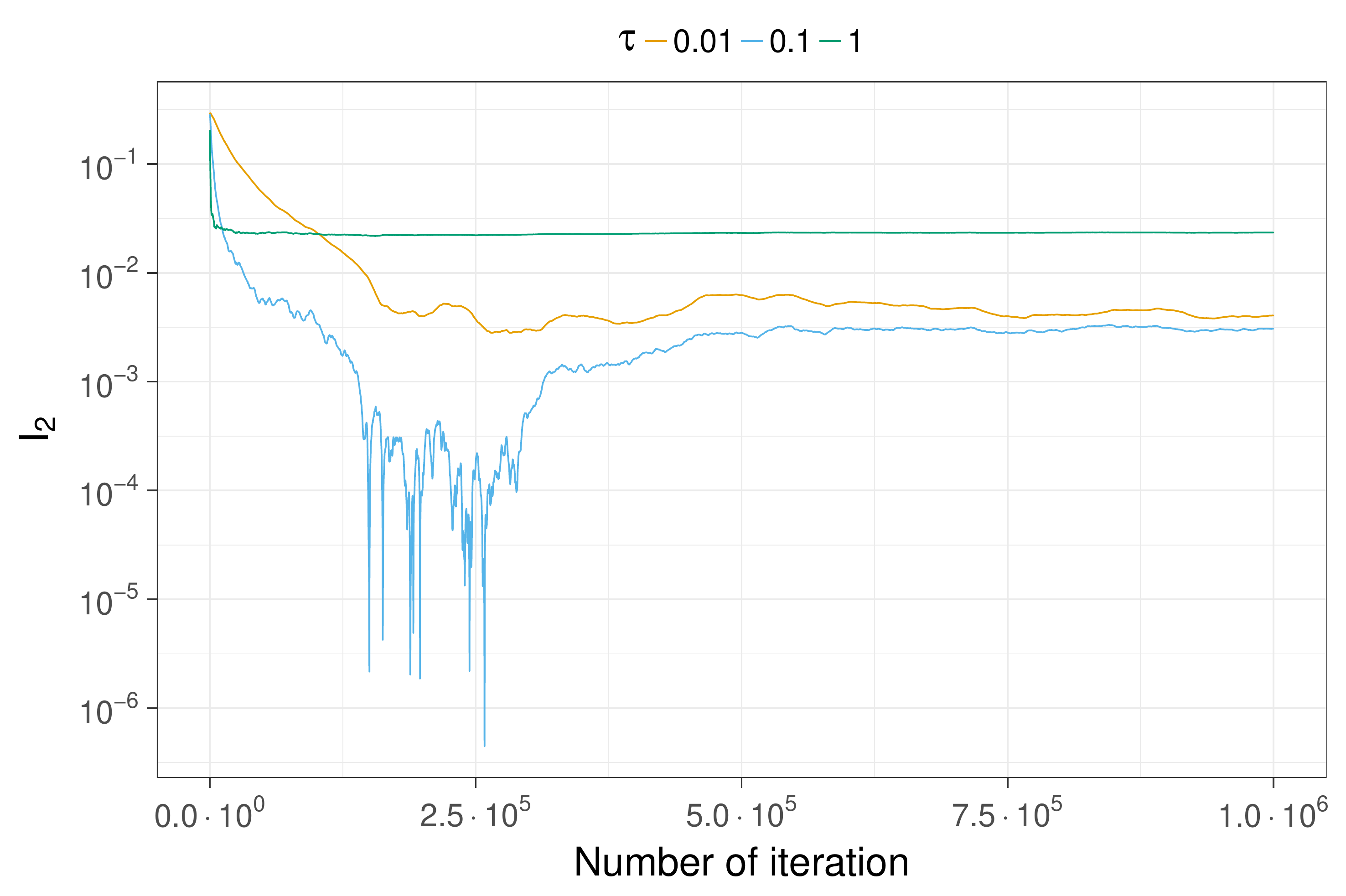}
\caption{ }
\end{subfigure}
\begin{subfigure}{0.32\textwidth}
\includegraphics[width=0.9\linewidth, height=4cm]{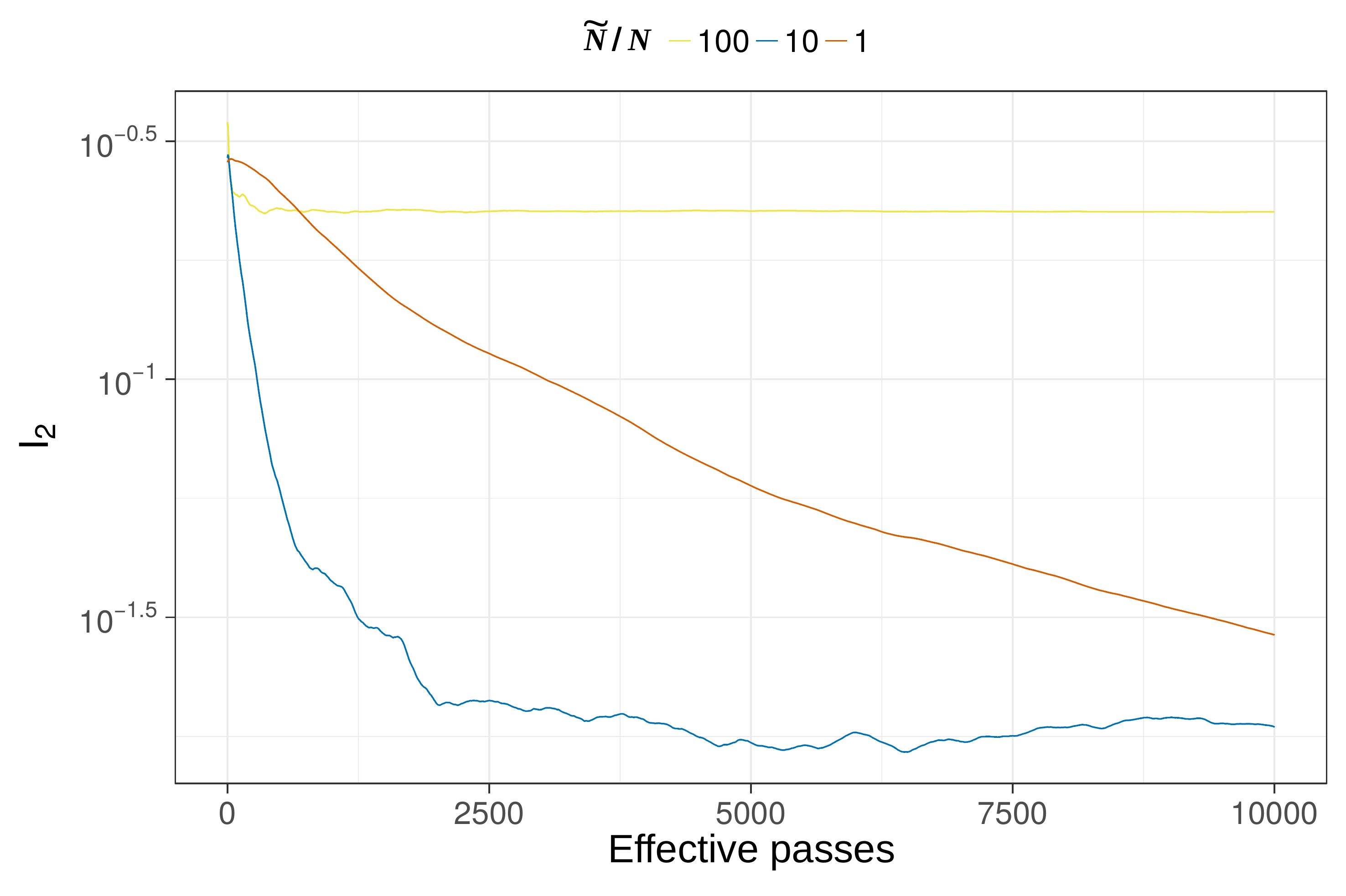}
\caption{ }
\end{subfigure}
\begin{subfigure}{0.32\textwidth}
\includegraphics[width=0.9\linewidth, height =4cm]{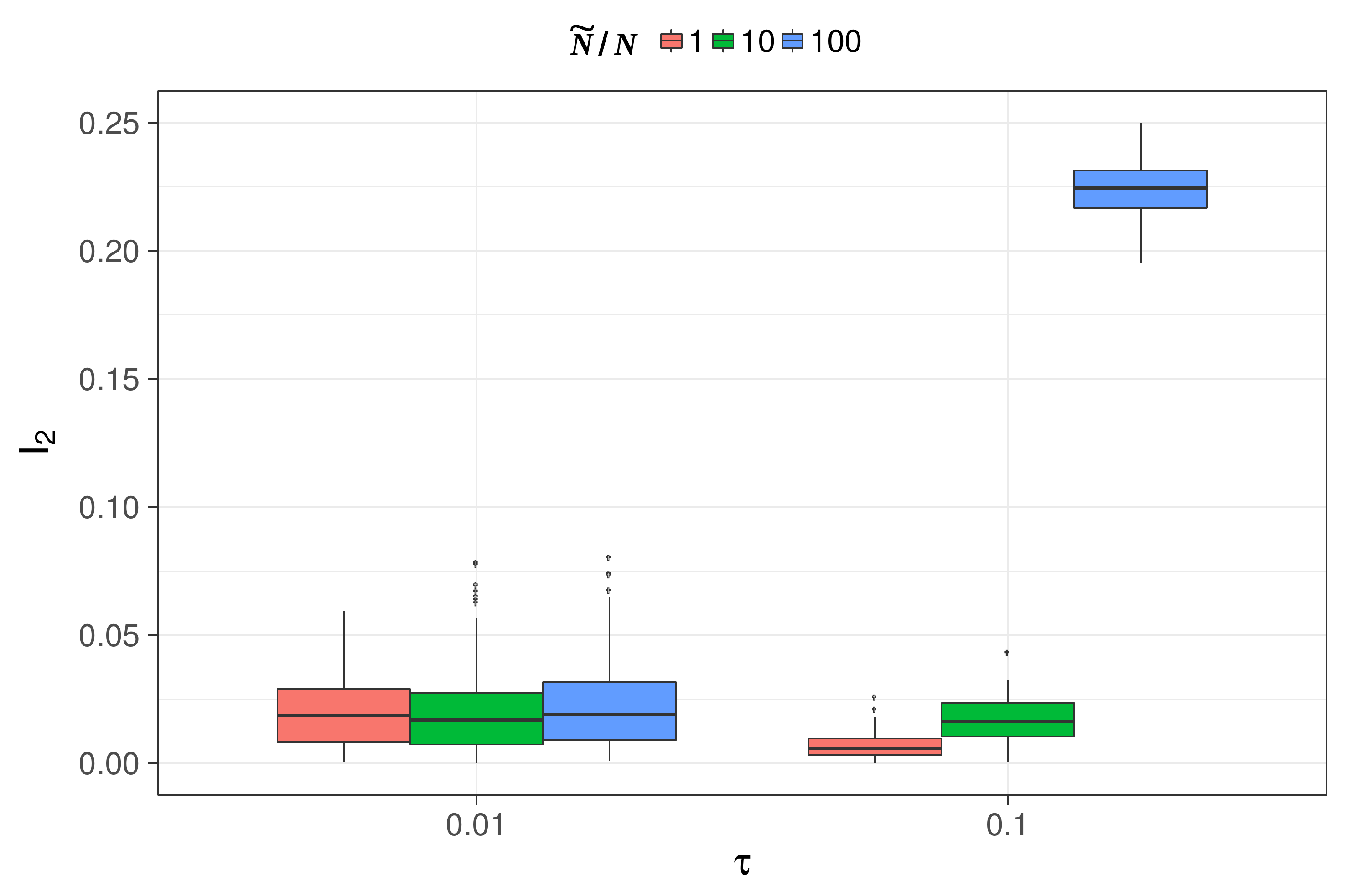}
\caption{ }
\end{subfigure}
 
\begin{subfigure}{0.32\textwidth}
\includegraphics[width=0.9\linewidth, height=4cm]{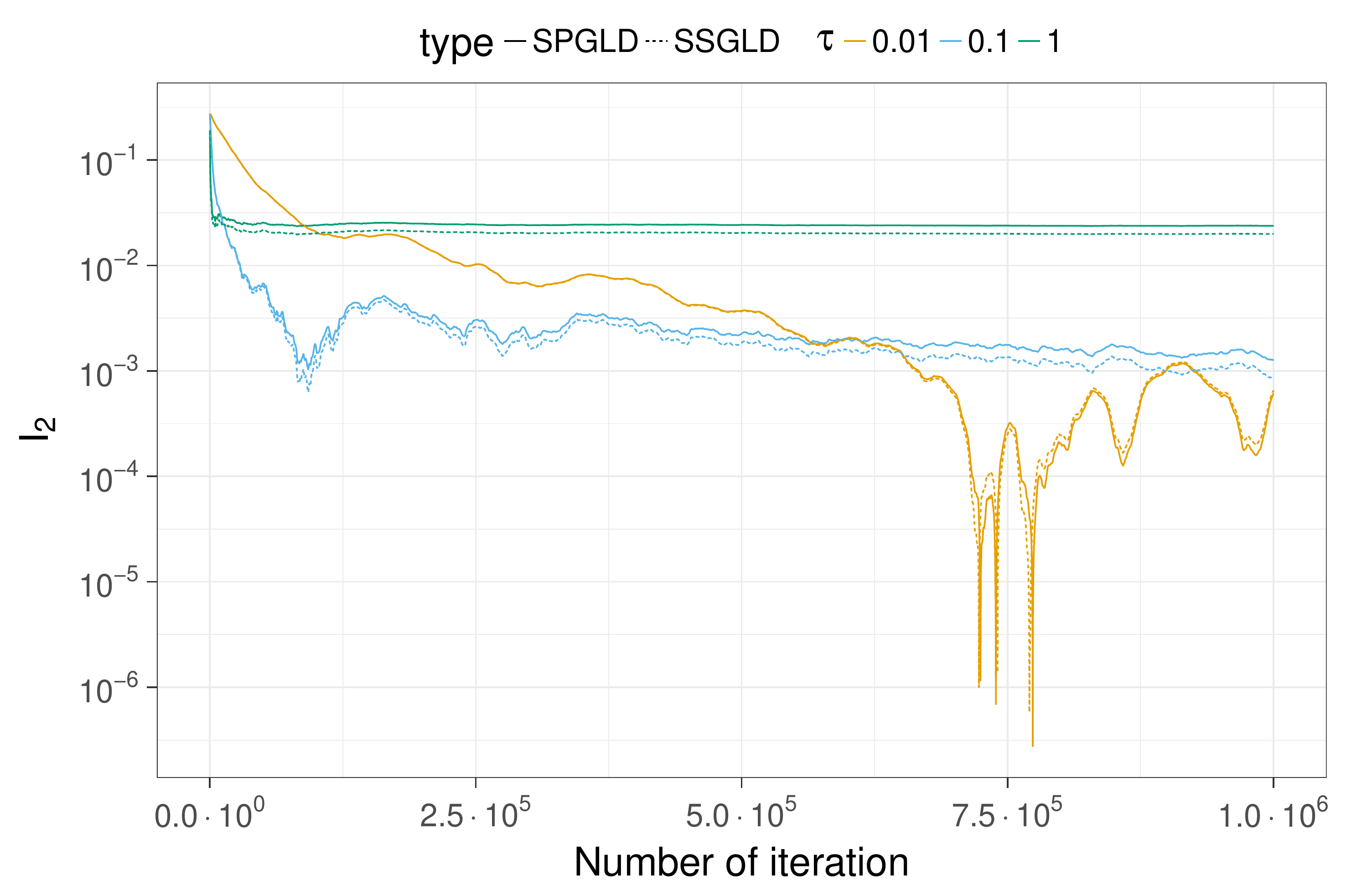}
\caption{ }
\end{subfigure}
\begin{subfigure}{0.32\textwidth}
\includegraphics[width=0.9\linewidth, height=4cm]{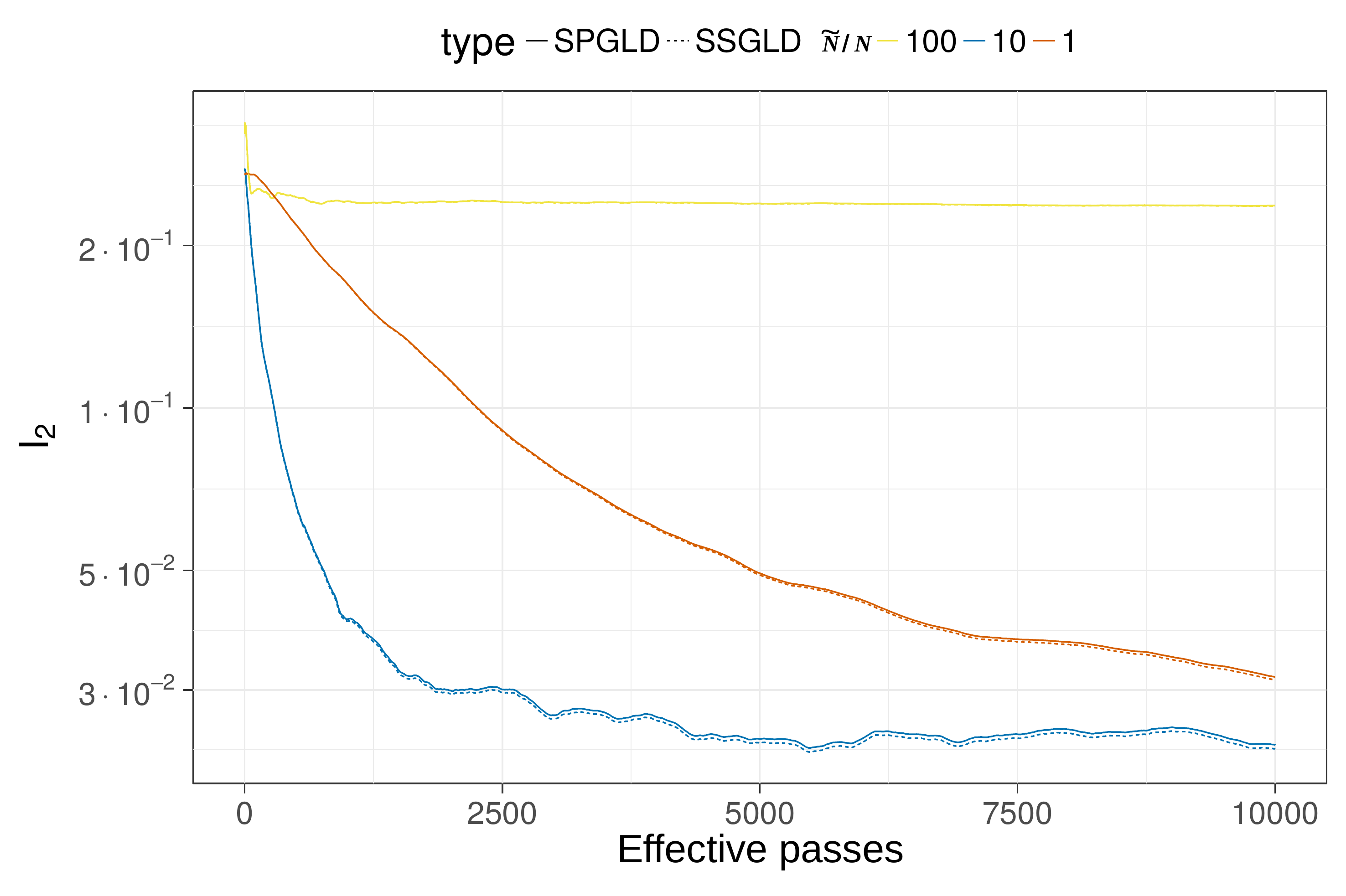}
\caption{ }
\end{subfigure}
\begin{subfigure}{0.32\textwidth}
\includegraphics[width=0.9\linewidth, height =4cm]{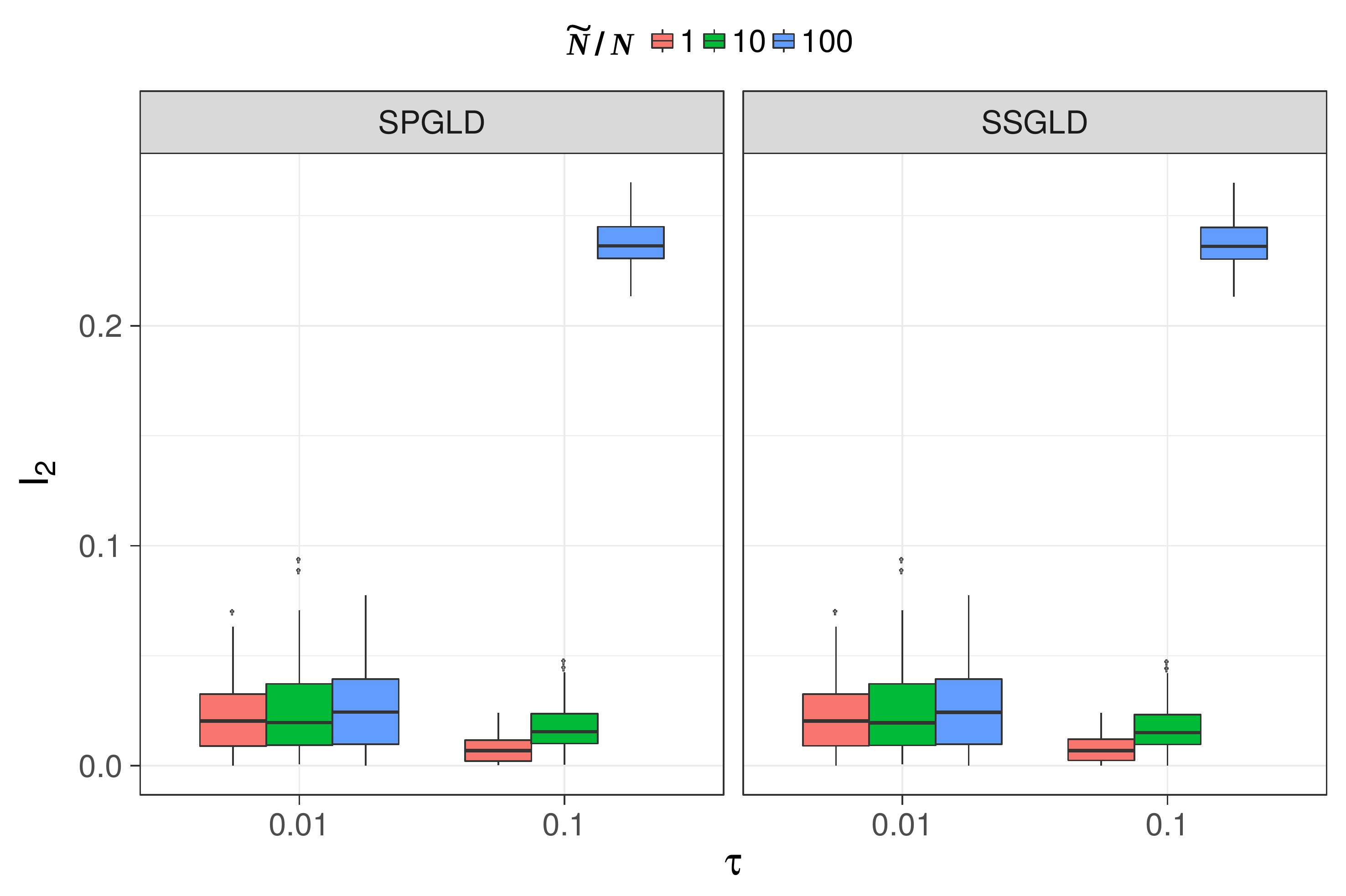}
\caption{}
\end{subfigure}

 \caption{Mean absolute error of estimator of $I_2$ for Australian Credit Approval dataset: (top row) results for $\mathrm{p}_{1,2}(\cdot |(X,Y)_{i \in \{1,\ldots,N\}})$; 
 (a) convergence of SPGLD for $\tilde{N}=1$ , 
(b) convergence of SPGLD in terms of effective passes for $\tau=0.1$, (c) boxplot of SPGLD for full runs;
(bottom row) results for $\mathrm{p}_1(\cdot |(X,Y)_{i \in \{1,\ldots,N\}})$; (d) convergence of SPGLD and SSGLD for $\tilde{N}=N$ , 
(e) convergence of SPGLD and SSGLD in terms of effective passes for $\tau=0.1$, (f) boxplot of SPGLD and SSGLD for full run}
\label{fig:austra}
\end{figure}

 \begin{figure}[!h]
\begin{subfigure}{0.32\textwidth}
\includegraphics[width=0.9\linewidth, height=4cm]{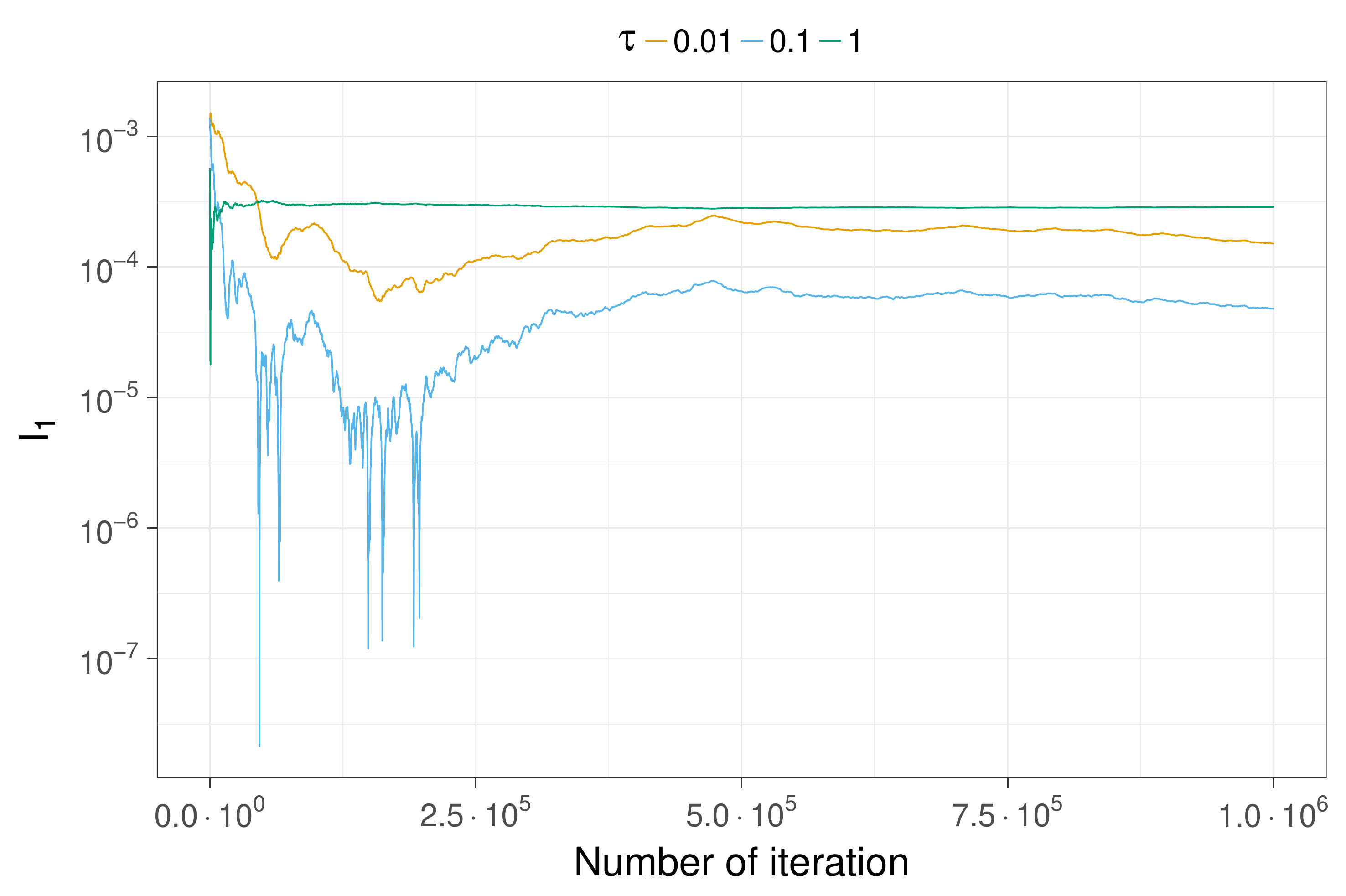}
\caption{ }
\end{subfigure}
\begin{subfigure}{0.32\textwidth}
\includegraphics[width=0.9\linewidth, height=4cm]{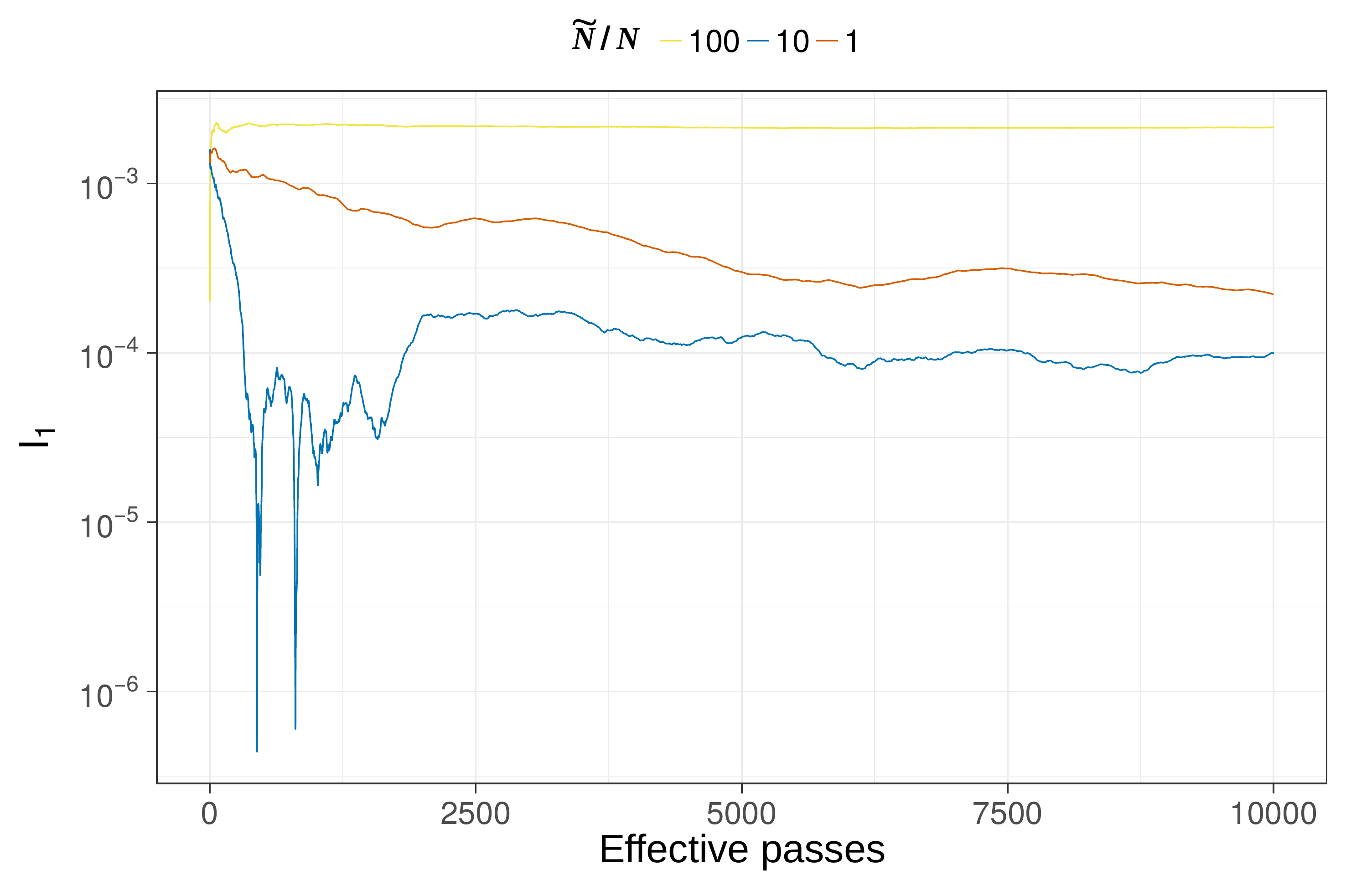}
\caption{ }
\end{subfigure}
\begin{subfigure}{0.32\textwidth}
\includegraphics[width=0.9\linewidth, height =4cm]{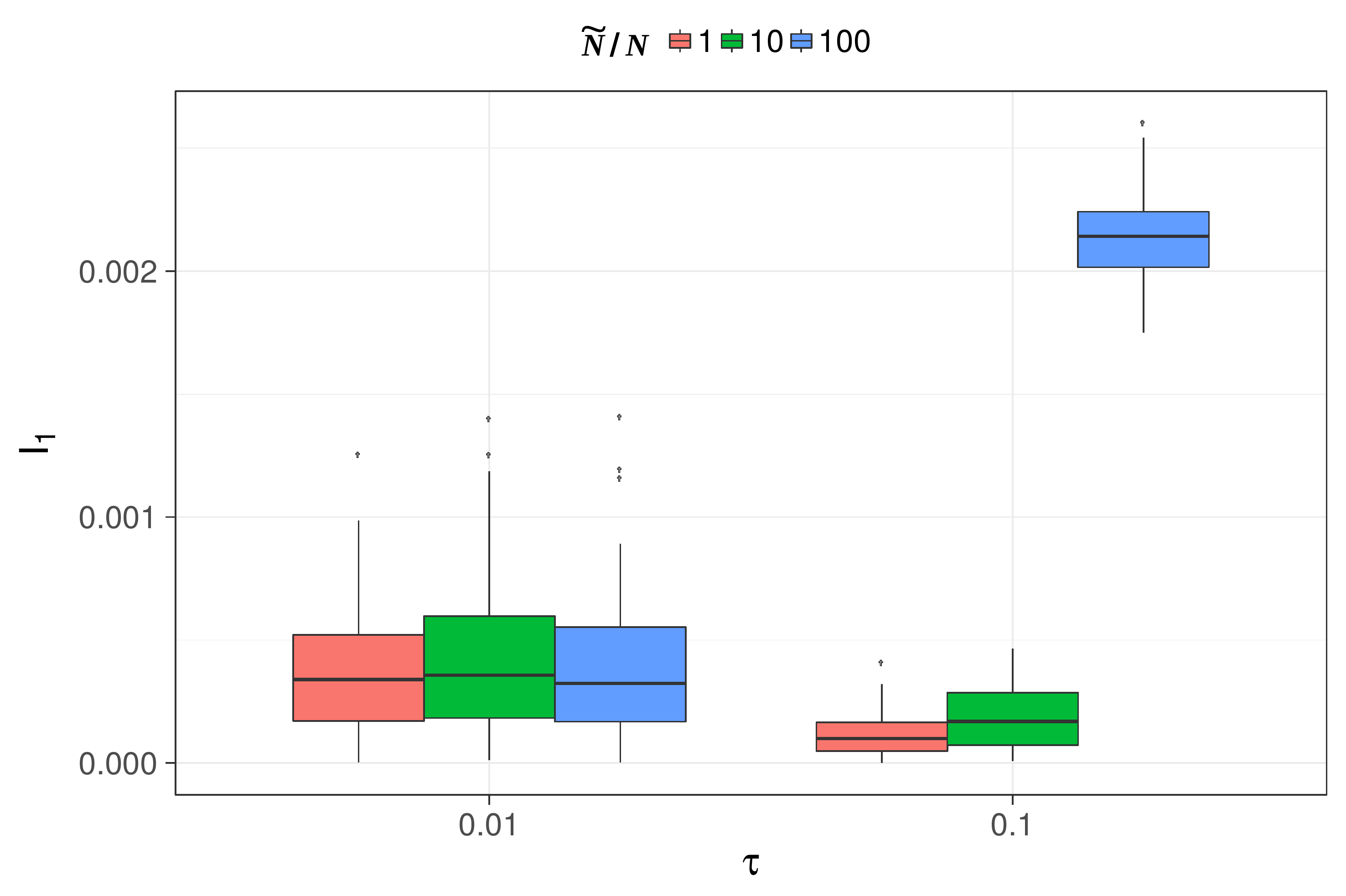}
\caption{ }
\end{subfigure}
 
\begin{subfigure}{0.32\textwidth}
\includegraphics[width=0.9\linewidth, height=4cm]{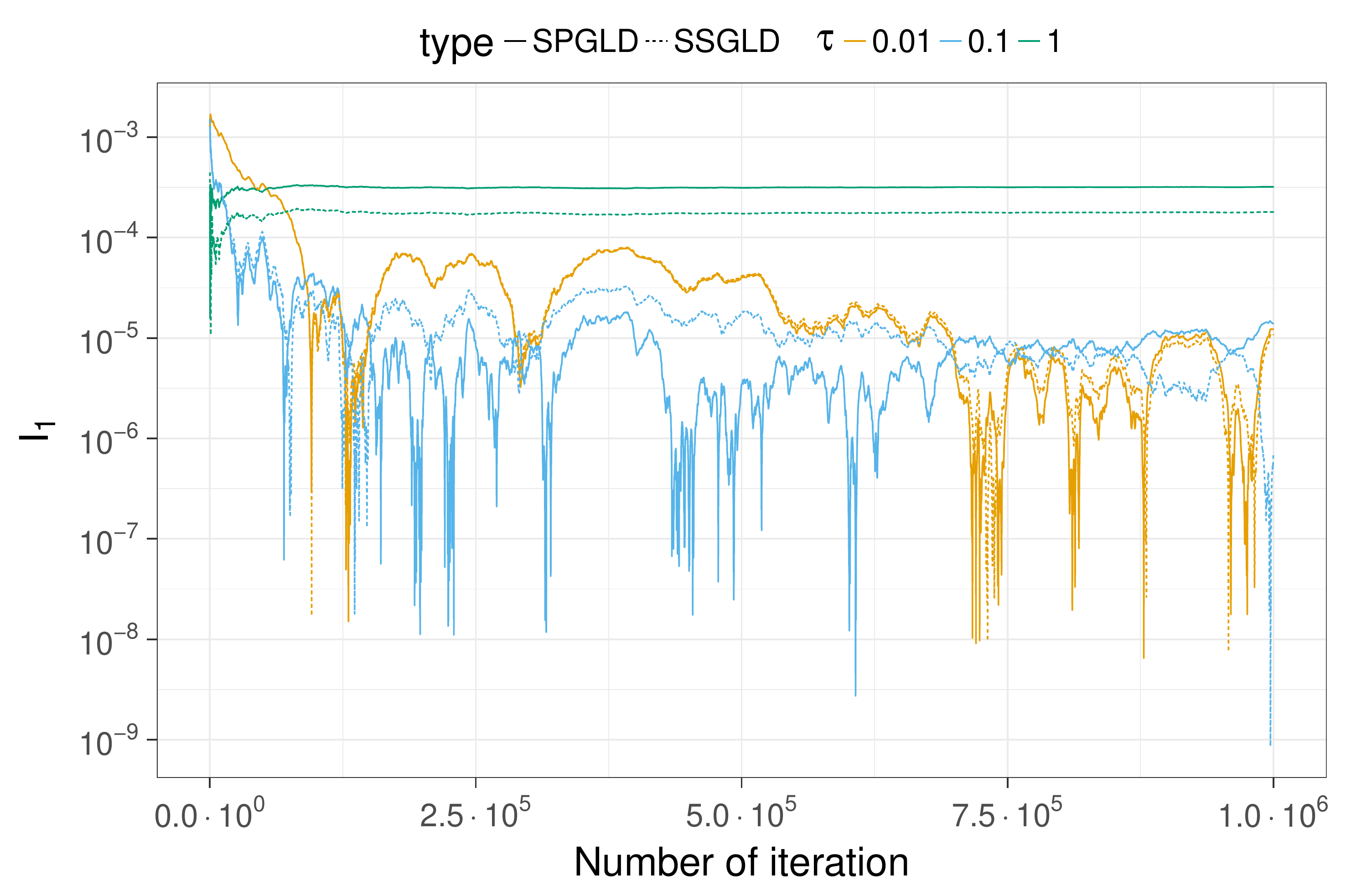}
\caption{ }
\end{subfigure}
\begin{subfigure}{0.32\textwidth}
\includegraphics[width=0.9\linewidth, height=4cm]{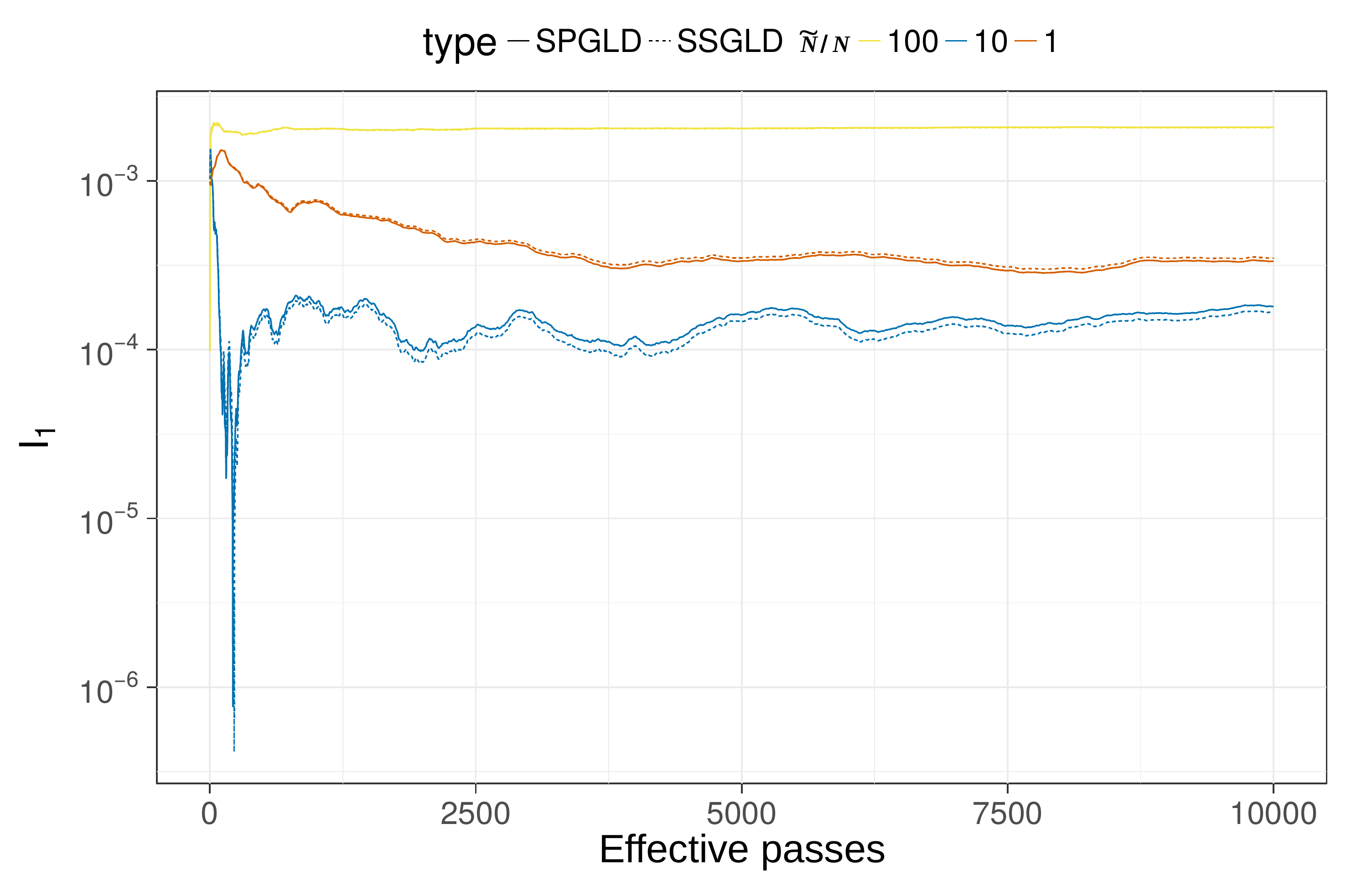}
\caption{ }
\end{subfigure}
\begin{subfigure}{0.32\textwidth}
\includegraphics[width=0.9\linewidth, height =4cm]{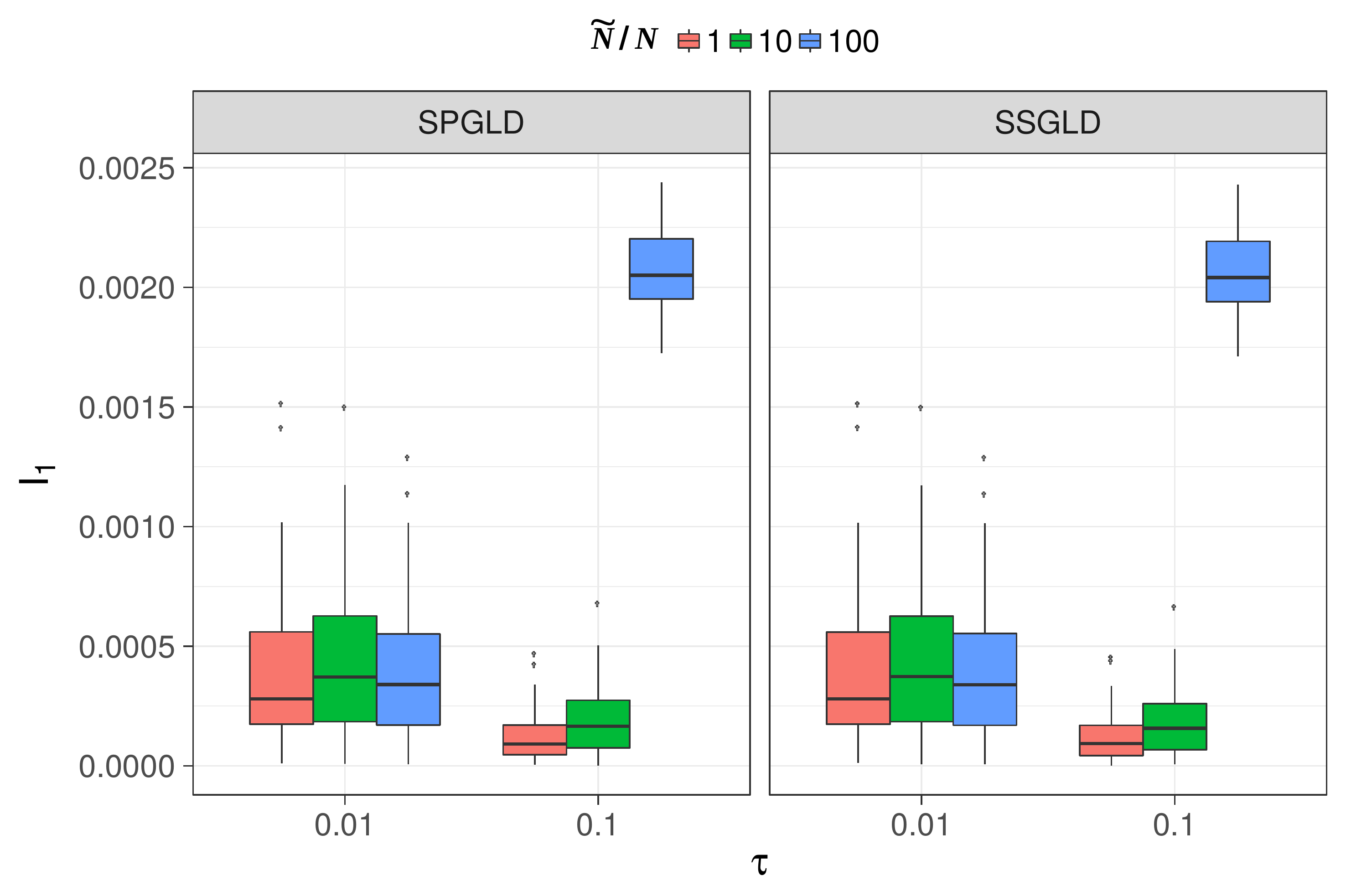}
\caption{}
\end{subfigure}

 \caption{Mean absolute error of estimator of $I_1$ for Australian Credit Approval dataset: (top row) results for $\mathrm{p}_{1,2}(\cdot |(X,Y)_{i \in \{1,\ldots,N\}})$; 
 (a) convergence of SPGLD for $\tilde{N}=1$ , 
(b) convergence of SPGLD in terms of effective passes for $\tau=0.1$, (c) boxplot of SPGLD for full runs;
(bottom row) results for $\mathrm{p}_1(\cdot |(X,Y)_{i \in \{1,\ldots,N\}})$; (d) convergence of SPGLD and SSGLD for $\tilde{N}=N$ , 
(e) convergence of SPGLD and SSGLD in terms of effective passes for $\tau=0.1$, (f) boxplot of SPGLD and SSGLD for full run}
\label{fig:austra_theta}
\end{figure}

 \begin{figure}[!h]
\begin{subfigure}{0.32\textwidth}
\includegraphics[width=0.9\linewidth, height=4cm]{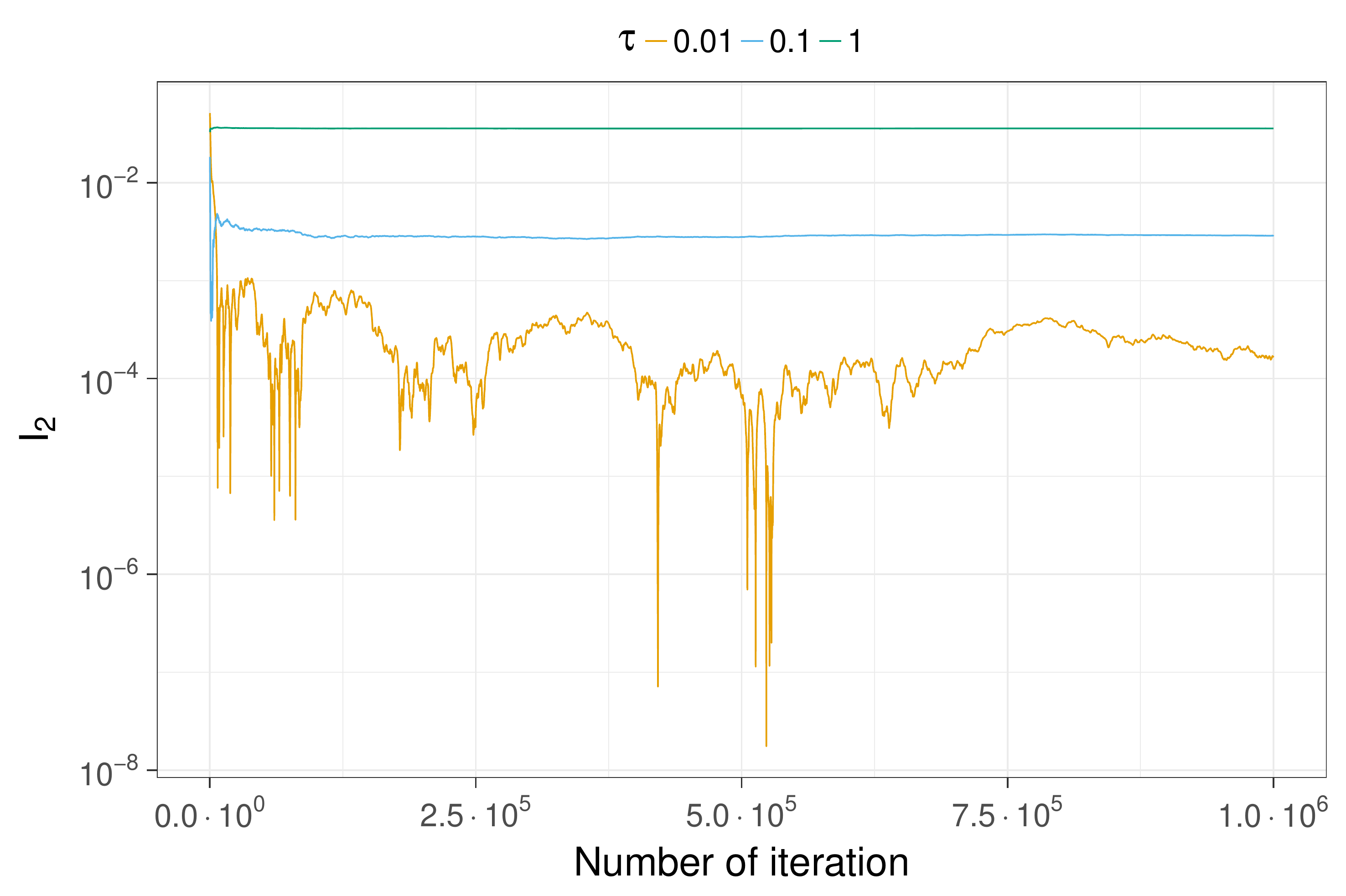}
\caption{ }
\end{subfigure}
\begin{subfigure}{0.32\textwidth}
\includegraphics[width=0.9\linewidth, height=4cm]{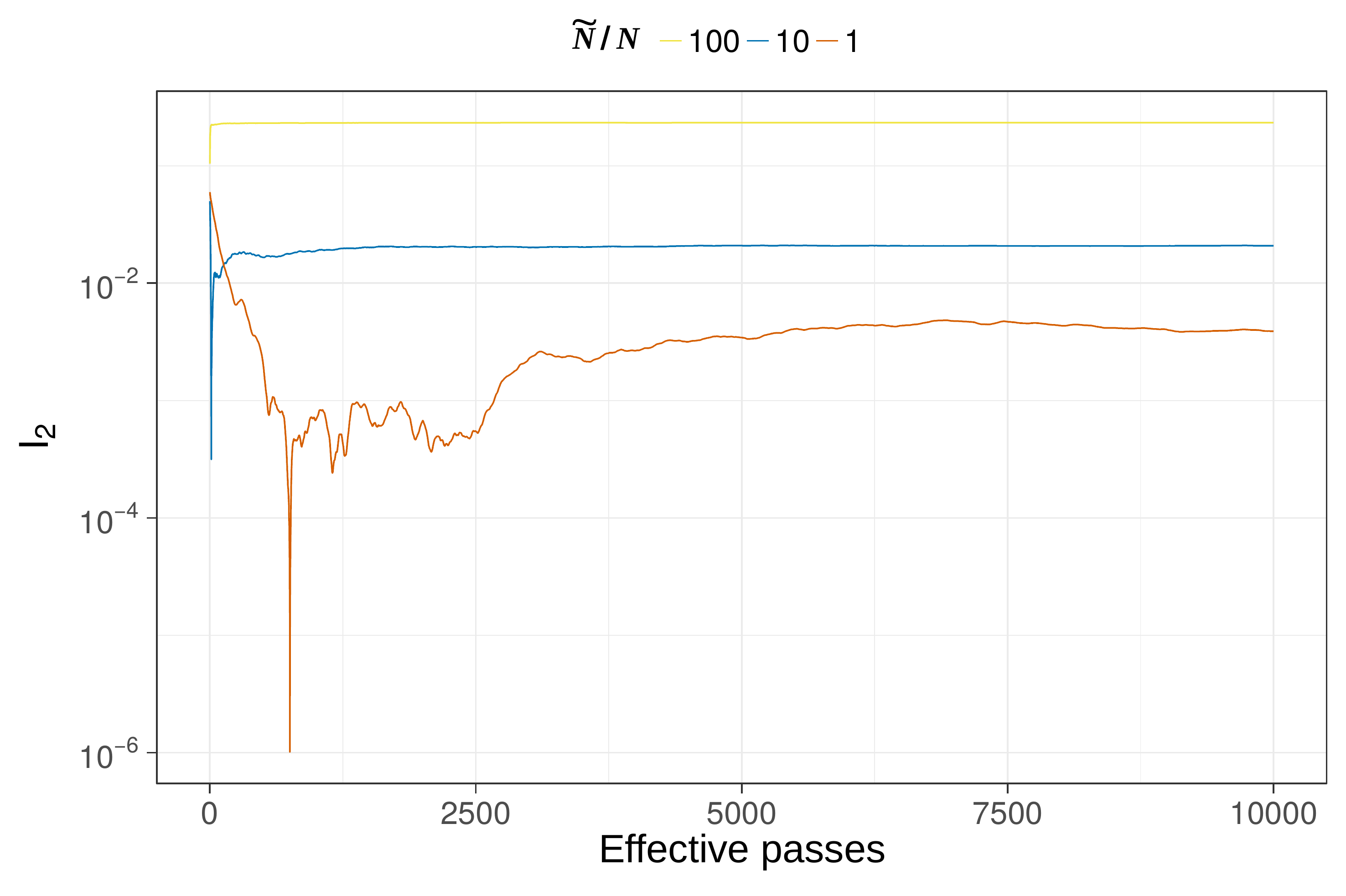}
\caption{ }
\end{subfigure}
\begin{subfigure}{0.32\textwidth}
\includegraphics[width=0.9\linewidth, height =4cm]{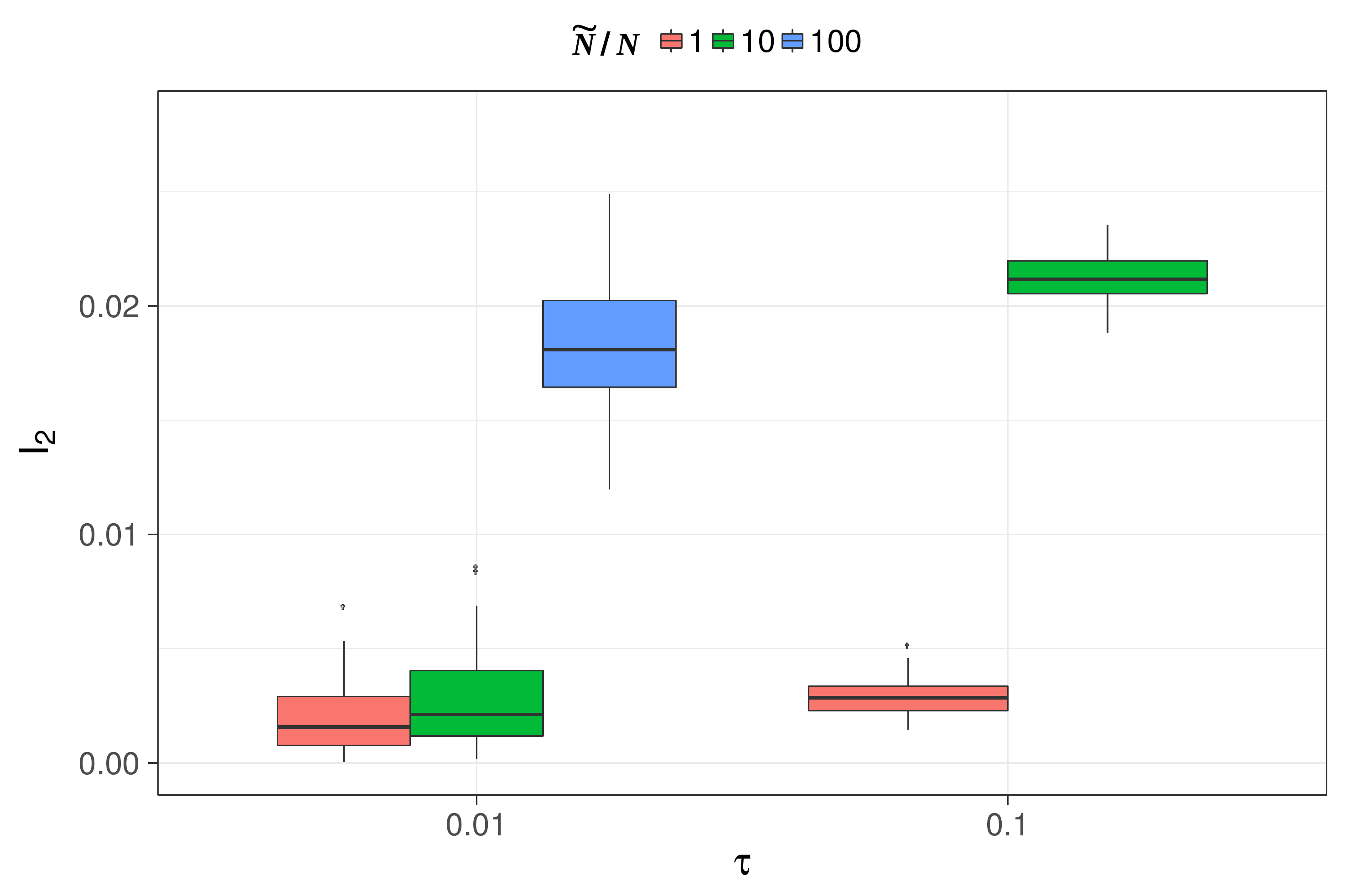}
\caption{ }
\end{subfigure}
 
\begin{subfigure}{0.32\textwidth}
\includegraphics[width=0.9\linewidth, height=4cm]{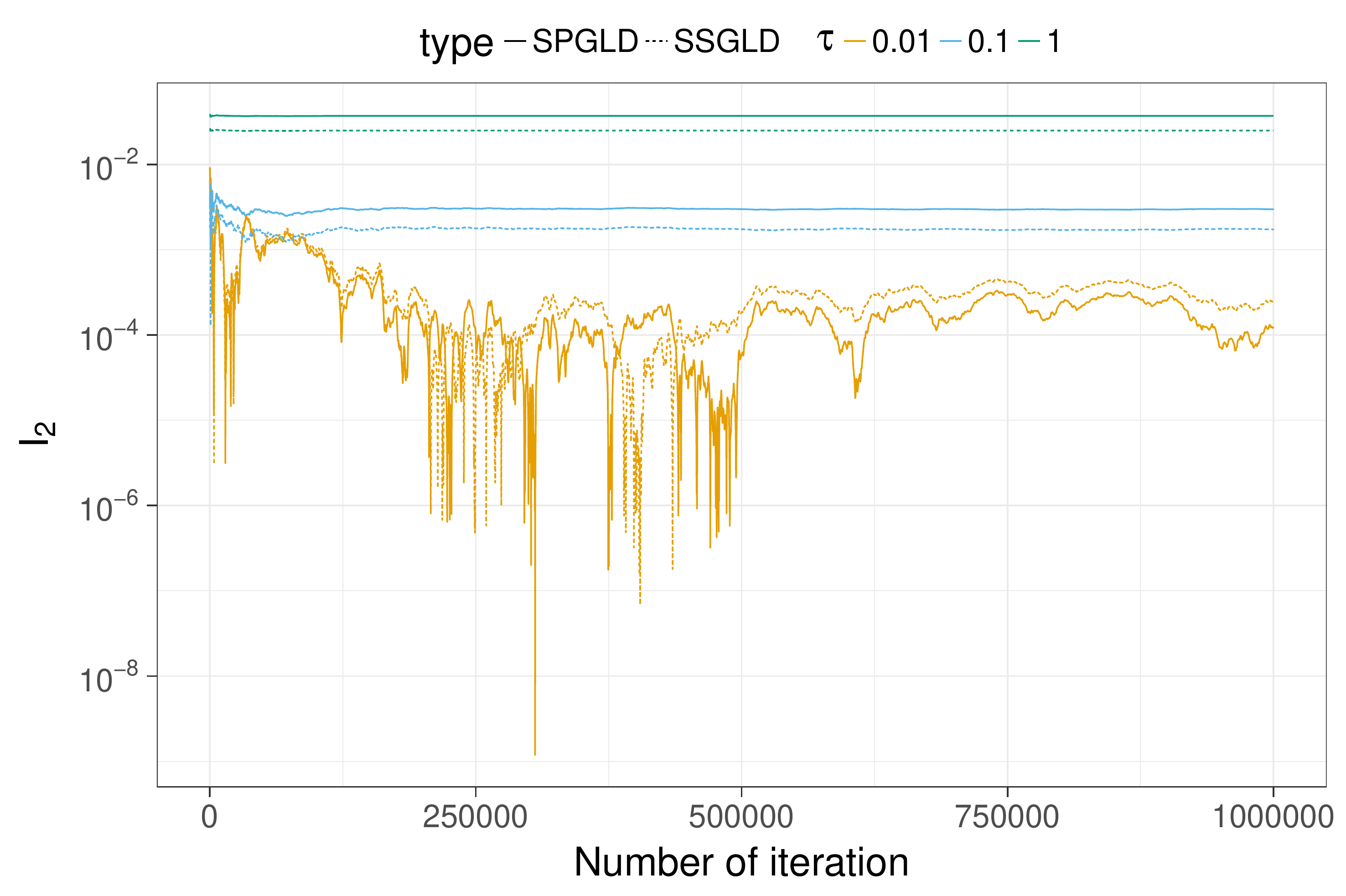}
\caption{ }
\end{subfigure}
\begin{subfigure}{0.32\textwidth}
\includegraphics[width=0.9\linewidth, height=4cm]{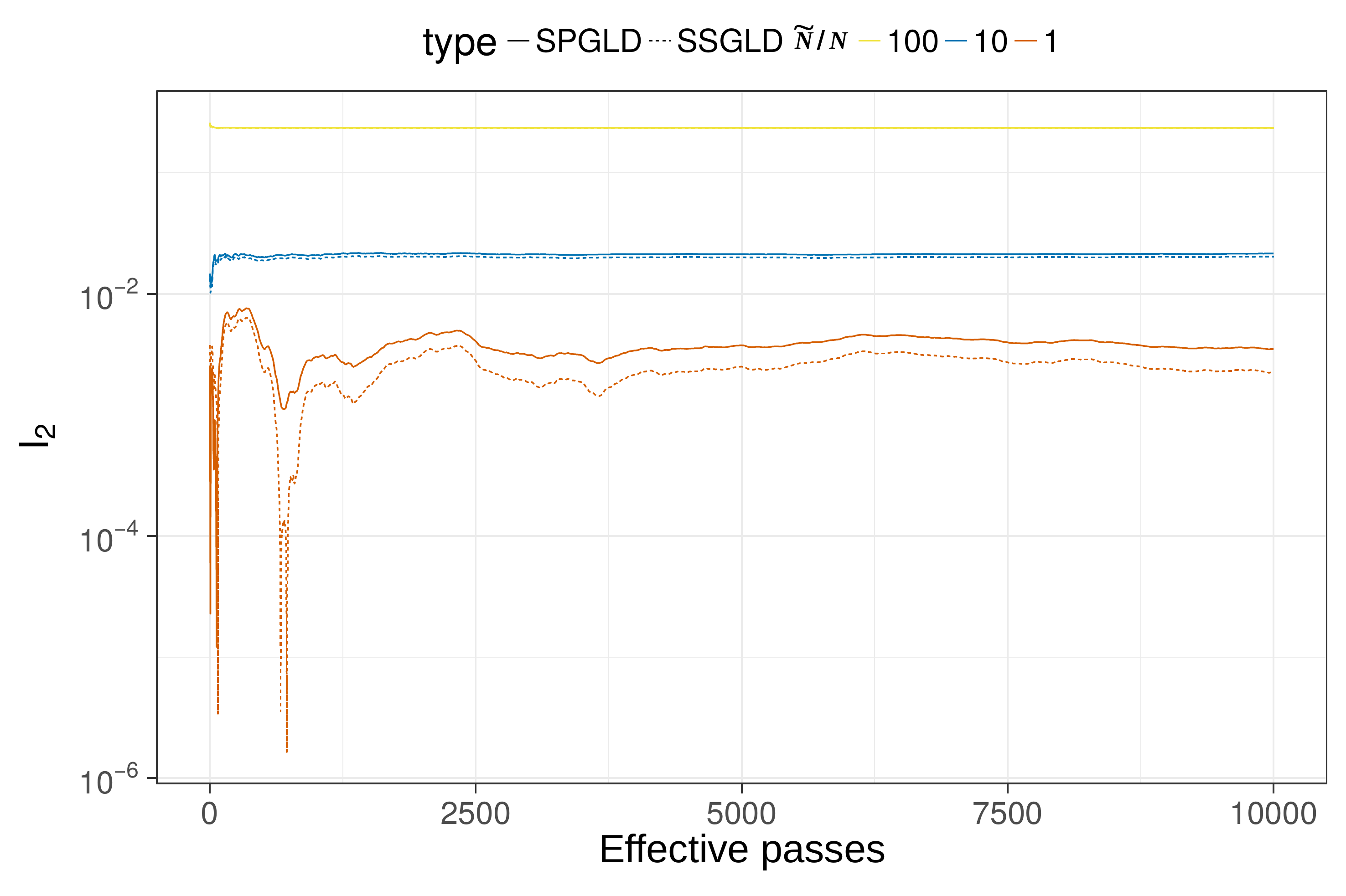}
\caption{ }
\end{subfigure}
\begin{subfigure}{0.32\textwidth}
\includegraphics[width=0.9\linewidth, height =4cm]{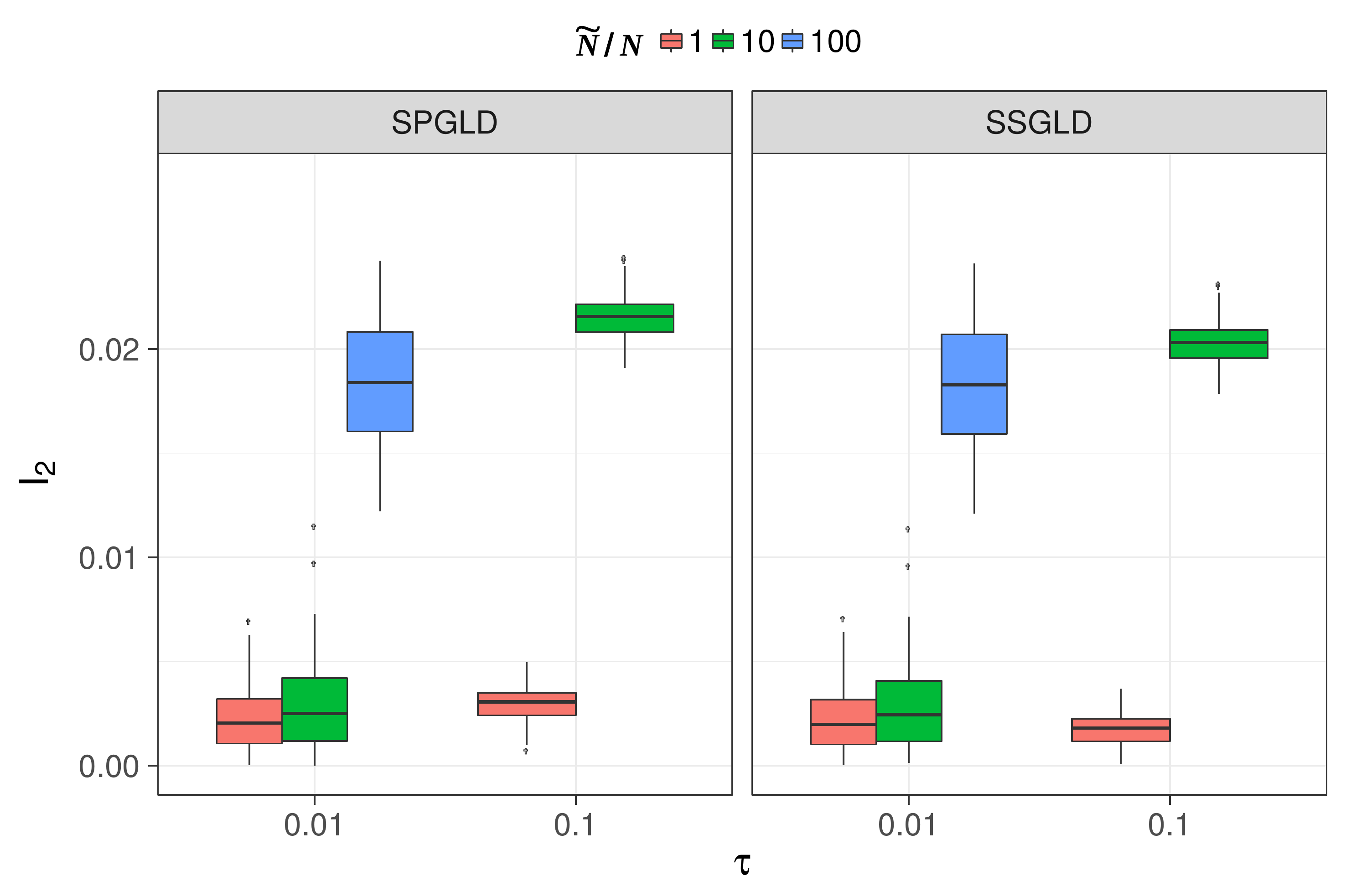}
\caption{}
\end{subfigure}

 \caption{Mean absolute error of estimator of $I_2$ for Heart disease dataset: (top row) results for  $\mathrm{p}_{1,2}(\cdot |(X,Y)_{i \in \{1,\ldots,N\}})$; (a) convergence of SPGLD for $\tilde{N}=N$ , 
(b) convergence of SPGLD in terms of effective passes for $\tau=0.1$, (c) boxplot of SPGLD for full run;
(bottom row) results for $\mathrm{p}_1(\cdot |(X,Y)_{i \in \{1,\ldots,N\}})$; (d) convergence of SPGLD and SSGLD for $\tilde{N}=N$ , 
(e) convergence of SPGLD and SSGLD in terms of effective passes for $\tau=0.1$, (f) boxplot of SPGLD and SSGLD for full run}
\label{fig:heart}
\end{figure}

 \begin{figure}[!h]
\begin{subfigure}{0.32\textwidth}
\includegraphics[width=0.9\linewidth, height=4cm]{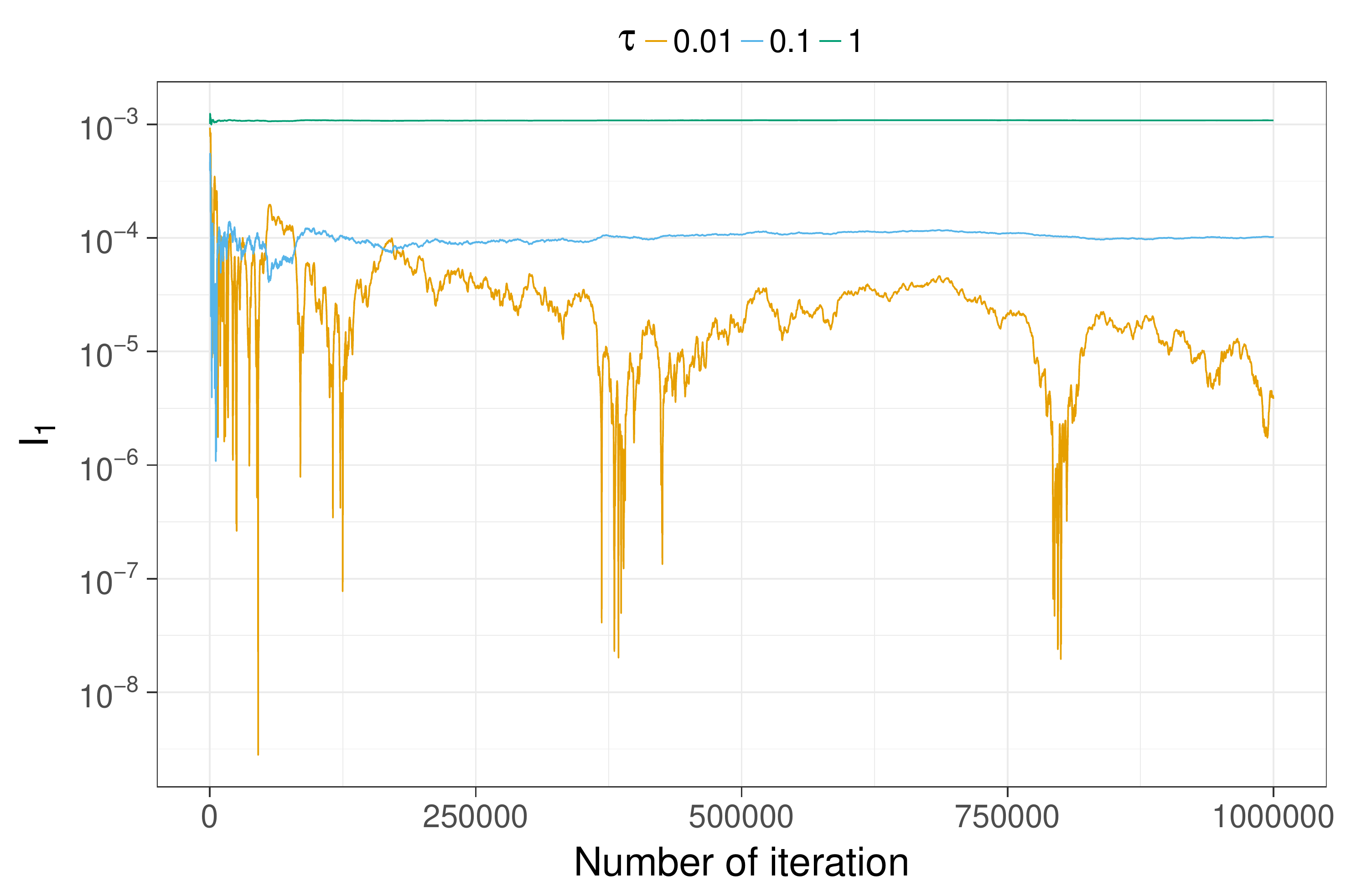}
\caption{ }
\end{subfigure}
\begin{subfigure}{0.32\textwidth}
\includegraphics[width=0.9\linewidth, height=4cm]{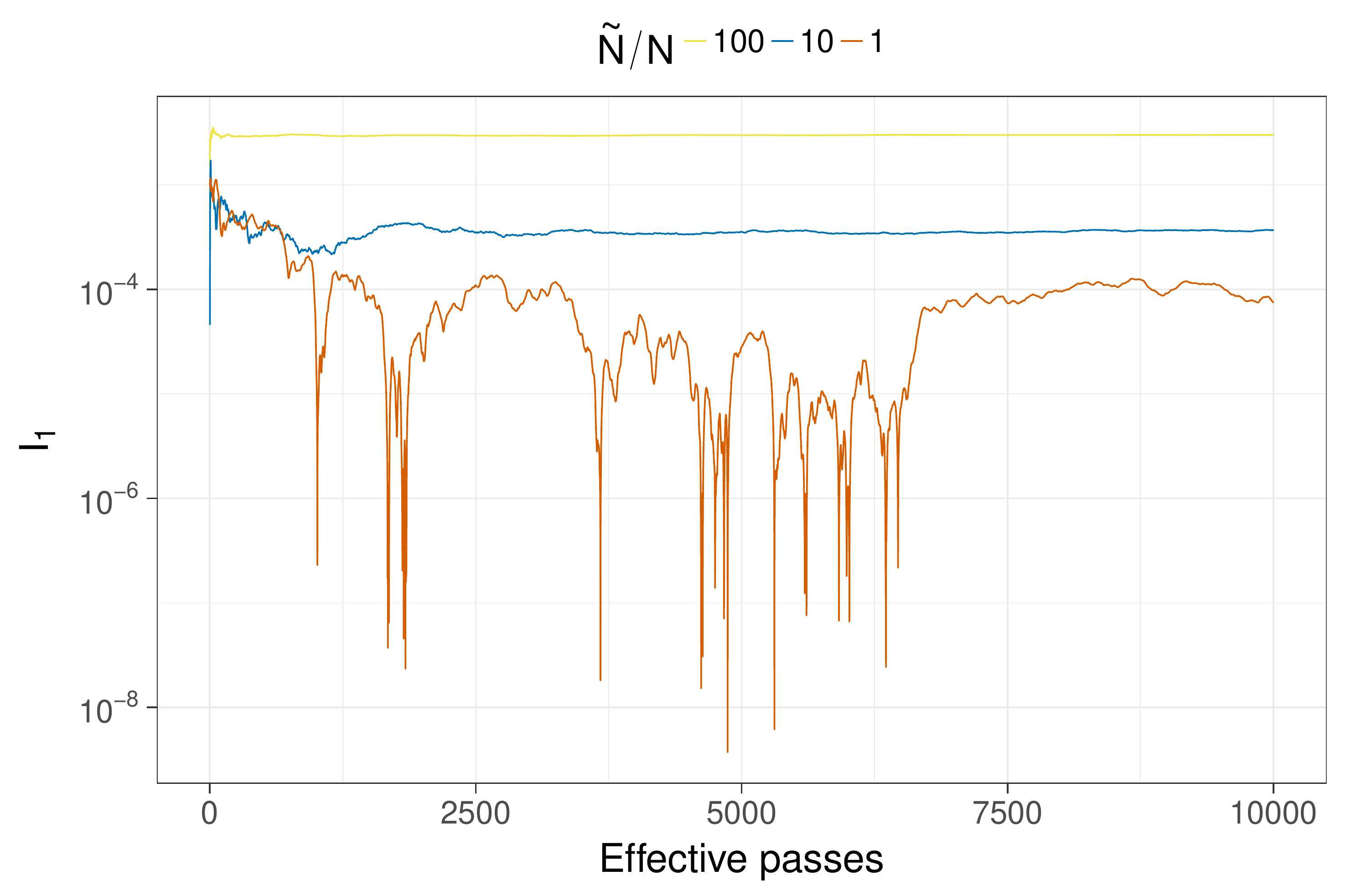}
\caption{ }
\end{subfigure}
\begin{subfigure}{0.32\textwidth}
\includegraphics[width=0.9\linewidth, height =4cm]{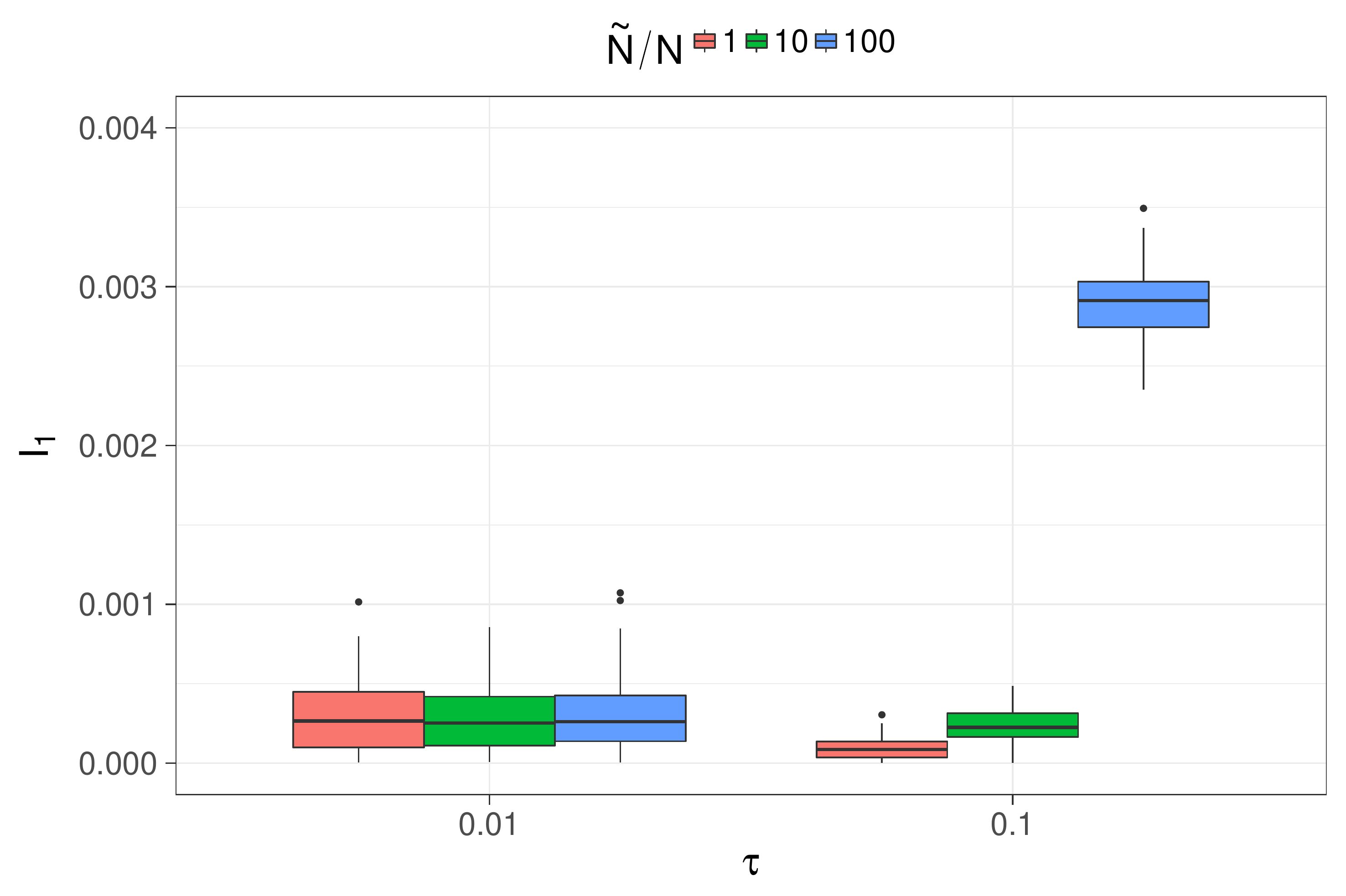}
\caption{ }
\end{subfigure}
 
\begin{subfigure}{0.32\textwidth}
\includegraphics[width=0.9\linewidth, height=4cm]{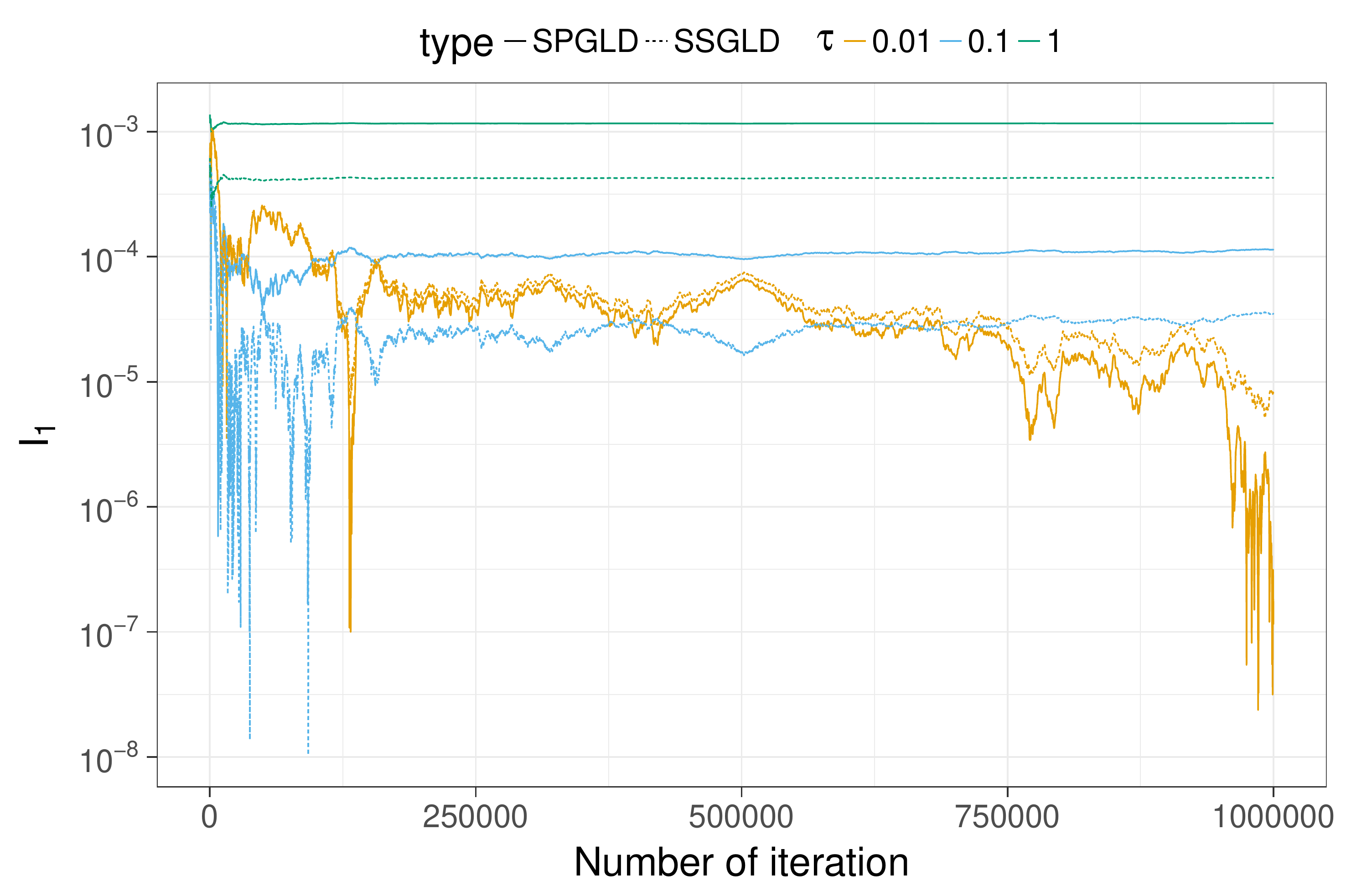}
\caption{ }
\end{subfigure}
\begin{subfigure}{0.32\textwidth}
\includegraphics[width=0.9\linewidth, height=4cm]{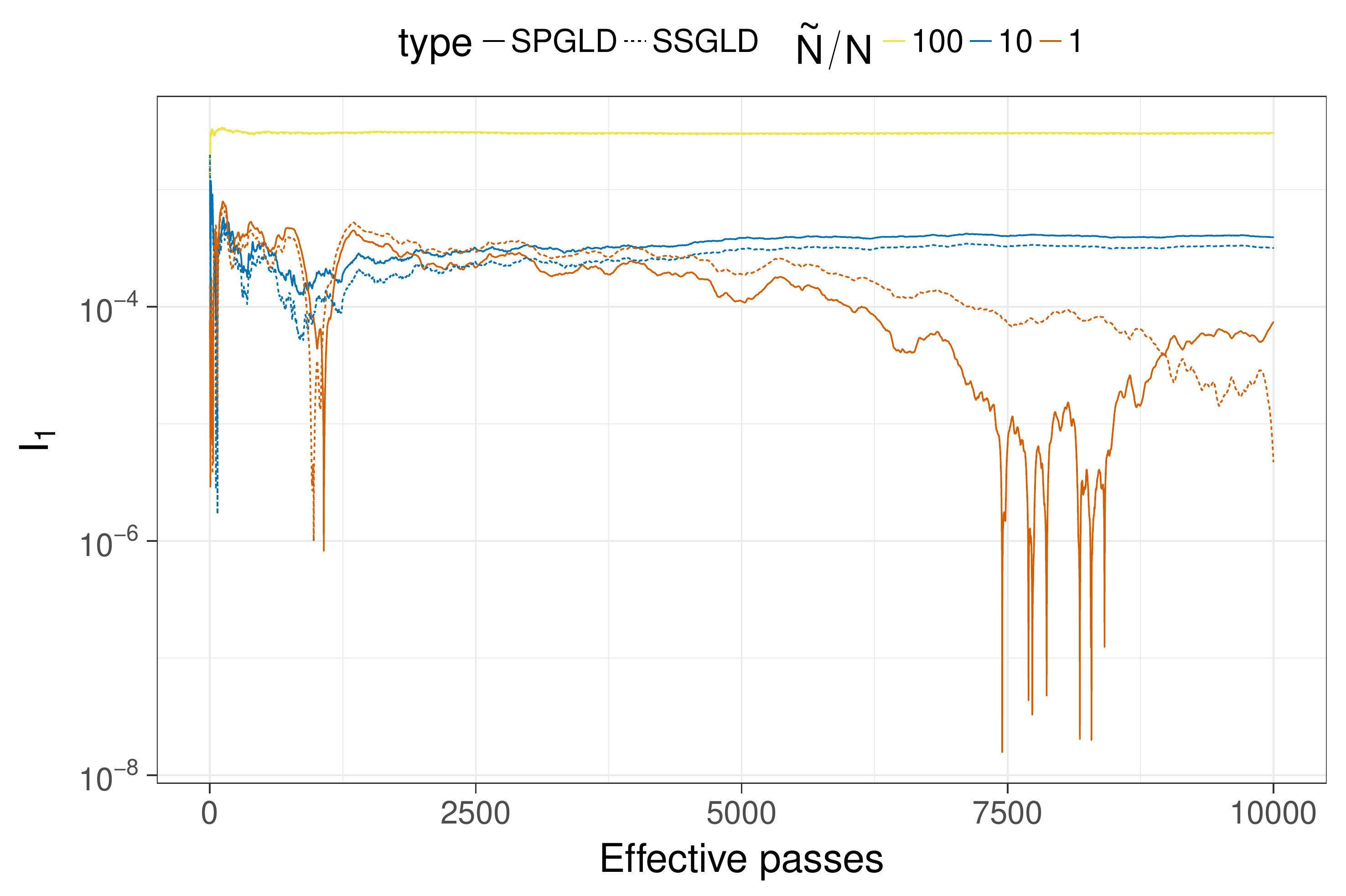}
\caption{ }
\end{subfigure}
\begin{subfigure}{0.32\textwidth}
\includegraphics[width=0.9\linewidth, height =4cm]{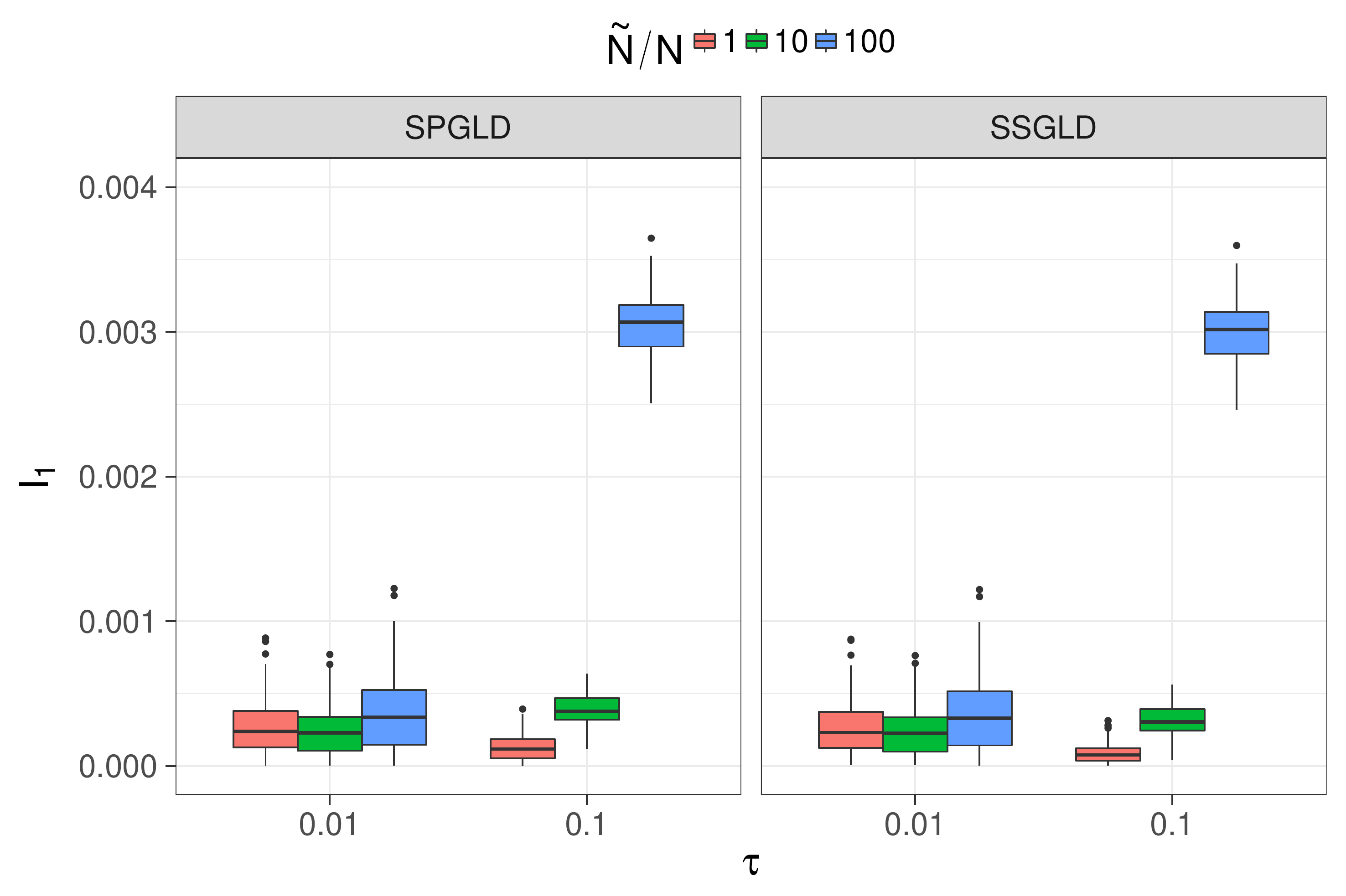}
\caption{}
\end{subfigure}

 \caption{Mean absolute error of estimator of $I_1$ for Heart disease dataset: (top row) results for  $\mathrm{p}_{1,2}(\cdot |(X,Y)_{i \in \{1,\ldots,N\}})$; (a) convergence of SPGLD for $\tilde{N}=N$ , 
(b) convergence of SPGLD in terms of effective passes for $\tau=0.1$, (c) boxplot of SPGLD for full run;
(bottom row) results for $\mathrm{p}_1(\cdot |(X,Y)_{i \in \{1,\ldots,N\}})$; (d) convergence of SPGLD and SSGLD for $\tilde{N}=N$ , 
(e) convergence of SPGLD and SSGLD in terms of effective passes for $\tau=0.1$, (f) boxplot of SPGLD and SSGLD for full run}
\label{fig:heart_theta}
\end{figure}

We consider three data sets from UCI repository \cite{Dua:2017} Heart
disease dataset ($N=270$, $d=14$), Australian Credit Approval dataset
($N=690$, $d=34$) and Musk dataset ($N=476$, $d=166$).   We approximate 
$\mathrm{p}_1(\cdot |(X,Y)_{i \in \{1,\ldots,N\}})$  using SPGLD and
SSGLD, since the associated potential is Lipschitz, whereas regarding
$\mathrm{p}_{1,2}(\cdot |(X,Y)_{i \in \{1,\ldots,N\}})$ we only apply
SPGLD. 

SPGLD is performed using the following stochastic gradient
\begin{equation*}
\tilde\Theta_1(\beta,Z) = (N/\tilde{N}) \sum_{n\in Z} \nabla \ell_n(\beta) + a_2 \beta \eqsp,
\end{equation*}
where $a_2$ is set to $0$ in the case of
$\mathrm{p}_1(\cdot |(X,Y)_{i \in \{1,\ldots,N\}})$ and $Z$ is a
uniformly distributed random subset of $\{1,\ldots,N\}$ with cardinal
$\tilde{N} \in \{1,\ldots,N\}$. In addition, the proximal operator
associated with $\beta \mapsto a_1 \sum_{i=1}^d \abs{\beta_i}$ is
given for all $\beta \in \rset^d$ and $\gamma >0$ by (see
\eg~\cite{parikh:boyd:2013})
\begin{equation*}
  (\prox_{a_1,\ell_1}^{\gamma}(\beta))_i = \signop (\beta_i)
  \max(\abs{\beta_i}-a_1 \gamma,0) \eqsp, \, \text{ for $i \in \{1,\ldots,d\}$}\eqsp.
\end{equation*}
SSGLD is performed using the following stochastic subgradient
\begin{equation*}
\Theta(\beta,Z) = (N/\tilde{N}) \sum_{n\in Z} \nabla \ell_n(\beta) + a_1 \sum_{i=1}^d \signop(\beta_i) \bfe_i \eqsp,
\end{equation*}
where $(\bfe_i)_{i \in \{1,\ldots,d\}}$ denotes the canonical basis and $Z$ is a 
uniformly distributed random subset of $\{1,\ldots,N\}$ with cardinal
$\tilde{N} \in \{1,\ldots,N\}$.

Based on the results of SPGLD and SSGLD, we estimate the posterior
mean $I_1$ and $I_2$ of the test functions $\beta \mapsto \beta_1$ and
$\beta \mapsto (1/d) \sum_{i=1}^d \beta_i^2$. For our experiments, we
use constant stepsizes $\gamma$ of the form $\tau(L+m)^{-1}$ with
$\tau = 0.01, 0.1, 1$ and for stochastic (sub) gradient we use
$\tilde{N}=N,\lfloor N/10\rfloor,\lfloor N/100\rfloor$.  For all datasets and all settings of $\tau$,
$\tilde{N}$ we run $100$ replications of SPGLD (SSGLD), where each run
was of length $10^6$.  For each set of parameters we estimate
$I_1,I_2$ and we compute the absolute errors, where the true value
were obtained by prox-MALA (see \cite{pereyra:2015}) with $10^7$ iterations and stepsize
corresponding to optimal acceptance ratio $\approx 0.5$, see
\cite{RobRos1998}. The results for $I_2$ are presented on \Cref{fig:austra},
\Cref{fig:heart} and \Cref{fig:musk} for Australian Credit Approval
dataset, Heart disease dataset and Musk data respectively.
The results for $I_1$ are presented on \Cref{fig:austra_theta},
\Cref{fig:heart_theta} and \Cref{fig:musk_theta} for Australian Credit Approval
dataset, Heart disease dataset and Musk data respectively.
We note that in the all cases, bias decreases but convergence
becomes slower with decreasing $\gamma$. When we look for stochastic
(sub)gradient then the bias of estimators and also their variance increase
when we decrease $\tilde{N}$.  However if we look for effective passes, i.e. number
of iteration is scaled with the cost of computing gradients,
we observe
that convergence is faster with reasonably small $\tilde N$. If we compare SSGLD
with SPGLD we see that in almost all cases, except Musk dataset, SSGLD leads to slightly smaller bias.
For the Musk dataset differences between SSGLD and SPGLD are negligible and we do not present the results for SPGLD.
In the presented experiments, all results agrees with our theoretical
findings and suggest that SPGLD or SSGLD could be an alternative for
other MCMC methods.

%
%
%
  
 \begin{figure}[!h]
\begin{subfigure}{0.32\textwidth}
\includegraphics[width=0.9\linewidth, height=4cm]{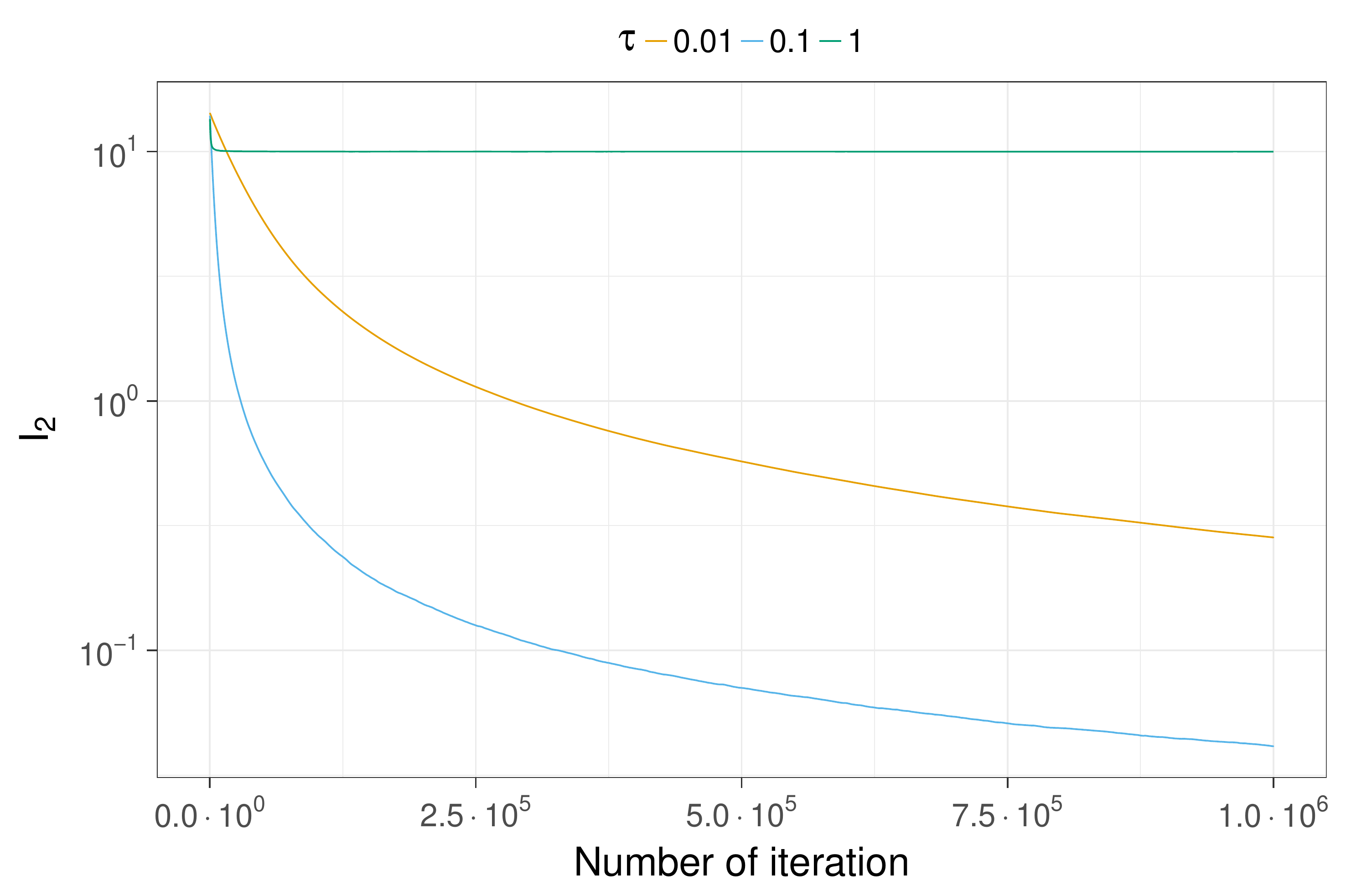}
\caption{ }
\end{subfigure}
\begin{subfigure}{0.32\textwidth}
\includegraphics[width=0.9\linewidth, height=4cm]{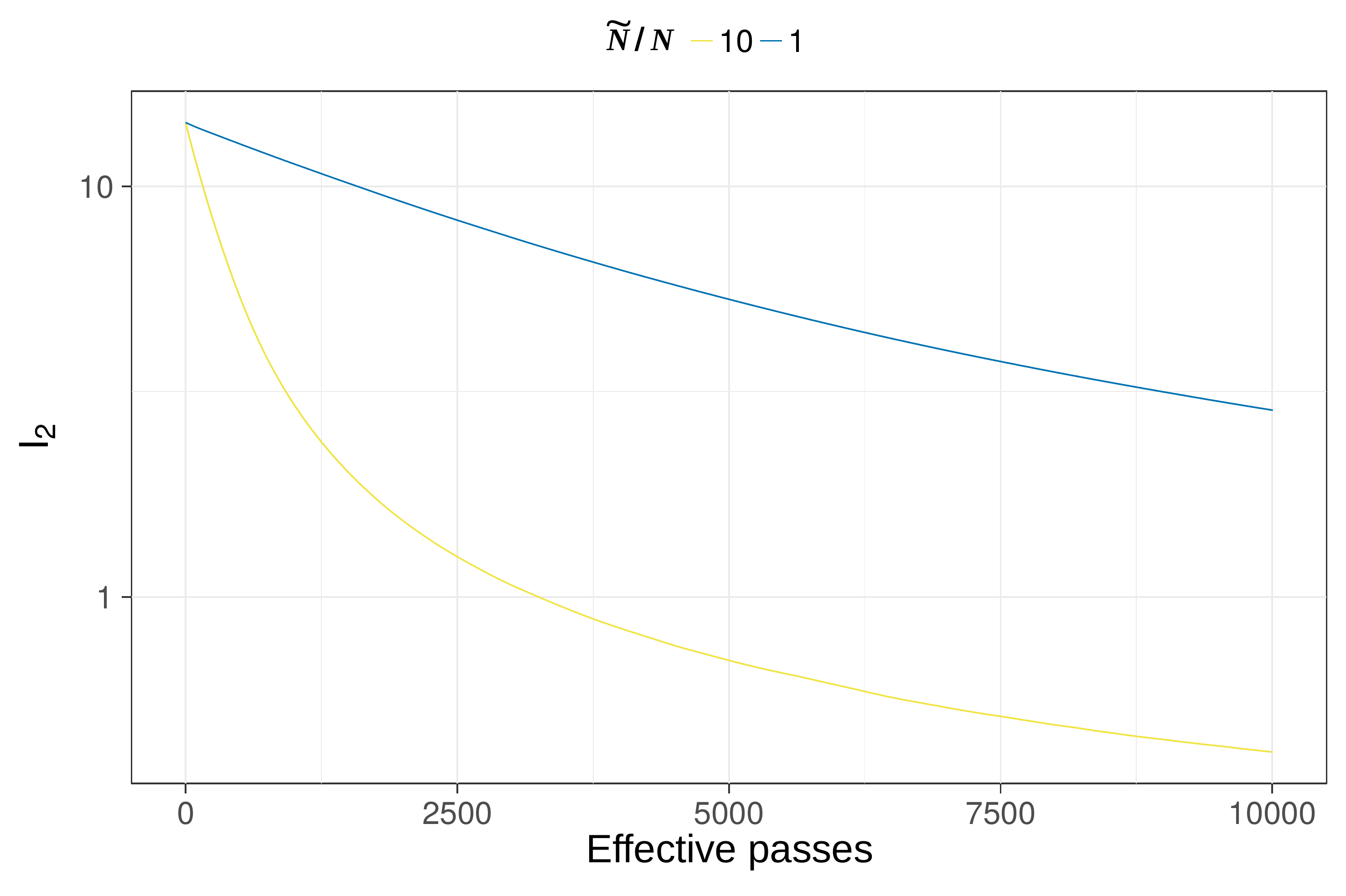}
\caption{ }
\end{subfigure}
\begin{subfigure}{0.32\textwidth}
\includegraphics[width=0.9\linewidth, height =4cm]{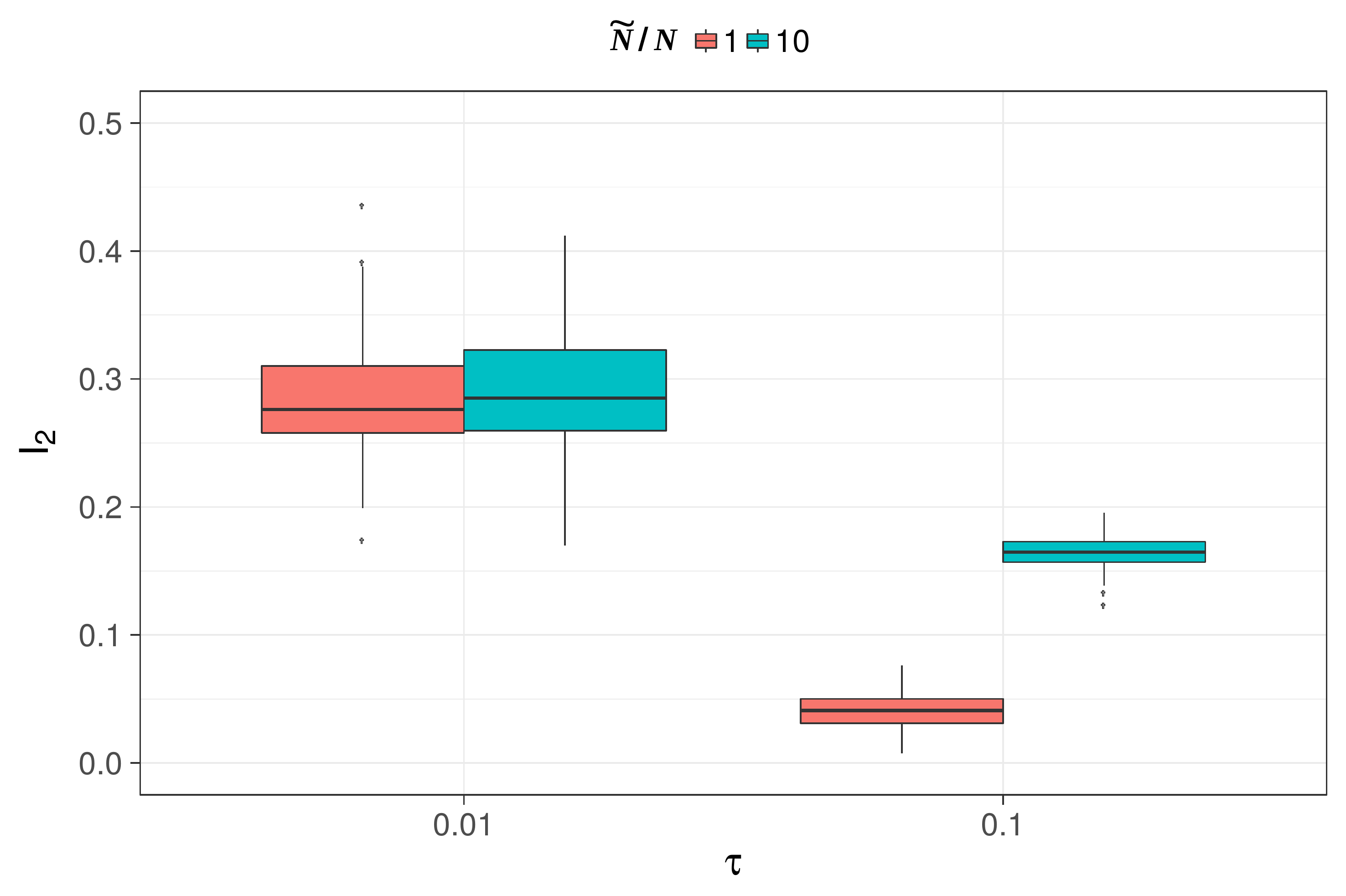}
\caption{ }
\end{subfigure}
 
\begin{subfigure}{0.32\textwidth}
\includegraphics[width=0.9\linewidth, height=4cm]{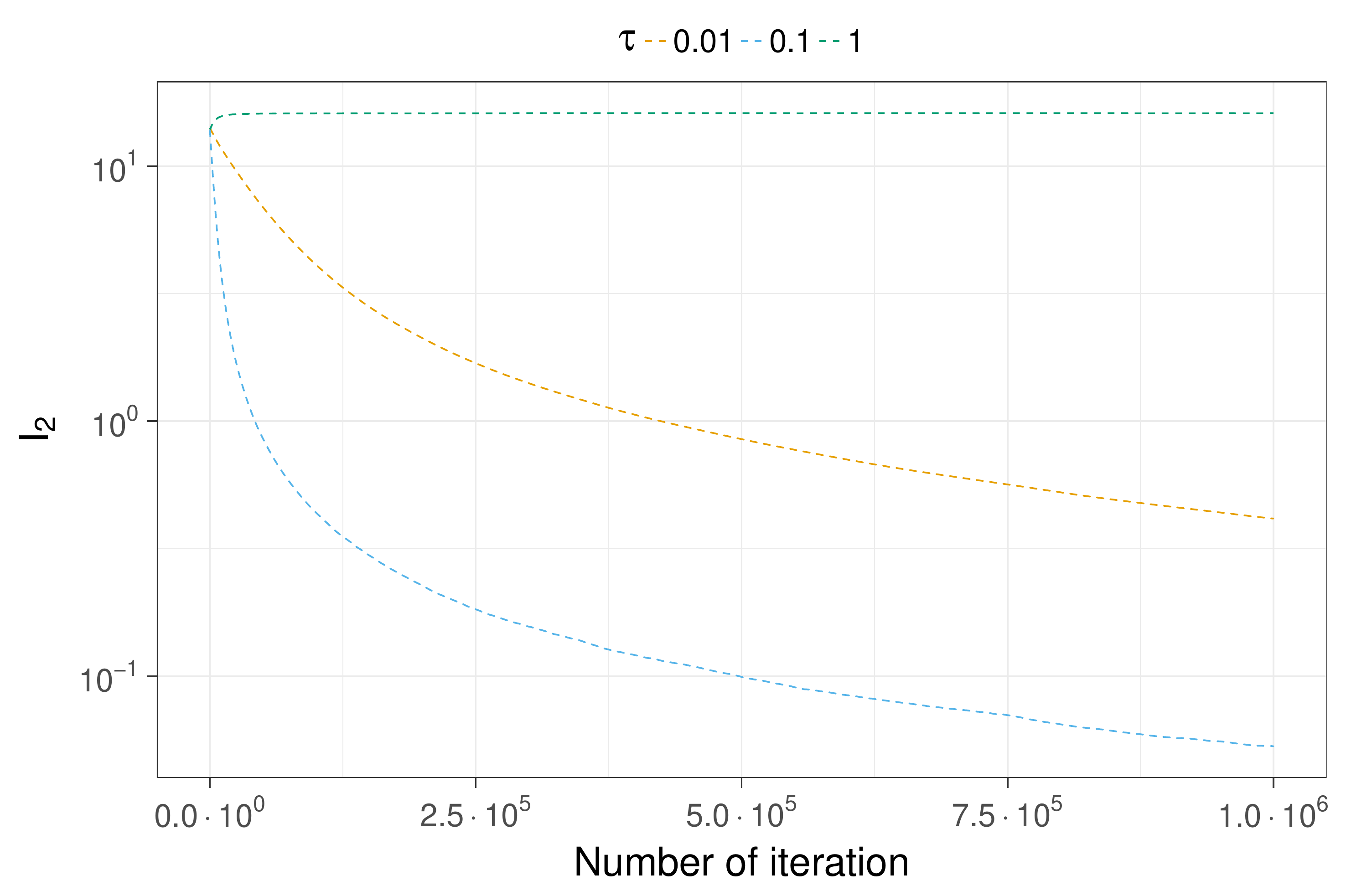}
\caption{ }
\end{subfigure}
\begin{subfigure}{0.32\textwidth}
\includegraphics[width=0.9\linewidth, height=4cm]{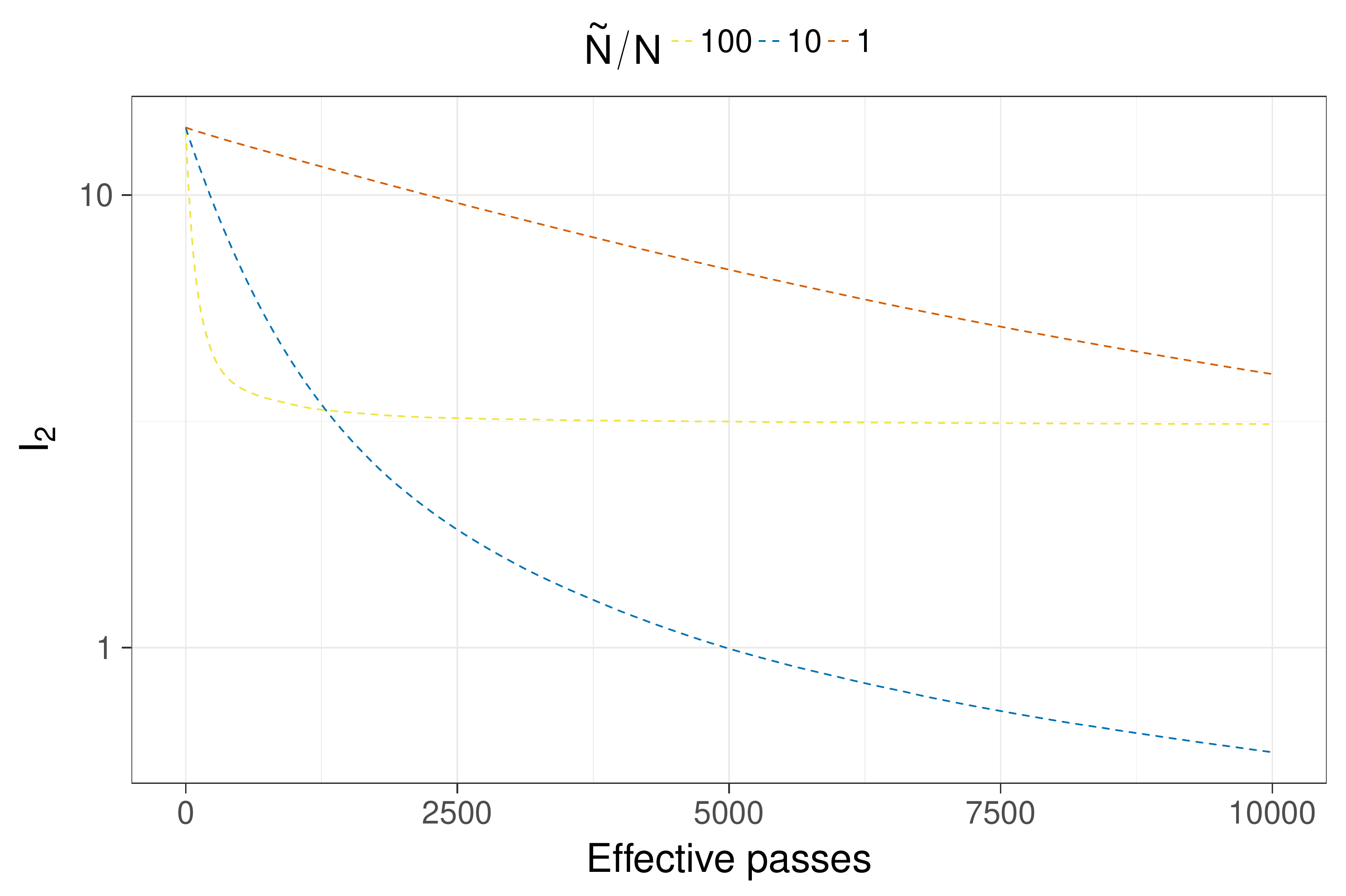}
\caption{ }
\end{subfigure}
\begin{subfigure}{0.32\textwidth}
\includegraphics[width=0.9\linewidth, height =4cm]{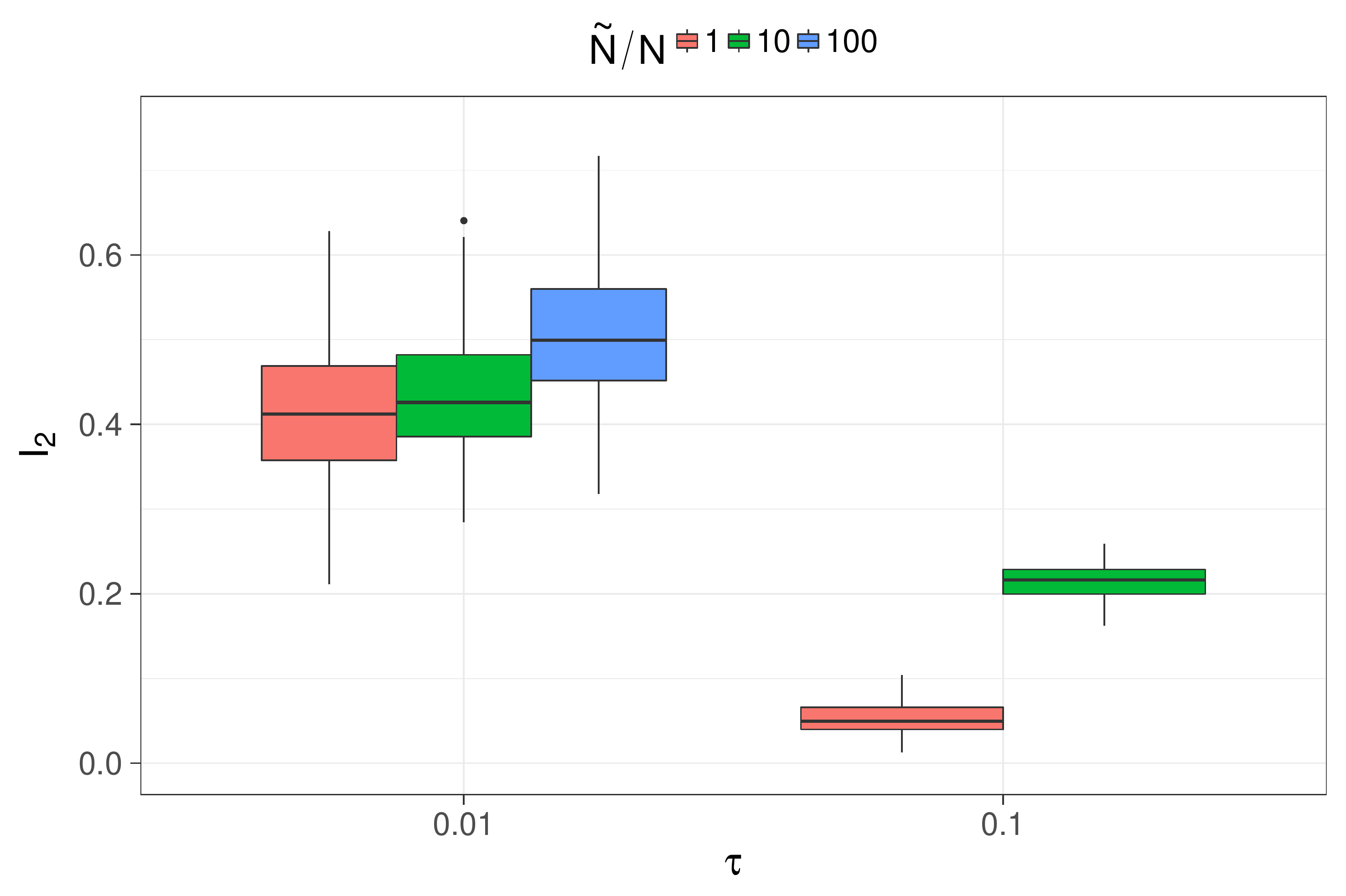}
\caption{}
\end{subfigure}

 \caption{Mean absolute error of estimator of $I_2$ for Musk dataset: (top row) results for  $\mathrm{p}_{1-2}$ prior; (a) convergence of SPGLD for $\tilde{N}=1$ , 
(b) convergence of SPGLD in terms of effective passes for $\tau=0.1$, (c) boxplot of SPGLD for full run;
(bottom row) results for  $\mathrm{p}_{1}$ prior; (d) convergence of  SSGLD for $\tilde{N}=N$ , 
(e) convergence of  SSGLD in terms of effective passes for $\tau=0.1$, (f) boxplot of SSGLD for full run}
\label{fig:musk}
\end{figure}

 \begin{figure}[!h]
\begin{subfigure}{0.32\textwidth}
\includegraphics[width=0.9\linewidth, height=4cm]{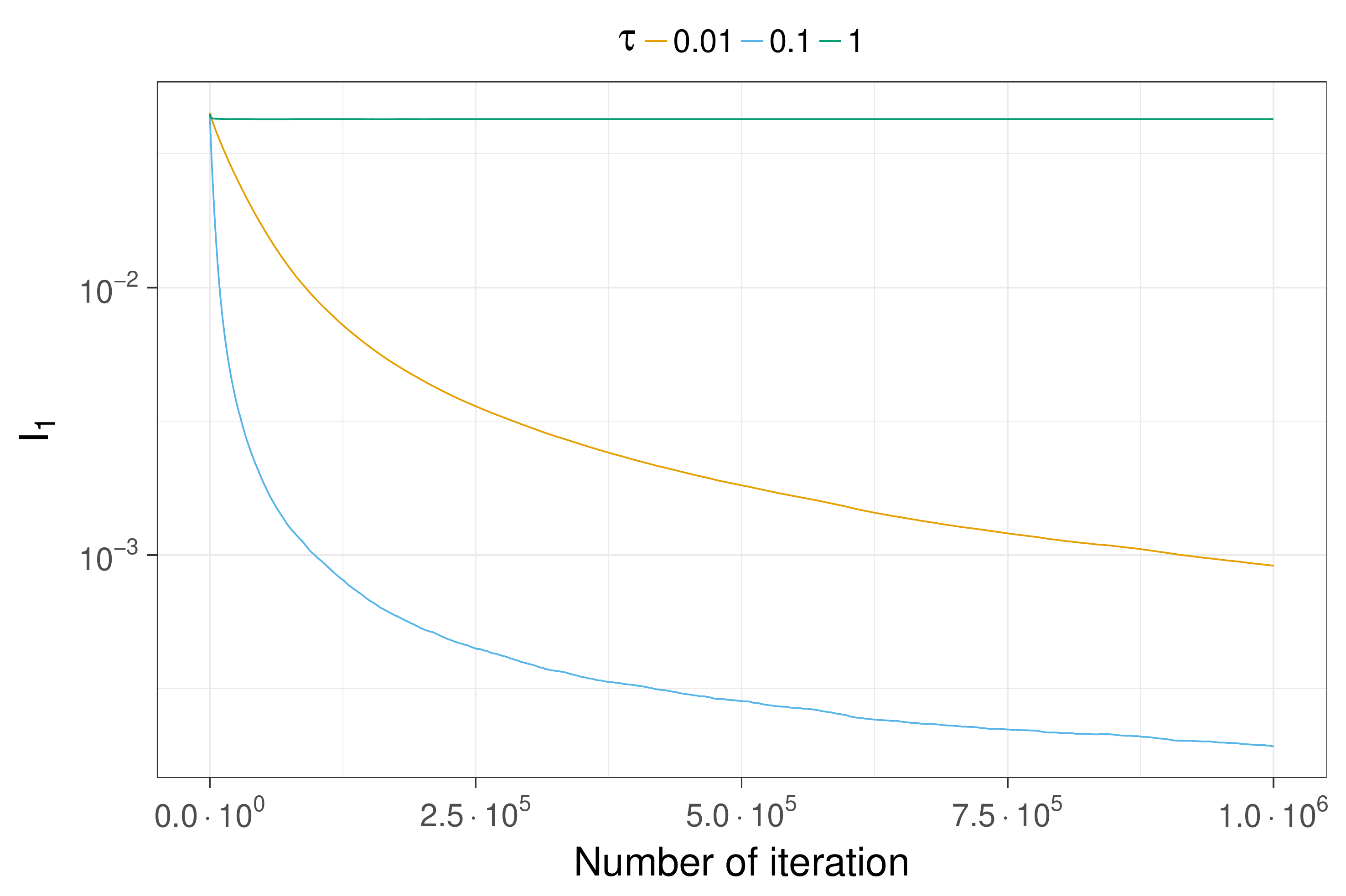}
\caption{ }
\end{subfigure}
\begin{subfigure}{0.32\textwidth}
\includegraphics[width=0.9\linewidth, height=4cm]{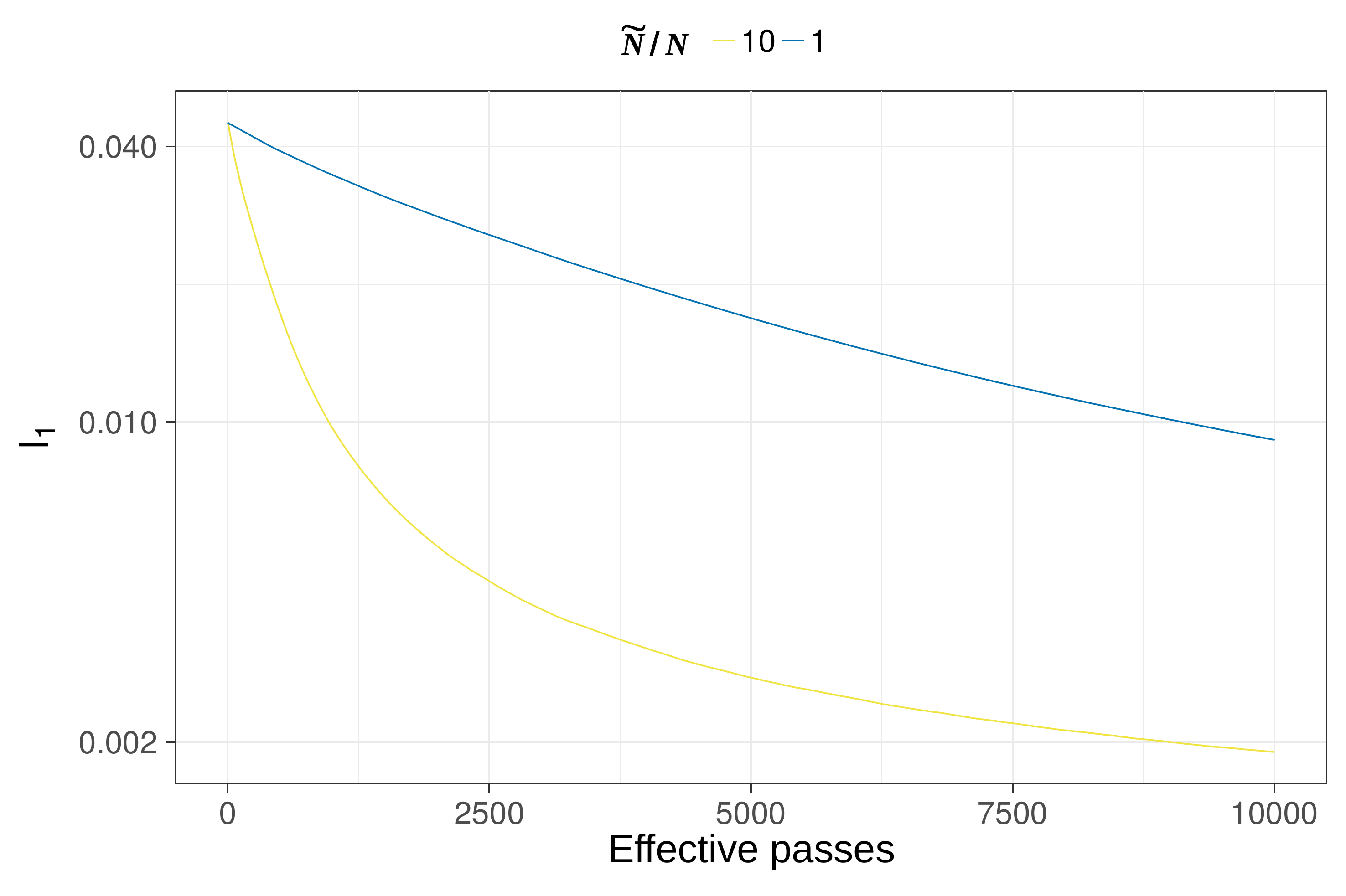}
\caption{ }
\end{subfigure}
\begin{subfigure}{0.32\textwidth}
\includegraphics[width=0.9\linewidth, height =4cm]{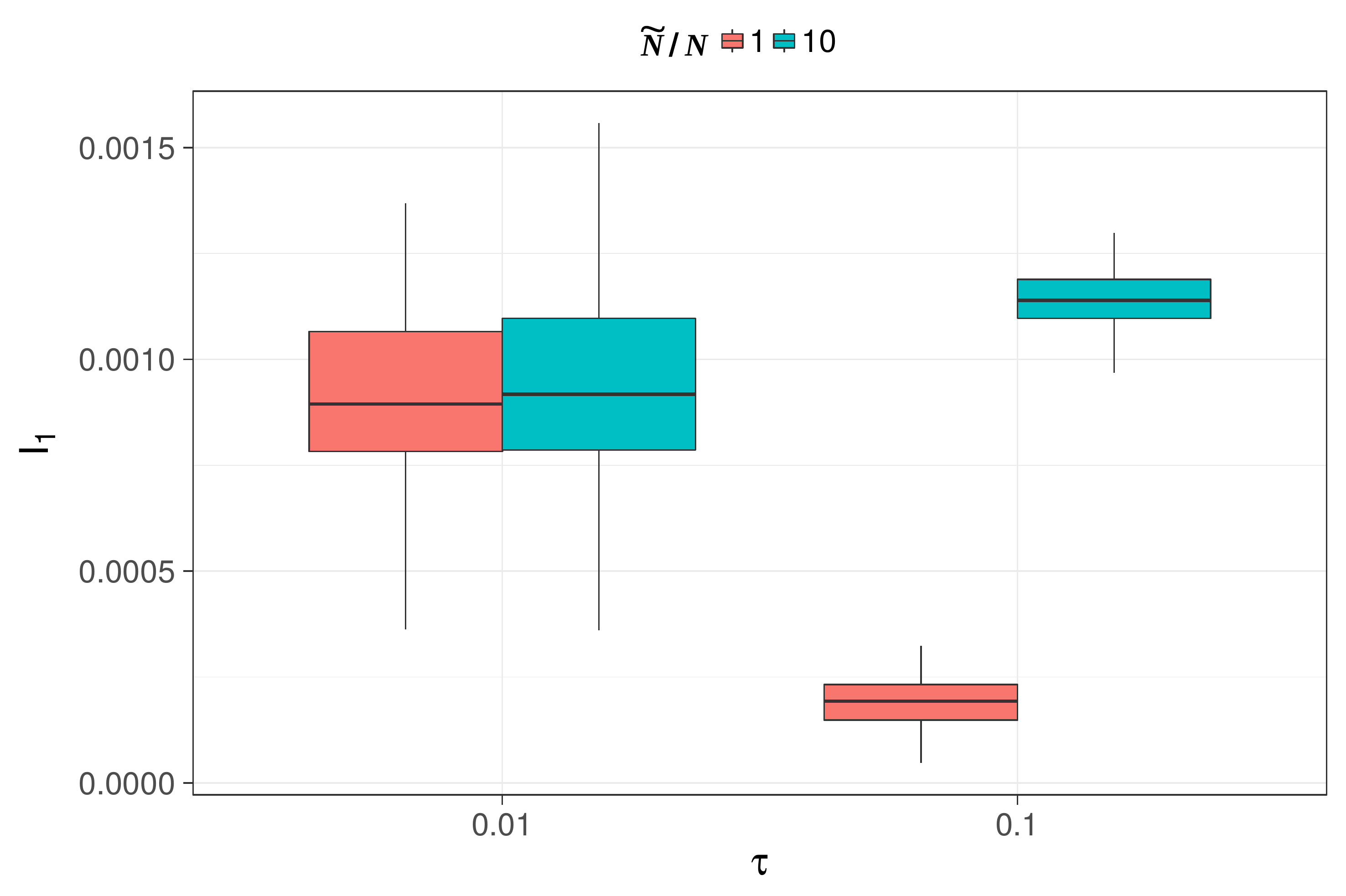}
\caption{ }
\end{subfigure}
 
\begin{subfigure}{0.32\textwidth}
\includegraphics[width=0.9\linewidth, height=4cm]{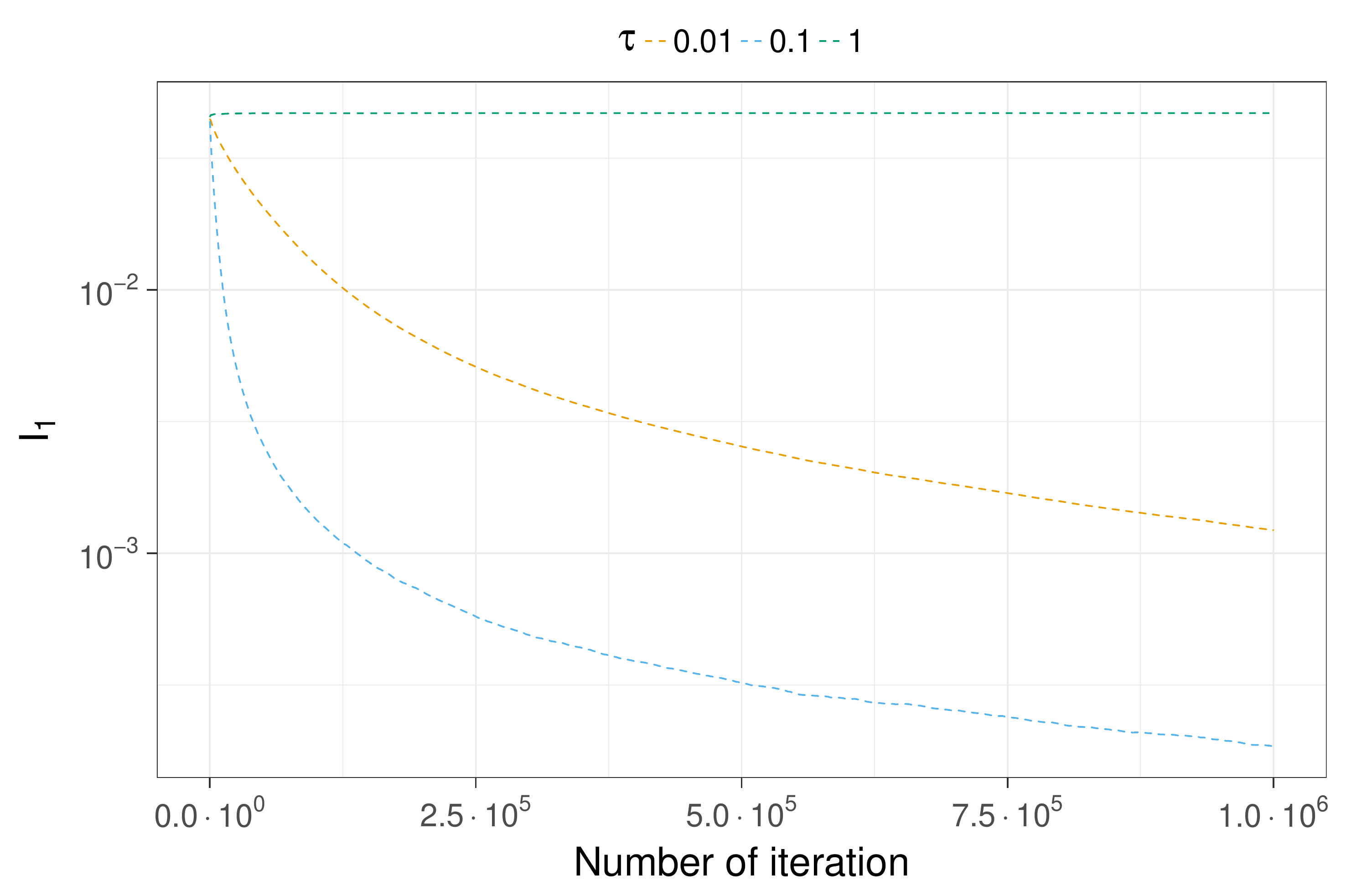}
\caption{ }
\end{subfigure}
\begin{subfigure}{0.32\textwidth}
\includegraphics[width=0.9\linewidth, height=4cm]{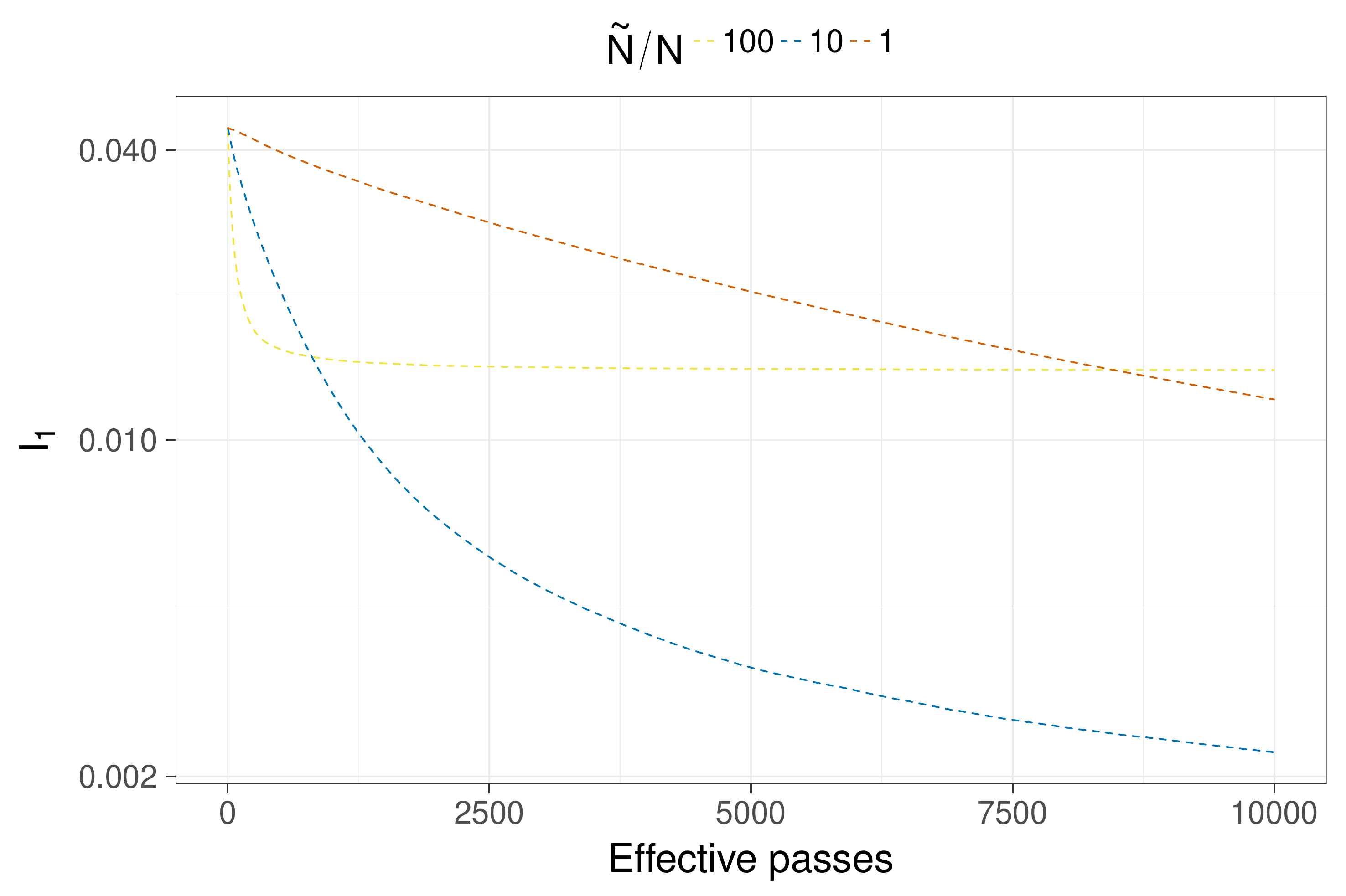}
\caption{ }
\end{subfigure}
\begin{subfigure}{0.32\textwidth}
\includegraphics[width=0.9\linewidth, height =4cm]{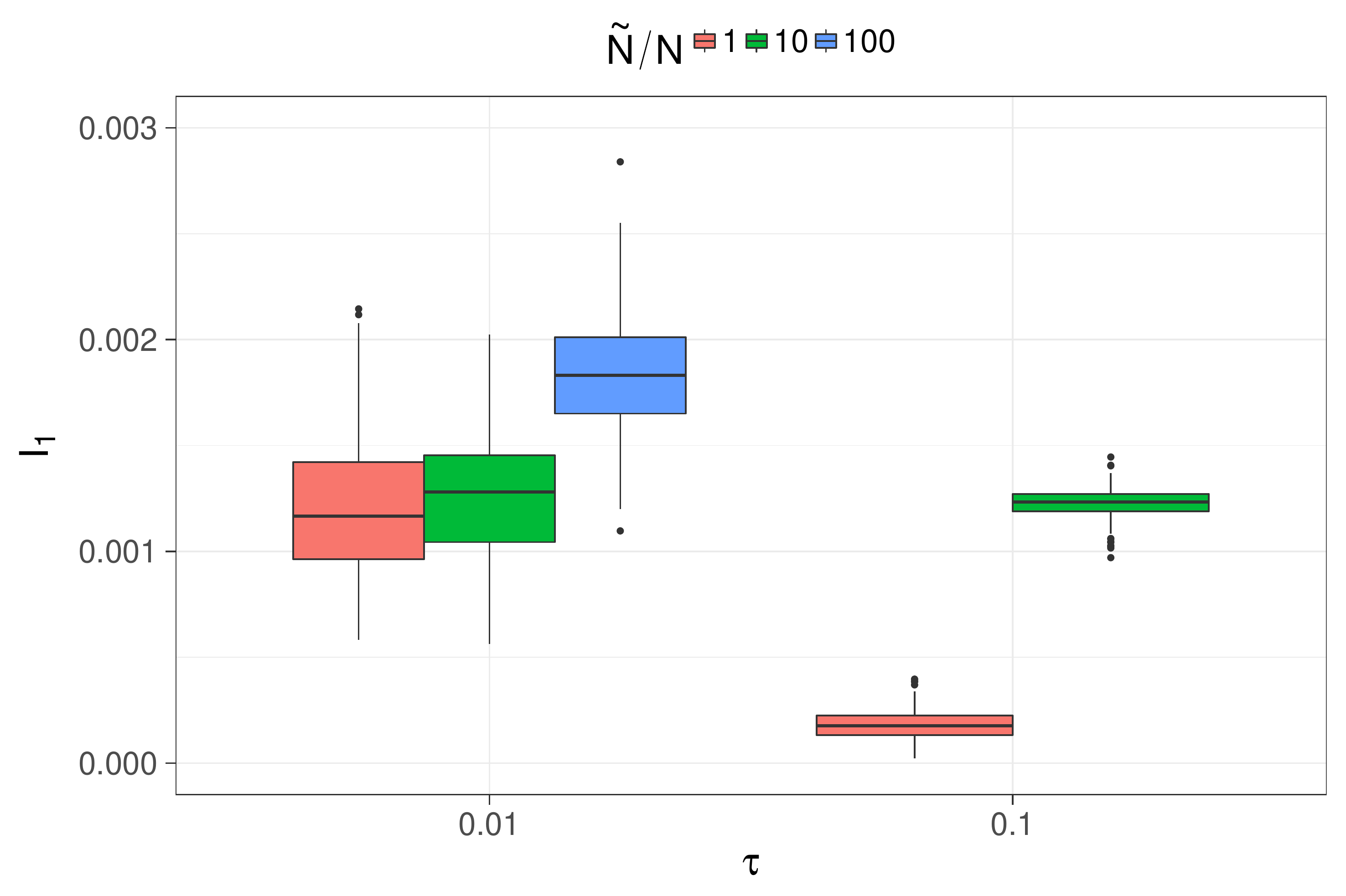}
\caption{}
\end{subfigure}

 \caption{Mean absolute error of estimator of $I_1$ for Musk dataset: (top row) results for  $\mathrm{p}_{1-2}$ prior; (a) convergence of SPGLD for $\tilde{N}=1$ , 
(b) convergence of SPGLD in terms of effective passes for $\tau=0.1$, (c) boxplot of SPGLD for full run;
(bottom row) results for  $\mathrm{p}_{1}$ prior; (d) convergence of SSGLD for $\tilde{N}=N$ , 
(e) convergence of  SSGLD in terms of effective passes for $\tau=0.1$, (f) boxplot of  SSGLD for full run}
\label{fig:musk_theta}
\end{figure}


\section{Discussion}
In this paper, we presented a novel interpretation of the Unadjusted Langevin Algorithm as a first order optimization algorithm, and a new technique of proving nonasymptotic bounds for ULA, based on the proof techniques known from convex optimization. Our proof technique gives simpler proofs of some of the previously known non-asymptotic results for ULA. It can be also used to prove non-asymptotic bound that were previously unknown. Specifically, to the best of the authors knowledge, we provide the first non-asymptotic results for Stochastic Gradient ULA in the non-strongly convex case, as well as the first non-asymptotic results in the non-smooth non-strongly convex case. Furthermore, our technique extends effortlessly to the stochastic non-smooth case, and to the best of the authors knowledge we provide the first nonasymptotic analysis of that case.

Furthermore our new perspective on the Unadjusted Langevin Algorithm, provides a starting point for further research into connections between Langevin Monte Carlo and Optimization. Specifically, we believe that a very promising direction for further research is translating well known efficient optimization algorithms into efficient sampling algorithms and proving non-asymptotic bounds for those more efficient algorithms.

\section{Postponed proofs}
\label{sec:postponed-proofs}

\subsection{Proof of of \Cref{lem:kl_minimizer}} 
\label{sec:proof-crefl}
\ref{item:1:lem:kl_minimizer}
Since $\rme^{-U}$ is integrable with respect to the Lebesgue measure, under \Cref{assum:convexity}$(m)$ for $ m \geq 0$, by \cite[Lemma 2.2.1]{brazitikos:giannopoulos:valettas:al:2014}, there exists $C_1,C_2 > 0$ such that for all $x \in \rset^d$, $ U(x) \geq C_1 \norm{x} -C_2$. This inequality and \Cref{assum:grad_lip} implies that $\pi \in \Pens_2(\rset^d)$. In addition, since the function $x \mapsto U(x)\rme^{-U(x)/2}$ is bounded on $[-C_2, \plusinfty)$, we have for all $x \in\rset^d$,
\[
\left| \left( U(x) \rme^{-U(x)/2} \right) \rme^{-U(x)/2} \right| \leq  C_3 \rme^{-U(x)/2}
\]
for some constant $C_3$. From this, and $U(x) \geq C_1 \norm{x} - C_2$ we conclude that $\Escr(\pi) < \plusinfty$.
Using the same reasoning, we have $\Hscr(\pi) < \plusinfty$ which finishes the proof of the first part.


\noindent
\ref{item:2:lem:kl_minimizer}
If $\mu$ does not admit a density with respect to Lebesgue measure, then
both sides of \eqref{eq:lem:kl_minimizer} are $+\infty$.  Second if
$\mu$ admits a density still denoted by $\mu$ with respect to the
Lebesgue measure, we have by \eqref{eq:def_free_energy}:

\begin{equation*}
\Fscr(\mu) - \Fscr(\pi) 
 = \KLarg{\mu}{\pi} + \int_{\rset^d} \defEns{\mu(x) - \pi(x)} \defEns{U(x) + \log(\pi(x)) } \rmd x = \KLarg{\mu}{\pi}  \eqsp.
\end{equation*}

\subsection{Proof of \Cref{coro:ula_non_increas_sz}}
\label{sec:proof-crefc_coro:ula_non_increas_sz}
  Using \Cref{thm:const_step_conv} we first get
\begin{equation} \label{eq:cor_dec_step}
\KLarg{\nu_n}{\pi} \leq \left. \wassersteinTarg{\mu_0}{ \pi} \middle/ (2\Stepa[0,n]) \right.
+ (Ld/\Stepa[0,n]) \sum_{k=1}^{n} \stepa[k]^2  \eqsp. 
\end{equation}
Note that  using a simple integral test, we have $\Stepa[0,n]  \geq C_1 n^{1 - \alpha}$
for some constant $C_1 \geq 0$. On the other hand, for some constant $C_2 \geq 0$  we  have  $\sum_{k=1}^n \steps_k^2 \leq C_2 (1+n^{1-2\alpha})$ if $\alpha \not = 1/2$, and 
 $\sum_{k=1}^n \steps_k^2 \leq C_2(1+\log(n))$ if $\alpha =1/2$.  Combining all these inequalities in \eqref{eq:cor_dec_step} concludes the proof.

\subsection{Proofs of \Cref{sec:stoch-sub-grad}}
\label{sec:proofs-crefthm:st}

Note that for all $\step,\tstep >0$, $\bRker_{\step,\tstep}$ can be decomposed as $\bSker_{\step} \Tker_{\tstep} $  where $\Tker_{\tstep}$ is defined in \eqref{eq:def_Sker_Tker}
and $\bSker_{\step}$ is given by \eqref{eq:definition_sker_ssgld}.
Then similarly to the proof of \Cref{thm:const_step_conv}, we first
give a preliminary bound on $ \Fscr(\mu \bRker_{\steps,\tstep}) - \Fscr( \pi
)$ for $\mu \in \Pens_2(\rset^d)$ and $\step,\tstep >0$ as in \Cref{thm:basic-one-step}.
\begin{lemma}
  \label{lem:basic-one-step_ssgld}
Assume \Cref{assum:convexity}($0$) and \Cref{assum:stochastic_subgradient}. For all $\steps >0 $ and $\mu \in \Pens_2(\rset^d)$, 
\[
2\steps\defEns{\Escr(\mu ) - \Escr(\pi)} \leq  \wasserstein_2^2(\mu, \pi) - \wasserstein_2^2(\mu \bSkera, \pi) + \step^2\defEns{\ClyapU^2 + \vargrad_{\gradst}(\mu)} \eqsp,
\]
where $\Escr$ and $\Tkers$ are defined in
\eqref{eq:def_potential_energy} and \eqref{eq:def_Sker_Tker}
respectively, $\vargrad_{\gradst}(\mu)$ in \eqref{eq:definition_var_grad_sto} and  $\bSkera$ in \eqref{eq:definition_sker_ssgld}.
 \end{lemma}
\begin{proof}
Let $Z$ be a random variable with distribution $\eta$, $\steps >0$ and $\mu \in \Pens_2(\rset^d)$. For all $x,y \in \rset^d$, we have using the definition of $\partial U(x)$ \eqref{eq:definition_partial_U} and \Cref{assum:stochastic_subgradient}-\ref{assum:stochastic_subgradient_ii}
\begin{align*}
&  \norm[2]{y-x+\stepa \gradst(x,Z)}  = \norm[2]{y-x} + 2 \stepa \ps{\gradst(x,Z)}{y-x} + \stepa^2\norm[2]{\gradst(x,Z)} \\
& \qquad \qquad \leq \norm[2]{y-x} -2\stepa\defEns{U(x)-U(y)} + 2 \stepa \ps{\gradst(x,Z)- \expe{\gradst(x,Z)}}{y-x} +\stepa^2\norm[2]{\gradst(x,Z)} \eqsp.
\end{align*}
Let $(X,Y)$ be an optimal coupling between $\mu$ and $\pi$ independent of $Z$. Then by
\Cref{assum:stochastic_subgradient}-\ref{assum:stochastic_subgradient_ii}
and rearranging the terms in the previous inequality, we obtain
\begin{equation*}
2\steps\defEns{  \Escr(\mu)-\Escr(\nu)} \leq \wasserstein_2^2(\mu,\pi) - \expe{  \norm[2]{Y-X+\stepa \gradst(X,Z)}} + \stepa^2\expe{\norm[2]{\gradst(X,Z)}} \eqsp.
\end{equation*}
The proof is concluded upon noting that $\wasserstein_2^2(\mu \bSkera, \pi) \leq \expeLigne{  \norm[2]{Y-X+\stepa \gradst(X,Z)}}$ and $\expeLigne{\norm[2]{\gradst(X,Z)}} \leq \ClyapU^2 + \vargrad_{\gradst}(\mu)$.
\end{proof}
\begin{proposition}
  \label{propo:basic-one-step_ssgld}
Assume \Cref{assum:convexity}($0$) and \Cref{assum:stochastic_subgradient}. For all $\steps,\tstep >0 $ and $\mu \in \Pens_2(\rset^d)$, 
\begin{equation*}
2\tstep \defEns{\Fscr(\mu \bRker_{\step,\tstep}) - \Fscr(\pi)} \leq  \defEns{\wassersteinTarg{\mu \bSker_{\step}}{\pi} - \wassersteinTarg{\mu \bRker_{\step,\tstep} \bSker_{\tstep}}{\pi}} 
+ \tstep^2  \defEns{\ClyapU^2 + \vargrad_{\gradst}(\mu \bRker_{\step,\tstep})} \eqsp.
\end{equation*}
where $\Fscr$ is defined in
\eqref{eq:def_potential_energy}, $\vargrad_{\gradst}(\mu)$ in \eqref{eq:definition_var_grad_sto}, $\bRker_{\step,\tstep}$ and  $\bSker_{\step}$ in \eqref{eq:definition_rker_ssgld} in  \eqref{eq:definition_sker_ssgld} respectively.
 \end{proposition}
 \begin{proof}
Note that by \Cref{lem:basic-one-step_ssgld}, we have 
\begin{equation}
\label{eq:thm:step_conv_ss_1}
2\tstep \defEns{ \Escr(\mu \bRker_{\step,\tstep}) - \Escr(\pi)} \leq \wassersteinTarg{\mu \bRker_{\step,\tstep}}{\pi} - \wassersteinTarg{\mu \bRker_{\step,\tstep} \bSker_{\tstep}}{\pi}  + \tstep^2 \defEns{\ClyapU^2 + \vargrad_{\gradst}(\mu \bRker_{\step,\tstep})} \eqsp.
\end{equation}
In addition by \Cref{lem:ent-grad-flow-step-inequality}, it holds
\begin{equation*}
2 \tstep \defEns{  \Hscr(\mu\bRker_{\step,\tstep}) - \Hscr(\pi)}  \leq  \wassersteinTarg{\mu \bSker_{\step}}{\pi} - \wassersteinTarg{\mu \bRker_{\step,\tstep} }{\pi} \eqsp.
\end{equation*}
The proof then follows from combining this inequality with \eqref{eq:thm:step_conv_ss_1}.
 \end{proof}
\subsubsection{Proof of \Cref{thm:step_conv_ss}}
\label{sec:proof-crefthm:st}
By \Cref{propo:basic-one-step_ssgld}, for all $k \in \nsets$, we have
\begin{multline*}
\Fscr(\mu \bQkera[k]) - \Fscr(\pi) \leq (2\stepa[k+1])^{-1} \defEns{\wassersteinTarg{\mu \bQkera[k-1] \bSkera[k]}{\pi} - \wassersteinTarg{\mu \bQkera[k] \bSkera[k+1]}{\pi}} \\
+ (\stepa[k+1]/2)  \defEns{\ClyapU^2 + \vargrad_{\gradst}(\mu \bQkera[k])} \eqsp.
\end{multline*}
Similarly to the proof of \Cref{thm:const_step_conv} using the convexity of Kullback-Leibler divergence and the condition
that  $\sequenceks{\weighta[k]/\stepa[k+1]}$ is
non-increasing concludes the proof.

\subsubsection{Proof of \Cref{coro:fixed_step_conv_ss}}
\label{sec:proof-coro_fixed_step_conv_ss}

On the one hand, using \Cref{thm:step_conv_ss}, we get:
\[
\KLarg{\tnu^N_n}{\pi} \leq  \left. (2 \steps n)^{-1} \wassersteinTarg{\mu_0 \bQkers[N] \bSker_{\steps}}{ \pi}
+ \steps M^2/2 + (\steps/(2n)) \sum_{k=N+1}^{N+n}  \vargrad_{\gradst} (\mu_0 \bQkers[k]  ) \right. \eqsp.
\]
On the other hand, using \Cref{propo:bound_variance_sto_grad}, we obtain:
\begin{multline*}
2\steps(\tilde{L}^{-1}-\steps)\left(\sum_{k=N+1}^{N+n} \vargrad_{\gradst}(\mu_0 \bQkers[k] ) \right)
\leq  
\int_{\rset^d} \norm{x-\xstar}^2 \rmd \mu_0 \bQkers[N+1](x)\\  - \int_{\rset^d} \norm[2]{x-\xstar} \rmd \mu_0\bQkers[N+n+1] + 2n \steps^2 \vargrad_{\gradst}(\updelta_{\xstar})+  2n \steps d \eqsp.
\end{multline*}
Combining the two inequalities above finishes the proof of the first part of \Cref{coro:fixed_step_conv_ss}.  For the second part, first observe that since $\steps_{\varepsilon} \leq (2\tilde{L})^{-1}$ we have $(2(\tilde{L}^{-1} - \steps))^{-1} \leq \tilde{L}$. Furthermore, from the definition of $\steps_{\varepsilon}$ we have $\steps_{\varepsilon} ( \frac{M^2}{2} + \tilde{L} d) \leq \varepsilon / 4$, as well as $\steps_{\varepsilon}^2 \tilde{L} \vargrad_{\gradst}(\updelta_{\xstar}) \leq \varepsilon / 4$. On the other hand, from the definition of $n_{\varepsilon}$ we have $ W_2^2(\mu_0 \bSker_{\steps_{\varepsilon}}, \pi) / (2 \steps{\varepsilon} n_{\varepsilon}) \leq \varepsilon / 4$ as well as $\tilde{L}(2n_{\varepsilon})^{-1}  \int_{\rset^d} \norm[2]{x-\xstar} \rmd \mu_0 \bRker_{\steps_{\varepsilon},\steps_{\varepsilon}} (x)          \leq \varepsilon / 4$. Combining those four bounds together finishes the proof.

\subsection{Proof of \Cref{sec:stoch-prox-grad}}
\label{sec:proof-crefthm:spgld}
We proceed for the proof of \Cref{thm:step_conv_sp} similarly to the
one of \Cref{thm:const_step_conv}, by decomposing
$\Fscr(\mu \tRker_{\step,\tstep}) - \Fscr(\pi) = \Escr(\mu \tRker_{\step,\tstep}) -
\Escr(\pi) + \Hscr(\mu \tRker_{\step,\tstep})- \Hscr(\pi)$, for
$\mu \in \Pens_2(\rset^d)$ and $\step,\tstep >0$. The main difference is that
we now need to handle carefully the proximal step in the first term of
the decomposition. To this end, we decompose the potential energy
functional according to the decomposition of $U$,
$\Escr = \Escr_1 + \Escr_2$ where for all $\mu \in \Pens_2(\rset^d)$,
\begin{equation}
  \label{eq:def_escr_1_2}
  \Escr_1(\mu) = \int_{\rset^d} U_1 \rmd \mu(x) \eqsp,  \, \qquad \qquad  \Escr_2(\mu) = \int_{\rset^d} U_2 \rmd \mu(x) \eqsp,
\end{equation}
and consider 
\begin{multline}
  \label{eq:decomp_sgpld}
  \Fscr(\mu \tRker_{\step,\tstep}) - \Fscr(\pi) = \Escr_1(\mu \tRker_{\step,\tstep}) - \Escr_1(\mu \tS_{\step}^2\tS_{\tstep}^1)\\ +\Escr_1(\mu \tS_{\step}^2\tS_{\tstep}^1) -\Escr_1(\pi)  + \Escr_2(\mu \tRker_{\step,\tstep})- \Escr_2(\pi) + \Hscr(\mu \tRker_{\step,\tstep})- \Hscr(\pi) \eqsp.
\end{multline}
The first and last terms in the right hand side will be controlled using \Cref{lem:conv-potential-change} and \Cref{lem:ent-grad-flow-step-inequality}. In the next lemmas, we bound the other terms separately. 
\begin{lemma}
  \label{lem:basic-potential-gradient-step_spgld}
Assume \Cref{assum:stochastic_gradient_proximal}$(m)$, for $m \geq 0$. 
For all $\mu,\nu \in \Pens_2(\rset^d)$ and $\steps \in \ocintLigne{0,L^{-1}}$, 
\begin{multline*}
2 \step\{  \Escr_1(\mu \tS_{\step}^1) - \Escr_1(\nu) \}\leq (1-m \step) W_2^2(\mu,\nu) - W_2^2(\mu \tS_{\step}^1,\nu)  \\ -\steps^2(1-\step L)\int_{\rset^d} \norm[2]{\nabla U_1(x)} \rmd \mu(x)+\steps^2(1+\steps L) \vargrad_1(\mu) \eqsp,
\end{multline*}
where $  \Escr_1,  \tS_{\step}^1 $ is defined by \eqref{eq:def_escr_1_2}-\eqref{eq:definition_sker_spgld} and $\vargrad_1(\mu)$ by \eqref{eq:definition_var_grad_sto_spgl}. 
\end{lemma}

\begin{proof}
Let $\mu,\nu\in \Pens_2(\rset^d)$ and $\step >0$.  Since $U_1$ satisfies \Cref{assum:grad_lip} by \cite[Lemma 1.2.3]{nesterov:2004}, for all $x,\tilde{x}\in \rset^d$, we have
$  | U_1(\tilde{x}) - U_1(x)-  \ps{\nabla U_1(x)}{\tilde{x}-x}| \leq (L/2) \norm[2]{\tilde{x}-x}$.
  Using that $U_1$ is $m$-strongly convex by \Cref{assum:stochastic_gradient_proximal}$(m)$, for all $x,y \in \rset^d$, $z \in \msz$, we get 
\begin{align*}
&  U_1(x-\steps\tgradst_1(x,z)) - U_1(y) =   U_1(x-\steps\tgradst_1(x,z))-U_1(x)+U_1(x) - U_1(y)  \\
&\qquad \leq  -\steps \ps{\nabla U_1(x)}{\tgradst_1(x,z)}+(L\steps^2/2) \norm[2]{\tgradst_1(x,z)} +\ps{\nabla U_1(x)}{x-y} - (m/2) \norm[2]{y-x} \eqsp.
\end{align*}
Then multiplying both sides by $\steps$, we obtain  
\begin{multline}
  \label{eq:2:lem:basic-potential-gradient-step_spgld}
2 \steps\defEns{   U_1(x-\steps\tgradst_1(x,z)) - U_1(y)} \leq (1-m\steps)\norm[2]{x-y} - \norm[2]{x-\steps \tgradst_1(x,z) -y}  \\
-2 \steps^2 \ps{\nabla U_1(x)}{\tgradst_1(x,z)}+\steps^2(1+\steps L)\norm[2]{\tgradst_1(x,z)} + 2 \steps\ps{\nabla U_1(x)-\tgradst_1(x,z)}{x-y} \eqsp.
\end{multline}
Let now $(X,Y)$ be an optimal coupling between $\mu$ and $\nu$ and $Z$ with distribution $\eta$ independent of $(X,Y)$. Note that \Cref{assum:stochastic_gradient_proximal} implies that $\mathbb{E}[\tgradst_1(X,Z) | (X,Y)] = \nabla U_1(X)$.
Then by definition and \eqref{eq:2:lem:basic-potential-gradient-step_spgld}, we get 
\begin{align*}
  2\steps\defEns{\Escr(\mu \tS^1_\step) - \Escr(\nu)} &\leq (1-m\steps)\wasserstein_2^2(\mu, \nu) - \expe{\norm[2]{X-\steps \tgradst_1(X) -Y}}\\
                                                      & \qquad  -2\steps^2\expe{\norm[2]{\nabla U_1(X)}}+\steps^2(1+\steps L)\expe{\norm[2]{\tgradst_1(X)}}\\
  &\leq (1-m\steps)\wasserstein_2^2(\mu, \nu) - \expe{\norm[2]{X-\steps \tgradst_1(X) -Y}}\\
                                                      & \qquad  -\steps^2(1-\step L)\expe{\norm[2]{\nabla U_1(X)}}+\steps^2(1+\steps L) \vargrad_1(\mu) \eqsp.
\end{align*}
Using that $\wasserstein^2_2(\mu \tS^1_\step,\nu) \leq \expeLigne{\normLigne[2]{X-\steps \tgradst_1(X) -Y}}$ concludes the proof.

\end{proof}

\begin{lemma} \label{lem:proximal-step}
  Assume \Cref{assum:stochastic_gradient_proximal}$(m)$ for $m \geq 0$.
  For all $\mu,\nu\in \Pens_2(\rset^d)$ and $\step >0$, we have
  \begin{equation*}
    2 \step\defEns{\Escr_2(\mu) - \Escr_2(\nu)} \leq W_2^2(\mu, \nu) - W_2^2(\mu \tS^2_\step , \nu) + 2\steps^2 M^2_2 \eqsp,
  \end{equation*}
  where $\Escr_2,\tS^2_\step$ are defined by \eqref{eq:def_escr_1_2} and \eqref{eq:definition_sker_spgld} respectively.
\end{lemma}
\begin{proof}
Let $\mu,\nu\in \Pens_2(\rset^d)$ and $\step >0$.  First we bound for any $x,y \in \rset^d$, $U_2(x) - U_2(y)$ using the decomposition $U_2(x)-U_2(\prox_{U_2}^{\step}(x))+U_2(\prox_{U_2}^{\step}(x)) - U_2(y)$.
 For any $x,y \in \rset^d$, we have using that
  $\step^{-1}(x-\prox_{U_2}^{\step}(x)) \in \partial U_2(\prox_{U_2}^{\step}(x))$ (see \cite[Chapter 1 Section G]{rockafellar:wets:1998}), where
  $\partial U_2$ is the sub differential of $U_2$ defined by
  \eqref{eq:definition_partial_U},
  \begin{equation*}
   U_2(\prox_{U_2}^{\step}(x))- U_2(y) \leq \step^{-1} \ps{x-\prox_{U_2}^{\step}(x)}{\prox_{U_2}^{\step}(x)-y}\eqsp.
  \end{equation*}
  Since $\normLigne[2]{x-y}  = \normLigne[2]{x-\prox_{U_2}^{\step}(x)} + \normLigne[2]{\prox_{U_2}^{\step}(x)-y} + 2 \psLigne{x-\prox_{U_2}^{\step}(x)}{\prox_{U_2}^{\step}(x)-y}$, we get for all $x,y \in \rset^d$,
  \begin{equation}
    \label{eq:proximal-step_eq_1}
    U_2(\prox_{U_2}^{\step}(x))-U_2(y) \leq (2\step)^{-1}(\normLigne[2]{x-y} - \normLigne[2]{\prox_{U_2}^{\step}(x)-y}) \eqsp.
  \end{equation}

  Second, since $U_2$ is $M_2$-Lipschitz, we get for any
  $x \in \rset^d$,
  $\absLigne{U_2(x)-U_2(\prox_{U_2}^{\step}(x)) }\leq M_2
  \normLigne{x-\prox_{U_2}^{\step}(x)}$. Then using that
  $\step^{-1}(x-\prox_{U_2}^{\step}(x)) \in \partial
  U_2(\prox_{U_2}^{\step}(x))$, and for any
  $v \in \partial U_2(\prox_{U_2}^{\step}(x))$, since $U_2$ is $M_2$-Lipschiz,
  $\norm{v} \leq M_2$, we obtain $\absLigne{U_2(x)-U_2(\prox_{U_2}^{\step}(x))} \leq \step M_2^2$. Combining this result and \eqref{eq:proximal-step_eq_1} yields for any $x,y \in \rset^d$
  \begin{equation*}
   2 \step\defEns{ U_2(x)-U_2(y)} \leq \normLigne[2]{x-y} - \normLigne[2]{\prox_{U_2}^{\step}(x)-y} + 2\step^2M_2^2 \eqsp.
  \end{equation*}
Let $(X,Y)$ be an optimal coupling for $\mu$ and $\nu$. The proof then follows from using the inequality above for $(X,Y)$, taking the expectation and because $W_2^2(\mu \tS_\step^1,\nu) \leq \normLigne[2]{\prox_{U_2}^{\step}(X)-Y}$. 
\end{proof}

\begin{lemma}
  \label{lem:bound_decompo_spgld}
  Assume \Cref{assum:stochastic_gradient_proximal}$(m)$, for $m \geq 0$. 
  For all $\mu_0\in \Pens_2(\rset^d)$ and $\steps, \tstep  \in \ocint{0,L^{-1}}$,
  \begin{multline*}
    2 \tstep\defEnsLigne{  \Fscr(\mu_0 \tRker_{\step,\tstep}) - \Fscr(\pi)} \leq (1-m \tstep)W_2^2(\mu_0\tS^2_\step,\pi) - W_2^2(\mu_0 \tRker_{\step,\tstep}\tS^2_{\tstep},\pi) \\+ \tstep^2\{ 2Ld + (1+\tstep L) \vargrad_1(\mu_0 \tS^2_{\step}) + 2 M_2^2 \}
    \eqsp,
   \end{multline*}
   where  $\Fscr$, $\tRker_{\step,\tstep}$ and $\tS^2_\step$ are  defined by   \eqref{eq:def_free_energy}-\eqref{eq:definition_rker_spgld}-\eqref{eq:definition_sker_spgld} respectively. 
\end{lemma}
\begin{proof}
Let $\mu_0 \in \Pens_2(\rset^d)$ and $\step,\tstep \in \ocint{0,L^{-1}}$. By \Cref{lem:conv-potential-change} and since $\tRker_{\step,\tstep} = \tS^2_{\step} \tS^1_{\tstep} T_{\tstep}$, we have
\begin{equation}
  \label{eq:lem:bound_decompo_spgld_1}
  \Escr_1(\mu_0 \tRker_{\step,\tstep}) - \Escr_1(\mu_0 \tS_\step^2\tS_{\tstep}^1) \leq 2 Ld\tstep \eqsp.
\end{equation}
By \Cref{lem:basic-potential-gradient-step_spgld} since $\tstep \leq 1/L$,
\begin{multline}
  \label{eq:lem:bound_decompo_spgld_2}
 2\tstep\defEnsLigne{ \Escr_1(\mu_0 \tS_\step^2\tS_{\tstep}^1) -\Escr_1(\pi)} \leq (1-\tstep m) W_2^2(\mu_0 \tS_\step^2, \pi)-  W_2^2(\mu_0 \tS_\step^2\tS_{\tstep}^1, \pi) \\+ \tstep^2(1+\tstep L) \vargrad_1(\mu_0 \tS_\step^2) \eqsp.
\end{multline}
By \Cref{lem:proximal-step}, we have
\begin{equation}
  \label{eq:lem:bound_decompo_spgld_3}
2  \tstep\defEnsLigne{  \Escr_2(\mu_0 \tRker_{\step,\tstep})- \Escr_2(\pi)} \leq  W_2^2(\mu_0 \tRker_\step, \pi) - W_2^2(\mu_0 \tRker_{\step,\tstep} \tS_{\tstep}^2, \pi) + 2\tstep^2 M_2^2 \eqsp.
\end{equation}
Finally by \Cref{lem:ent-grad-flow-step-inequality}, we have
\begin{equation}
  \label{eq:lem:bound_decompo_spgld_4}
  2 \tstep\{\Hscr(\mu_0 \tRker_{\step,\tstep})- \Hscr(\pi)\} \leq  W_2^2(\mu_0 \tS^2_\step \tS^1_{\tstep}, \pi) - W_2^2(\mu_0 \tRker_{\step,\tstep}, \pi) \eqsp.
\end{equation}
Combining \eqref{eq:lem:bound_decompo_spgld_1}-\eqref{eq:lem:bound_decompo_spgld_2}-\eqref{eq:lem:bound_decompo_spgld_3}-\eqref{eq:lem:bound_decompo_spgld_4} in \eqref{eq:decomp_sgpld} concludes the proof. 
\end{proof}
\subsubsection{Proof of \Cref{thm:step_conv_sp}}
\label{sec:proof-crefthm:st-1}
  Using the convexity of Kullback-Leibler divergence and \Cref{lem:bound_decompo_spgld}, we obtain
\begin{align*}
&\KLarg{\tnu^N_n}{\pi} \leq \Weight_{N,N+n} ^{-1}\sum_{k=N+1}^{N+n} \weight_{k} \KLarg{\mu_0 \tQkers[k]}{\pi} \\
& \leq (2 \Weight_{N,N+n}) ^{-1} \left[ \frac{(1-m\stepa[N+2])\weighta[N+1]}{\stepa[N+2]} \wassersteinTarg{\mu_0 \tQkers[N] \tS^2_{\step_{N+1}}}{ \pi} - 
\frac{\weighta[N+n]}{\stepa[N+n+1]} \wassersteinTarg{\mu_0 \tQkers[N+n]\tS^2_{\step_{N+n+1}}}{ \pi} \right. \\
& \qquad \qquad \left. + \sum_{k=N+1}^{N+n-1}  \defEns{\frac{(1-m \stepa[k+2])\weighta[k+1]}{\stepa[k+2]}-\frac{\weighta[k]}{\stepa[k+1]}}  \wassersteinTarg{\mu_0 \Qkers[k-1] \tS^2_{\step_{k+1}}}{ \pi}  \right.\\
 & \qquad \qquad \left. + \sum_{k=N+1}^{N+n}   \weighta[k] \stepa[k+1]\{2Ld + (1+\stepa[k+1] L) \vargrad_1(\mu_0 \Qkers[k] \tS^2_{\step_{k}}) + 2 M_2^2\} \right] \eqsp.
\end{align*}
We get the thesis using that $\weighta[k+1](1-m\stepa[k+2])/\stepa[k+2] \leq\weighta[k]/\stepa[k+1]$ for all $k\in \nset$.

\subsubsection{Proof of \Cref{coro:fixed_step_conv_sp}}
\label{sec:proof-cor:fixed_step_conv_sp}
Using \Cref{thm:step_conv_sp} we get:
\[
\KLarg{\tnu^N_n}{\pi} \leq  \left. \wassersteinTarg{\mu_0 \tQkers[N] \tSker_{\steps}^2}{ \pi} \middle/(2 \steps n)
+ \steps ( Ld + M_2^2) + \frac{\steps}{2n} \sum_{k=N+1}^{N+n}  (1+\steps L) \vargrad_1(\mu_0 \Qkers[k] \tS^2_{\steps}) \right.
\]
and using   \Cref{propo:bound_variance_sto_grad_spgld} we obtain:
\begin{multline*}
2\gamma (\tilde{L}^{-1}-\gamma) \left(\sum_{k=N+1}^{N+n} \vargrad_1(\mu_0 \Qkers[k] \tS^2_{\steps}) \right) \leq  \int_{\rset^d} \norm[2]{y-\xstar} \rmd \mu_0 \Qkers[N+1] \tS^2_{\steps} (y) \\
-\int_{\rset^d} \norm[2]{y-\xstar}
\rmd \mu_0 \Qkers[N+n+1] \tS^2_{\steps} (y)
 +2n \gamma^2 \vargrad_1(\updelta_{\xstar}) + 2n \gamma d \eqsp,
\end{multline*}

Combining the two inequalities above finishes the proof of the first part of \Cref{coro:fixed_step_conv_sp}. For the second part, observe that since $\steps_{\varepsilon} \leq L^{-1}$ and $\steps_{\varepsilon} \leq (2\tilde{L})^{-1}$ we have $(1 + \step L)(2(\tilde{L}^{-1}-\gamma))^{-1} \leq 2 \tilde{L}$. Therefore from definition of $\steps_{\varepsilon}$ we have $\steps_{\varepsilon}(Ld + M_2^2 + 2\tilde{L} d) \leq \varepsilon / 4$, as well as $\steps_{\varepsilon}^2 2\tilde{L}\vargrad_1(\updelta_{\xstar}) \leq \varepsilon/4$. On the other hand, from definition of $n_{\varepsilon}$ we have  $W_2^2(\mu_0\tSker_{\step_{\varepsilon}}^{2}, \pi)/(2n_{\varepsilon} \steps_{\varepsilon}) \leq \varepsilon /4$ as well as $2\tilde{L}(2n_{\varepsilon})^{-1} \int_{\rset^d}  \norm[2]{x-\xstar} \rmd \mu_0\tS^2_{\steps}(y) \leq \varepsilon / 4$. Combining this four bounds we get the thesis.


\subsubsection{Proof of \Cref{thm:step_conv_sp_wasser}}
\label{sec:proof-crefthm:st-2}
  Using \Cref{lem:bound_decompo_spgld} and since the
  Kullback-Leibler divergence is non-negative, we get for all $k \in \{1,\ldots,n\}$,
\begin{multline*}
  \wassersteinTarg{\mu_0 \tQkers[k] \tS^2_{\step_{k+1}}}{ \pi} \leq (1-m\stepa[k+1])\wassersteinTarg{\mu_0 \Qkers[k-1]\tS^2_{\step_{k}}}{ \pi} \\ + \stepa[k+1]^2\{2Ld + (1+\stepa[k+1] L) \vargrad_1(\mu_0 \Qkers[k-1] \tS^2_{\step_{k}}) + 2 M_2^2\} \eqsp.
\end{multline*}
The proof then follows from a direct induction.

\subsubsection{Proof of \Cref{propo:coco_st_convex}}
\label{sec:proof-crefpr}
Let $\gamma >0$, $x \in \rset^d$ and consider
  $ \tX_1 = \prox_{U_2}^{\step}\defEns{x-\step \tilde{\Theta}_1(x,\gradrv_1) +
    \sqrt{2\step} G_1}$, where $\gradrv_1$ and $G_1$ are two
  independent random variables, $\gradrv_1$ has distribution $\eta_1$
  and $G_1$ is a standard Gaussian random variable, so that $\tX_1$ has distribution $\tS^1_\step  T_{\step}\tS^2_\step(x,\cdot)$.  First by
  \cite[Theorem 26.2(vii)]{bauschke:combettes:2011}, we have that
  $\xstar = \prox_{U_2}^{\step}(\xstar - \step \nabla U_1(\xstar))$
  and by \cite[Proposition 12.27]{bauschke:combettes:2011}, the
  proximal is non-expansive, for all $x,y \in \rset^d$,
  $\normLigne{\prox^{\step}_{U_2}(x) -\prox^{\step}_{U_2}(y)} \leq
  \norm{x-y}$. Using these two results and the fact that $\tilde{\Theta}_1$ satisfies
  \Cref{assum:cocoercitivity_sto_grad}, we have
\begin{align*}
 & \expe{\norm[2]{\tX_1 - \xstar}} = \expe{\norm{ \prox_{U_2}^{\step}\defEns{x-\step \tilde{\Theta}_1(x,\gradrv_1) +
    \sqrt{2\step} G_1} -\prox_{U_2}^{\step}\{\xstar - \step \nabla U_1(\xstar)\} }^2} \\
& \qquad \leq \expe{\norm{ \left(x-\step \tilde{\Theta}_1(x,\gradrv_1) +
    \sqrt{2\step} G_1\right) - \left( \xstar - \step \nabla U_1(\xstar)\right) }^2} \\
&  \qquad \leq  \norm[2]{x-\xstar}\\
& \qquad \qquad +  \expe{2\gamma \ps{x-\xstar }{\nabla U_1(\xstar) - \tilde{\Theta}_1(x,\gradrv_1)} + \gamma^2 \norm[2]{\nabla U_1(\xstar) - \tilde{\Theta}_1(x,\gradrv_1)}} + 2 \gamma d \\
& \qquad \leq (1-\tm_1\gamma)\norm[2]{x-\xstar} - 2 \gamma(\tilde{L}^{-1}_1-\gamma) \expe{\norm[2]{ \tilde{\Theta}_1(x,\gradrv_1) -  \tilde{\Theta}_1(\xstar,\gradrv_1)}}\\
&\phantom{aaaaaaaaaaaaaaaaaaaaaaaaaaaa} +2 \gamma^2\expe{\norm[2]{ \tgradst_1(\xstar,\gradrv_1) - \nabla U_1(\xstar)}}+  2 \gamma d  \eqsp.
\end{align*}
The proof is completed upon noting that $\vargrad_1(\updelta_x) \leq \expeLigne{\norm[2]{\gradst(x,\gradrv_1) - \gradst(\xstar,\gradrv_1)}}$.

\subsubsection{Proof of \Cref{cor:spgula_str_conv_wass}}
\label{sec:proof-cor:spgld}
Using  \Cref{thm:step_conv_sp_wasser} we get:
\begin{align}
  \nonumber
  W_2^2(\mu_0\tRker_{\gamma,\gamma}^n \tSker_{\step}^2, \pi) &\leq (1-m\step)^n W_2^2(\mu_0  \tSker_{\step}^2, \pi)\\
    \nonumber
                                                             &+ \step^2\sum_{k=1}^{n}  (1-m\step)^{n-k} \left(2Ld + (1+\step L) \vargrad_1(\mu_0 \tRker_{\gamma,\gamma}^k \tS^2_{\step}) + 2 M_2^2\right) \\
    \nonumber
                                                           &  \leq (1-m\step)^n W_2^2(\mu_0  \tSker_{\step}^2, \pi) + 2(Ld+M_2)\gamma/m\\
  \label{cor:spgula_str_conv_wass_eq_2}
&+ \step^2\sum_{k=1}^{n}  (1-\tm\step)^{n-k}  (1+\step L) \vargrad_1(\mu_0 \tRker_{\gamma,\gamma}^k \tS^2_{\step})  \eqsp.
\end{align}
In addition, using \Cref{propo:coco_st_convex} and $\steps \leq (2\tilde{L}_1)^{-1}$, we have
\begin{multline*}
  \gamma \tL_1^{-1} \sum_{k=1}^{n}  (1-\tm\step)^{n-k} \vargrad_1(\mu_0 \tRker_{\gamma,\gamma}^k \tS^2_{\step})   \leq  2\gamma \sum_{k=1}^{n}(1-\tm\step)^{n-k} (\gamma \vargrad_1(\updelta_{\xstar}) + d) \\
+  \sum_{k=1}^{n}  (1-\tm\step)^{n-k+1} \int_{\rset^d} \norm[2]{x-\xstar} \rmd \mu_0 \Rker_{\gamma,\gamma}^k \tS^2_{\step}( x) \\-  \sum_{k=1}^{n}  (1-\tm\step)^{n-k} \int_{\rset^d} \norm[2]{x-\xstar} \rmd \mu_0 \Rker_{\gamma,\gamma}^{k+1}\tS^2_{\step}( x) \eqsp.
\end{multline*}
Combining this result and \eqref{cor:spgula_str_conv_wass_eq_2} concludes the proof of \eqref{cor:spgula_str_conv_wass_eq_1}.

Now, for $\steps_{\varepsilon}, n_{\varepsilon}$ as defined in the thesis of the corollary we have $\steps_{\varepsilon} \Delta_1 \leq \varepsilon/4$ and $\Delta_2 \steps_{\varepsilon}^2 \leq \varepsilon/4$. Furthermore, $(1-m\step_{\varepsilon})^{n_{\varepsilon}} W_2^2(\mu_0  \tSker_{\step_{\varepsilon}}^2, \pi) \leq \exp(- n_{\varepsilon} m \steps_{\varepsilon})  W_2^2(\mu_0  \tSker_{\step_{\varepsilon}}^2, \pi) \leq \varepsilon/4$, 
and $(1 - \steps_{\varepsilon} \tm) \Delta_3  \leq \varepsilon/4$ similarly. 
Together, the above inequalities conclude the proof.







\bibliographystyle{plain}
\bibliography{biblio}

\appendix
\section{Definitions and useful results from theory of gradient flows}
\label{sec:defin-usef-results}
Let $I \subset \rset$ be an open interval of $\rset$ and $(\mu_t)_{t \in
  I}$ be a curve on $\Pens_2(\rset^d)$, \ie~a family of probability
measures belonging to $\Pens_2(\rset^d)$. $(\mu_t)_{t \in I}$ is said
to be absolutely continuous if there exists $\ell \in \Lone(I)$ such
that for all $s,t \in I$, $s \leq t$, $\wasserstein_2(\mu_s,\mu_t)
\leq \int_{s}^t \abs{\ell}(u) \rmd u$. Denote by $\AC{I}$ the set of
absolutely continuous curves on $I$ 
 and 
 \begin{equation*}
 \ACloc{\rset^*_+} = \defEns{(\mu_t)_{t \geq 0} \ : \,  (\mu_t)_{t \in I} \in \AC{I} \text{ for any open interval } I \subset \rset^*_+ }   \eqsp.
 \end{equation*}
Note that if $(\mu_t)_{t \in I} \in \AC{I}$, then for any $\nu \in
\Pens_2(\rset^d)$, $t \mapsto \wasserstein_2(\nu,\mu_t)$ is absolutely
continuous on $I$ (as a curve from $I$ to $\rset_+$). Therefore by
\cite[Theorem 20.8]{nielsen:1997} and \cite[Exercice 4,
p.45]{mitrovic:zubrinic:1997}, $t \mapsto \wasserstein_2(\nu,\mu_t)$
has derivative for almost all $t \in I$ and there exists $\delta : I
\to \rset$ satisfying
 \begin{equation}
   \label{eq:carac_derivative_wasser_nu}
   \int_{I} \abs{\delta}(u) \rmd u < \plusinfty \text{ and }
\wasserstein_2^2(\nu,\mu_t)-\wasserstein_2^2(\nu,\mu_s) = \int_{s}^t
\delta(u) \rmd u \eqsp, \text{ for all } s,t \in I
 \end{equation}


Let $\mu,\nu \in \Pens_2(\rset^d)$. A constant speed geodesic
$(\lambda_t)_{t \in \ccint{0,1}}$ between $\mu$ and $\nu$ is a curve
in $\Pens_2(\rset^d)$ such that $\lambda_0=\mu$, $\lambda_1=\nu$ and
for all for all $s,t \in \ccint{0,1}$,
$\wasserstein_2(\lambda_s,\lambda_t) = \abs{t-s}
\wasserstein_2(\mu,\nu)$. Note that by the triangle inequality, this
definition is equivalent to for all $s,t \in \ccint{0,1}$,
$\wasserstein_2(\lambda_s,\lambda_t) \leq \abs{t-s}
\wasserstein_2(\mu,\nu)$.  Indeed by the triangle inequality and the assumption $\wasserstein_2(\lambda_s,\lambda_t) \leq \abs{t-s}
\wasserstein_2(\mu,\nu)$, we have for all $s,t \in \ccint{0,1}$, $s <t$,
\begin{equation*}
\wasserstein_2(\mu,\nu) \leq \wasserstein_2(\mu,\lambda_t) +  \wasserstein_2(\lambda_t,\lambda_s) + \wasserstein_2(\lambda_s,\nu) 
 \leq \wasserstein_2(\mu,\nu) \eqsp.
\end{equation*}
Therefore the first inequality is in fact an equality, and therefore using again the assumption for $\wasserstein_2(\mu,\lambda_t)$ and $\wasserstein_2(\lambda_s,\nu) $ concludes the proof.
By definition of the Wasserstein distance of order
$2$, a constant speed geodesic $(\lambda_t)_{t \in \ccint{0,1}}$
between $\mu$ and $\nu$ is given for all $t \in \ccint{0,1}$ by
$\lambda_t = (t\proj_1+(1-t)\proj_2)_\sharp \zeta$ where $\zeta$ is an
optimal transport plan between $\mu$ and $\nu$ and $\proj_1,\proj_2 :
\rset^{2d} \to \rset^d$ are the projections on the first and last $d$
components respectively.

Let $\Sscr : \Pens_2(\rset^d) \to \ocint{-\infty,\plusinfty}$. The functional $\Sscr$
 is said to be  lower semi-continuous  if for all $M \in \rset$, $\{
\Sscr \leq M \}$ is a closed set of $\Pens_2(\rset^d)$ and 
$m$-geodesically convex for $m \geq 0$ if for any $\mu, \nu \in
\Pens_2(\rset^d)$ there exists a constant speed geodesic
$(\lambda_t)_{t \in \ccint{0,1}}$ between $\mu$ and $\nu$ such that for all $t \in \ccint{0,1}$ 
\begin{equation*}
  \Sscr(\lambda_t) \leq t \Sscr(\mu) + (1-t) \Sscr(\nu) - t(1-t)(m/2) \wasserstein_2^2(\mu,\nu) \eqsp.
\end{equation*}
If $m=0$, $\Sscr$ will be simply said geodesically convex. 

A curve $(\mu_t)_{t >0} \in \ACloc{\rset^*_+}$ is said to be a
gradient flow for the lower semi-continuous and $m$-geodesically
convex function $\Sscr : \Pens_2(\rset^d) \to
\ocint{-\infty,\plusinfty}$ if for all $\nu \in \Pens_2(\rset^d)$, 
$\Sscr(\nu) < \plusinfty$, and for almost all $t \in \rset^*_+$,
\begin{equation*}
  (1/2)\delta_t +(m/2) \wasserstein_2^2(\mu_t,\nu) \leq \Sscr(\nu) - \Sscr(\mu_t) \eqsp,
\end{equation*}
where $\delta : \rset^*_+ \to \rset$ satisfies
\eqref{eq:carac_derivative_wasser_nu} for all open interval of
$\rset^*_+$. We say that $(\mu_t)_{t \in \rset^*_+}$ starts at $\mu$ if
$\lim_{t \to 0} \wasserstein_2(\mu_t,\mu) = 0$ and then set $\mu_0 =
\mu$. By \cite[Theorem 11.1.4]{ambrosio2008gradient}, there exists at
most one gradient flow associated with $\Sscr$. 

Consider the functional $\tFscr : \Pens_2(\rset^d) \to \ocint{-\infty,
  \plusinfty}$ given by $\tFscr = \Hscr + \tEscr$ where $\Hscr$ is
defined by \eqref{eq:def_Boltz_entropy} and $\tEscr$ for all $\mu \in
\Pens_2(\rset^d)$ by
\begin{equation*}
  \tEscr(\mu) = \int_{\rset^d} V(x) \rmd \mu(x) \eqsp,
\end{equation*}
where $V : \rset^d \to \ocint{-\infty, \plusinfty}$ is a convex
lower-semicontinuous function (for all $M \geq 0$, $\{ V \leq M\}$ is
closed subset of $\rset^d$) with $\{ V < \plusinfty \} \not =
\emptyset$ and the interior of this set non empty as well. By
\cite[Proposition 9.3.2, Theorem 9.4.12]{ambrosio2008gradient},
$\tFscr$ is geodesically convex and \cite[Theorem 11.2.8,Theorem
11.1.4]{ambrosio2008gradient} shows that there exists a unique
gradient flow $(\mu_t)_{t \geq 0}$ starting at $\mu \in \Pens_2(\rset^d)$ and
this curve is the unique solution of the Fokker-Plank equation (in the sense of distributions) : 
\begin{equation*}
    \frac{\partial \mu_t}{\partial t} = \divergence (\nabla \mu_t^{x}+ \mu_t^{x} \nabla V(x)) \eqsp,
\end{equation*}
\ie~for all $\phi \in C_c^{\infty}( \rset^d)$ and $t >0$,
\begin{equation*}
  \frac{\partial }{\partial t} \int_{\rset^d} \phi(y) \mu_t(\rmd y) = \int_{\rset^d} \generator \phi(y) \, \mu_t(\rmd y) \eqsp.
\end{equation*}
In addition for all $t >0$, $\mu_t$ is absolutely continuous with respect to the Lebesgue measure. 
In particular for $V =0$, we get the following result.

\begin{theorem}
  \label{theo:heat_flow_prop}
For all $\mu \in \Pens_2(\rset^d)$, there exists a unique solution of the Fokker-Plank equation (in the sense of distributions) : 
\begin{equation*}
    \frac{\partial \mu_t}{\partial t} = \Delta \mu_t \eqsp.
\end{equation*}
In addition $(\mu_t)_{t \geq 0} \in \AC{\rset^*_+}$ and satisfies for almost all $t \in \rset^*_+$,
\begin{equation*}
\delta_t/2 \leq \Hscr(\nu)  -  \Hscr(\mu_t) \eqsp,
\end{equation*}
where $\delta_t$ is given in \eqref{eq:carac_derivative_wasser_nu}.
\end{theorem}


\section{On the second order moment of logconcave measures }
\label{sec:second-order-moment}
\begin{assumption}
  \label{assum:super_log_concave}
  There exist $\eta >0$, $M_{\eta} \geq 0$ such that for all $x \in \rset^d$, $x \not \in \boule{0}{M_{\eta}}$,
  \begin{equation*}
    U(x) - U(\xstar) \geq \eta \norm{x-\xstar} \eqsp. 
  \end{equation*}
\end{assumption}

In this section, we give some bounds on to deal with the distance of the  initial condition of the algorithms from $\pi$ in $W_2$.
\begin{proposition}
  \label{propo:bound_wasser_init_condition}
  Assume \Cref{assum:convexity}($0$)  and \Cref{assum:super_log_concave}.  Then, we have
  \begin{align*}
    \int_{\rset^d} \norm[2]{x - \xstar} \rmd \pi(x) &\leq 2 \eta^{-2}d (1+d) + M_{\eta}^2  \eqsp. 
  \end{align*}
\end{proposition}

\begin{proof}
  Note that under \Cref{assum:super_log_concave}, we have
  \begin{align}
    \nonumber
    &\int_{\rset^d} \norm[2]{x - \xstar} \rmd \pi(x) \leq     \eta^{-2} \int_{\rset^d} \abs{U(x) - U(\xstar)}^2   \rmd \pi(x) + M_{\eta}^2 \\
    \label{eq:proofbound_wasser_init_condition_1}
    &  \leq  2\eta^{-2} \int_{\rset^d} \abs{U(x)+\log(\rmZ) + \Hscr(\pi)}^2   \rmd \pi(x) + 2\eta^{-2} \abs{-\Hscr(\pi)-\log(\rmZ) - U(\xstar)}^2+ M_{\eta}^2 \eqsp. 
  \end{align}
  where $\Hscr$ is defined by \eqref{eq:def_Boltz_entropy} and $\rmZ = \int_{\rset^d} \rme^{-U(y)} \rmd y$. Then, by
  \cite[Proposition I.2]{bobkov:madiman:2011},
  $\abs{-\Hscr(\pi)-\log(\rmZ)- U(\xstar)} \leq d$ and by \cite[Theorem 2.3]{fradelizi:madiman:wang:2016}, (see also \cite{nguyen:2013phd} and \cite{wang:2014phd}), $ \int_{\rset^d} \abs{U(x)+\log(\rmZ) + \Hscr(\pi)}^2   \rmd \pi(x) \leq d$. Combining these two results in \eqref{eq:proofbound_wasser_init_condition_1} concludes the proof. 
  
\end{proof}



\end{document}